\PassOptionsToPackage{nameinlink}{cleveref}
\documentclass[a4paper, USenglish, cleveref, autoref, thm-restate]{lipics-v2021}

\usepackage{bbm}
\usepackage{stmaryrd}
\SetSymbolFont{stmry}{bold}{U}{stmry}{m}{n}
\usepackage{xspace}
\usepackage{relsize}
\usepackage{mathtools}
\usepackage{diagbox}
\usepackage{mathdots}
\usepackage{anyfontsize}

\theoremstyle{definition}
\newtheorem{assumption}[theorem]{Assumption}

\usepackage[usenames,dvipsnames,svgnames,table]{xcolor}
\definecolor{dark-gray}{gray}{0.4}
\definecolor{grayish}{gray}{0.4}
\definecolor{appledarkgray}{rgb}{0.55, 0.71, 0.0}
\definecolor{darkolive}{rgb}{0.0, 0.5, 0.0}

\usepackage{tikz}
\usetikzlibrary{arrows,automata,trees,backgrounds,decorations.pathmorphing,positioning,calc}
\tikzset{shorten >=1pt, >=stealth, auto, node distance=40, initial text=}

\title{Decision problems for origin-close top-down tree transducers (full version)} 

\titlerunning{Decision problems for origin-close top-down tree transducers} 

\author{Sarah Winter}{Université libre de Bruxelles, Belgium \and \url{http://www.sarahwinter.net}}{}{https://orcid.org/0000-0002-3499-1995}{}

\authorrunning{S. Winter} 

\Copyright{S. Winter} 



\category{} 

\funding{This research was supported by the DFG grant LO 1174/3.}

\acknowledgements{The author wants to thank Emmanuel Filiot for valuable feedback on the presentation.}

\nolinenumbers 

\hideLIPIcs  

\EventEditors{John Q. Open and Joan R. Access}
\EventNoEds{2}
\EventLongTitle{}
\EventShortTitle{}
\EventAcronym{}
\EventYear{}
\EventDate{}
\EventLocation{}
\EventLogo{}
\SeriesVolume{42}
\ArticleNo{23}

\makeatletter
\newcommand{\charfusion}[2]{%
  \def\ch@rfusion##1##2{%
    \ooalign{\hfil$##1#1$\hfil\cr\hfil$##2$\hfil\crcr}}%
      \mathop{%
      \vphantom{#1}%
      \mathpalette\ch@rfusion#2}\displaylimits}
\makeatother

\newcommand{\cupcdot}{\charfusion{\cup}{\cdot}}

\makeatletter
\newcommand{\fixed@sra}{\ensuremath{\vrule height 2\fontdimen22\textfont2 width 0pt\shortrightarrow}}
\newcommand{\shortarrow}[1]{%
  \mathrel{\text{\rotatebox[origin=c]{\numexpr#1*45}{\fixed@sra}}}
}
\makeatother

\makeatletter
\newcommand{\bigplus}{%
  \DOTSB\mathop{\mathpalette\mattos@bigplus\relax}\slimits@
}
\newcommand\mattos@bigplus[2]{%
  \vcenter{\hbox{%
    \sbox\z@{$#1\sum$}%
    \resizebox{!}{0.9\dimexpr\ht\z@+\dp\z@}{\raisebox{\depth}{$\m@th#1+$}}%
  }}%
  \vphantom{\sum}%
}
\makeatother

\newcommand{\dom}[1]{\ensuremath{\mathit{dom}_{#1}}}
\newcommand{\val}[1]{\ensuremath{#1}}

\newcommand{\SigmaI}{\ensuremath{\Sigma}}
\newcommand{\SigmaO}{\ensuremath{\Gamma}}

\newcommand{\SigmaIrk}[1]{\ensuremath{\Sigma_{#1}}}

\newcommand{\Out}{\textsf{Out}}
\newcommand{\In}{\textsf{In}}

\newcommand{\pta}{\ensuremath{\textup{\textsc{pta}}}\xspace}

\newcommand{\fst}{\ensuremath{\textup{\textsc{fst}}}\xspace}
\newcommand{\nfst}{\ensuremath{\textup{\textsc{fst}}}\xspace}

\newcommand{\tdtt}{\ensuremath{\textup{\textsc{tdtt}}}\xspace}
\newcommand{\ntdtt}{\ensuremath{\textup{\textsc{tdtt}}}\xspace}
\newcommand{\dtdtt}{\ensuremath{\textup{\textsc{dtdtt}}}\xspace}
\newcommand{\dtdttr}{\ensuremath{\textup{\textsc{dtdtt}}^{\textup{R}}}\xspace}

\newcommand{\mysubparagraph}[1]{\vspace{0mm}\subparagraph{#1}}

\begin{document}

\maketitle

\begin{abstract}
Tree transductions are binary relations of finite trees.
For tree transductions defined by non-deterministic top-down tree transducers, inclusion, equivalence and synthesis problems are known to be undecidable.
Adding origin semantics to tree transductions, i.e., tagging each output node with the input node it originates from, is a known way to recover decidability for inclusion and equivalence.
The origin semantics is rather rigid, in this work, we introduce a similarity measure for transducers with origin semantics and show that we can decide inclusion, equivalence and synthesis problems for origin-close non-deterministic top-down tree transducers.
\end{abstract}

\section{Introduction}
\label{sec:intro}

In this paper we study decision problems for top-down tree transducers over finite trees with origin semantics.
Rounds \cite{rounds1970mappings} and Thatcher \cite{thatcher1970generalized2} independently invented tree transducers (their model is known today as top-down tree transducer) as a generalization of finite state word transducers in the context of natural language processing and compilers in the beginning of the 1970s.
Nowadays, there is a rich landscape of various tree transducer models used in many fields, for example, syntax-directed translation \cite{fulop2012syntax}, databases \cite{MiloSV03,hakuta2014xquery}, linguistics \cite{maletti2009power,braune2013shallow}, programming languages \cite{voigtlander2004composition,matsuda2012polynomial}, and security analysis \cite{kusters2007transducer}.

%
Unlike tree automata, tree transducers have undecidable inclusion and equivalence problems \cite{esik1980decidability}.
This is already the case for word transducers \cite{DBLP:journals/jacm/Griffiths68,FischerR68}.
%
%
The intractability of, e.g., the equivalence problem for transducers (whether two given transducers recognize the same transduction, that is, the same relation) mainly stems from the fact that two transducers recognizing the same transduction may produce their outputs very differently.
One transducer may produce its output fast and be ahead of the other.
In general, there is an infinite number of transducers for a single transduction.
To overcome this difficulty Bojanczyk \cite{DBLP:conf/icalp/Bojanczyk14} has introduced origin semantics, that is, additionally, there is an origin function that maps output positions to their originating input positions.
The main result of \cite{DBLP:conf/icalp/Bojanczyk14} is a machine-independent characterization of transductions defined by deterministic two-way transducers with origin semantics.
Word transducers with origin semantics where further investigated in \cite{DBLP:conf/icalp/BojanczykDGP17}, and properties of subclasses of transductions with origin semantics definable by one-way word transducers have been studied in \cite{conf/stacs/FigueiraL14,DBLP:conf/stacs/DescotteFF19}.
Under origin semantics, many interesting problems become decidable, e.g., equivalence of one-way word transducers. 
This is not surprising as a transduction now incorporates \emph{how} it translates an input word into an output word providing much more information.

In \cite{DBLP:journals/iandc/FiliotMRT18}, the authors have initiated a study of several decision problems for different tree transducer models on finite trees with origin semantics.
More concretely, they studied inclusion, equivalence, injectivity and query determinacy problems for top-down tree transducers, tree transducers definable in monadic second order logic, and top-down tree-to-word transducers.
They showed (amongst other results) that inclusion and equivalence become decidable for all models except tree-to-string transducers with origin semantics.

In general, there has been an interest to incorporate some kind of origin information (i.e., \emph{how} a transduction works) into tree transductions, in order to gain more insight on different tree transductions, see, e.g., \cite{van1993origin,engelfriet2003macro,maletti2011tree}.

%
However, the origin semantics is rather rigid. To mitigate this, in \cite{DBLP:conf/icalp/FiliotJLW16}, the authors have introduced a similarity measure between (one-way) word transducers with origin semantics which amounts to a measure that compares the difference between produced outputs on the same input prefix, in short, the measure compares their output delays.
They show that inclusion, equivalence, and sequential uniformization (see next paragraph) problems become decidable for transducers that have bounded output delay.
These problem are undecidable for word transducers in general, see \cite{DBLP:journals/jacm/Griffiths68,FischerR68,CarayolL14}.
The introduction of this similarity measure has triggered similar works on two-way word transducers, see \cite{DBLP:conf/fsttcs/BoseMPP18,DBLP:conf/mfcs/BoseKMPP19}.

In order to obtain decidability results (in a less rigid setting than origin semantics), we initiate the study of inclusion, equivalence, and uniformization problems for top-down tree transducers under similarity measures which are based on the behavior of the transducers.

%
A uniformization of a binary relation is a function that selects for each element of the domain of the relation an element in its image.
Synthesis problems are closely related to \emph{effective} uniformization problems; algorithmic synthesis of specifications (i.e., relations) asks for effective uniformization by functions that can be implemented in a specific way.
The classical setting is Church's synthesis problem \cite{church1962logic}, where logical specifications over infinite words are considered. 
Büchi and Landweber \cite{buechi} showed that for specifications in monadic second order logic, that is, specifications that can be translated into synchronous finite automata, it is decidable whether they can be realized by a synchronous sequential transducer. 
Later, decidability has been extended to asynchronous sequential transducers \cite{hosch1972finite,holtmann2010degrees}.
Detailed studies of the synthesis of sequential transducers from synchronous and asynchronous finite automata on finite words are provided in \cite{DBLP:conf/icalp/FiliotJLW16,DBLP:conf/icalp/Winter18}, for an overview see \cite{CarayolL14}.

Uniformization questions in this spirit have been first studied for relations over finite trees in \cite{DBLP:journals/iandc/LodingW17,DBLP:conf/mfcs/LodingW16}.
The authors have considered tree-automatic relations, that is, relations definable by tree automata over a product alphabet.
They have shown that for tree-automatic relations definable by deterministic top-down tree automata uniformization by deterministic top-down tree transducers (which are a natural extension of sequential transducer on words) is decidable. However, for non-deterministic top-down tree automata it becomes undecidable.

%
Our contribution is the introduction of two similarity measures for top-down tree transducers.
The first measure is an extension of the output delay measure introduced for word transducers in \cite{DBLP:conf/icalp/FiliotJLW16} to tree transducers.
Comparing top-down tree transducers based on their output delay has also been done in e.g., \cite{DBLP:journals/tcs/EngelfrietMS16}, we use the same notion of delay to define our measure.
Unfortunately, while decidability for major decision problems is regained in the setting of word transducers, we show that it is not in the setting of tree transducers.
The second similarity measure is more closely connected to the origin semantics.
We define two transducers as origin-close if there is a bound on the distance of two positions which are origins of the same output node by the two transducers.
Our main result is that inclusion, equivalence and uniformization by deterministic top-down tree transducers is decidable for origin-close top-down tree transducers.

%
The paper is structured as follows.
In \cref{sec:prelims} we provide definitions and terminology used throughout the paper.
In \cref{sec:similarity} we present two similarity measures for (top-down tree) transducers and provide a comparison of their expressiveness, and in \cref{sec:origin-close} we consider decision problems for origin-close top-down tree transducers.

\section{Preliminaries}
\label{sec:prelims}


\mysubparagraph{Words, trees, and contexts.}

An \emph{alphabet} $\Sigma$ is a finite non-empty set of \emph{letters} or \emph{symbols}.
A finite \emph{word} is a finite sequence of letters.
The set of all finite words over $\Sigma$ is denoted by $\Sigma^*$.
The length of a word $w \in \Sigma^*$ is denoted by $|w|$, the empty word is denoted by $\varepsilon$.
We write $u \sqsubseteq w$ if there is some $v$ such that $w = uv$ for $u,v \in \Sigma^*$.
A subset $L \subseteq \Sigma^*$ is called \emph{language} over $\Sigma$.
A \emph{ranked alphabet} $\Sigma$ is an alphabet where each letter $f \in \Sigma$ has a rank $rk(f) \in \mathbbm{N}$.
The set of letters of rank $i$ is denoted by $\Sigma_i$.
A tree domain $\dom{}$ is a non-empty finite subset of $(\mathbbm N\setminus\{0\})^*$ such that $\dom{}$ is prefix-closed and for each $u \in (\mathbbm N\setminus\{0\})^*$ and $i \in \mathbbm N\setminus\{0\}$ if $ui \in \dom{}$, then $uj \in \dom{}$ for all $1 \leq j < i$.
We speak of $ui$ as successor of $u$ for each $u \in \dom{}$ and $i \in \mathbbm N\setminus\{0\}$, and the $\sqsubseteq$-maximal elements of \dom{} are called \emph{leaves}.

A (finite $\Sigma$-labeled) \emph{tree} is a mapping $\val{t}: \dom{t} \rightarrow \Sigma$ such that for each node $u \in \dom{t}$ the number of successors of $u$ is a rank of $\val{t}(u)$.
The height $h$ of a tree $t$ is the length of its longest path, i.e., $h(t) = max\{ |u| \mid u \in \dom{t}\}$.
The set of all $\Sigma$-labeled trees is denoted by $T_{\Sigma}$.
A subset $T \subseteq T_{\Sigma}$ is called \emph{tree language} over $\Sigma$.

A \emph{subtree} $t|_u$ of a tree $t$ at node $u$ is defined by $\dom{t|_u} = \{ v \in \mathbbm{N}^* \mid uv \in \dom{t} \}$ and $\val{t|_u}(v) = \val{t}(uv)$ for all $v \in \dom{t|_u}$.
In order to formalize concatenation of trees, we introduce the notion of special trees.
A \emph{special tree} over $\Sigma$ is a tree over \(\Sigma \cupcdot \{\circ\}\) such that $\circ$ has rank zero and occurs exactly once at a leaf.
Given $t \in T_{\Sigma}$ and $u \in \dom{t}$, we write $t[\circ/u]$ for the special tree that is obtained by deleting the subtree at $u$ and replacing it by $\circ$.
Let $S_{\Sigma}$ be the set of special trees over $\Sigma$.
For $t \in S_{\Sigma}$ and $s \in T_{\Sigma}$ or $s \in S_{\Sigma}$ let the \emph{concatenation} $t \cdot s$ be the tree that is obtained from $t$ by replacing $\circ$ with $s$.


Let $X_n$ be a set of $n$ variables $\{x_1,\dots,x_n\}$ and $\Sigma$ be a ranked alphabet.
We denote by $T_{\Sigma}(X_n)$ the set of all trees over $\Sigma$ which additionally can have variables from $X_n$ at their leaves.
We define $X_0$ to be the empty set, the set $T_\Sigma(\emptyset)$ is equal to $T_\Sigma$.
Let $X = \bigcup_{n > 0} X_n$.
A tree from $T_{\Sigma}(X)$ is called \emph{linear} if each variable occurs at most once.
For $t \in T_{\Sigma}(X_n)$ let $t[x_1 \leftarrow t_1, \dots, x_n \leftarrow t_n]$ be the tree that is obtained by substituting each occurrence of $x_i \in X_n$ by $t_i \in T_{\Sigma}(X)$ for every $1 \leq i \leq n$.

A tree from $T_{\Sigma}(X_n)$ such that all variables from $X_n$ occur exactly once and in the order $x_1,\dots,x_n$ when reading the leaf nodes from left to right, is called $n$-\emph{context} over $\Sigma$.
Given an $n$-context, the node labeled by $x_i$ is referred to as $i$th hole for every $1 \leq i \leq n$. 
A special tree can be seen as a $1$-context, a tree without variables can be seen a $0$-context.
If $C$ is an $n$-context and $t_1, \dots, t_n \in T_{\Sigma}(X)$ we write $C[t_1,\dots,t_n]$ instead of $C[x_1 \leftarrow t_1, \dots, x_n \leftarrow t_n]$.

\mysubparagraph{Tree transductions, origin mappings, and uniformizations.}

Let $\Sigma,\Gamma$ be ranked alphabets.
A \emph{tree transduction (from $T_\Sigma$ to $T_\Gamma$)} is a relation $R \subseteq T_\Sigma \times T_\Gamma$.
Its \emph{domain}, denoted $\mathrm{dom}(R)$, is the projection of $R$ on its first component. 
Given trees $t_1,t_2$, an \emph{origin mapping} of $t_2$ in $t_1$ is a function $o\colon \dom{t_2} \to \dom{t_1}$.
Given $v \in \dom{t_2}$, $u \in \dom{t_1}$, we say $v$ has origin $u$ if $o(v) = u$.
Examples are depicted in \cref{fig:intro-mapping,fig:k-origin}.
A \emph{uniformization} of a tree transduction $R \subseteq T_\Sigma \times T_\Gamma$ is a function $f\colon \mathrm{dom}(R) \to T_\Gamma$ such that $(t,f(t)) \in R$ for all $t \in \mathrm{dom}(R)$.

\mysubparagraph{Top-down tree transducers.}

We consider top-down tree transducers, which read the tree from the root to the leaves in a parallel fashion and produce finite output trees in each step that are attached to the already produced output.

A \emph{top-down tree transducer} (a \tdtt) is of the form $\mathcal T = (Q,\Sigma,\Gamma,q_0,\Delta)$ consisting of a finite set of states $Q$, a finite input alphabet $\Sigma$, a finite output alphabet $\Gamma$, an initial state $q_0 \in Q$, and $\Delta$ is a finite set of transition rules of the form
\[
  q(f(x_1,\dots,x_i)) \rightarrow w[q_1(x_{j_1}),\dots,q_n(x_{j_n})], 
\]
where $f \in \Sigma_i$, $w$ is an $n$-context over $\Gamma$, $q,q_1,\dots,q_n \in Q$ and variables $x_{j_1},\dots,x_{j_n} \in X_i$.
A  \emph{deterministic} \tdtt (a \dtdtt) has no two rules with the same left-hand side.

We now introduce a non-standard notion of configurations which is more suitable to prove our results.
Usually, a configuration is a partially transformed input tree; the upper part is the already produced output, the lower parts are remainders of the input tree.
Here, we keep the input and output tree separate and introduce a mapping from nodes of the output tree to nodes of the input tree from where the transducer continues to read.
A visualization of several configurations is given in \cref{fig:configuration}.

A \emph{configuration} of a top-down tree transducer is a triple $c = (t,t',\varphi)$ of an input tree $t \in T_{\Sigma}$, an output tree $t' \in T_{\Gamma \cup Q}$ and a function $\varphi: D_{t'} \rightarrow \dom{t}$, where
\begin{itemize}
 \item $\val{t'}(u) \in \Gamma_i$ for each $u \in \dom{t'}$ with $i > 0$ successors, and
 \item $\val{t'}(u) \in \Gamma_0$ or $\val{t'}(u) \in Q$ for each leaf $u \in \dom{t'}$, and
 \item $D_{t'} \subseteq \dom{t'}$ with $D_{t'} = \{ u \in \dom{t'} \mid \val{t'}(u) \in Q\}$, i.e., $\varphi$ maps every node from the output tree $t'$ that has a state-label to a node of the input tree $t$.
\end{itemize}
\indent Let $c_1 = (t,t_1,\varphi_1)$ and $c_2 = (t,t_2,\varphi_2)$ be configurations of a top-down tree transducer over the same input tree.
We define a \emph{successor relation} $\rightarrow_{\mathcal T}$ on configurations as usual by applying one rule.
Figure \ref{fig:configuration} illustrates a configuration sequence explained in Example \ref{ex:transducer} below.
Formally, for the application of a rule, we define $c_1 \rightarrow_{\mathcal T}c_2 :\Leftrightarrow$
\begin{itemize}
 \item There is a state-labeled node $u \in D_{t'}$ of the output tree $t_1$ that is mapped to a node $v \in \dom{t}$ of the input tree $t$, i.e., $\varphi_1(u) = v$, and
 \item there is a rule $\val{t_1}(u)\left(\val{t}(v)(x_1,\dots,x_i)\right) \rightarrow w[q_1(x_{j_1}),\dots,q_n(x_{j_n})] \in \Delta$ such that the output tree is correctly updated, i.e., $t_2 = t_1[\circ/u] \cdot w[q_1,\dots,q_n]$, and
 \item the mapping $\varphi_2$ is correctly updated, i.e., $\varphi_2(u') = \varphi_1(u')$ if $u' \in D_{t_1}\setminus \{u\}$ and $\varphi_2(u') = v.j_i$ if $u' = u.u_i$ with $u_i$ is the $i$th hole in $w$.
\end{itemize}
Furthermore, let $\rightarrow_{\mathcal T}^*$ be the reflexive and transitive closure of $\rightarrow_{\mathcal T}$.
From here on, let $\varphi_0$ always denote the mapping $\varphi_0(\varepsilon)=\varepsilon$.
A configuration $(t,q_0,\varphi_0)$ is called \emph{initial configuration} of $\mathcal T$ on $t$.
A configuration sequence starting with an initial configuration where each configuration is a successor of the previous one is called a \emph{run}.
For a tree $t \in T_{\Sigma}$ let $\mathcal T(t) \subseteq T_{\Gamma \cup Q}$ be the set of \emph{final transformed outputs} of a computation of $\mathcal T$ on $t$, that is the set $\{t' \mid (t,q_0,\varphi_0)\rightarrow_{\mathcal T}^* (t,t',\varphi) \text{ s.t.\ there is no successor configuration of } (t,t',\varphi)\}.$
Note, we explicitly do not require that the final transformed output is a tree over $\Gamma$.
In the special case that $\mathcal T(t)$ is a singleton set $\{t'\}$, we also write $\mathcal T(t) = t'$.
The \emph{transduction} $R(\mathcal T)$ induced by a \tdtt $\mathcal T$ is $R(\mathcal T) = \{(t,t') \mid t' \in \mathcal T(t) \cap T_\Gamma\}.$
The class of relations definable by \tdtt{}s is called the class of \emph{top-down tree transductions}, conveniently denoted by \tdtt.

Let $\mathcal T$ be a \tdtt, and let $\rho = c_0\dots c_n$ be a run of $\mathcal T$ on an input tree $t \in T_{\SigmaI}$ that results in an output tree $s \in T_{\SigmaO}$.
The \emph{origin function $o$} of $\rho$ maps a node $u$ of the output tree to the node $v$ of the input tree that was read while producing $u$, formally $o:\dom{s} \rightarrow \dom{t}$ with $o(u) = v$ if there is some $i$, such that $c_i = (t,t_i,\varphi_i)$, $c_{i+1} = (t,t_{i+1},\varphi_{i+1})$ and $\varphi_i(u) = v$ and $\val{t_{i+1}}(u) = \val{s}(u)$, see \cref{fig:configuration}.
We define $R_o(\mathcal T)$ to be the set 
\[\{(t,s,o) \mid t \in T_{\SigmaI}, s\in T_{\SigmaO} \text{ and } \exists\thinspace \rho:(t,q_0,\varphi) \rightarrow_{\mathcal T}^* (t,s,\varphi') \text{ with origin } o\}.\]


\begin{example}\label{ex:transducer}
 Let $\Sigma$ be a ranked alphabet given by $\Sigma_2 = \{f\}$, $\Sigma_1 = \{g,h\}$, and $\Sigma_0 = \{a\}$.
 Consider the \tdtt\ $\mathcal T$ given by $(\{q\},\Sigma,\Sigma,\{q\},\Delta)$ with $\Delta$ $=$ $\{$ $q(a) \rightarrow a$, $q(g(x_1)) \rightarrow q(x_1)$, $q(h(x_1)) \rightarrow h(q(x_1))$, $q(f(x_1,x_2)) \rightarrow f(q(x_1),q(x_2))$  $\}$.
 For each $t \in T_{\Sigma}$ the transducer deletes all occurrences of $g$ in $t$.
 Consider $t := f(g(h(a)),a)$. A possible sequence of configurations of $\mathcal T$ on $t$ is $c_0 \rightarrow_{\mathcal T}^5 c_5$ such that $c_0 := (t,q,\varphi_0)$ with $\varphi_0(\varepsilon) = \varepsilon$, $c_1 := (t,f(q,q),\varphi_1)$ with $\varphi_1(1) = 1$, $\varphi_1(2) = 2$, $c_2 := (t,f(q,q),\varphi_2)$ with $\varphi_2(1) = 11$, $\varphi_2(2) = 2$, $c_3 := (t,f(q,a),\varphi_3)$ with $\varphi_3(1) = 11$, $c_4 := (t,f(h(q),a),\varphi_4)$ with $\varphi_4(11) = 111$, and $c_5 := (t,f(h(a),a),\varphi_5)$.
 A visualization of this sequence and resulting origin mapping is shown in \cref{fig:configuration}.
\end{example}

\begin{figure}[t!]
  \centering
  \begin{subfigure}{0.2\textwidth}
  \begin{tikzpicture}[scale=0.8,baseline=(current bounding box.base)]
    \tikzstyle{level 1}=[sibling distance=8mm]
    \tikzstyle{level 2}=[sibling distance=8mm]
    \tikzstyle{level 3}=[sibling distance=8mm]
    \tikzstyle{level 4}=[sibling distance=8mm]
    \begin{scope}[xshift = -1.25cm]
    \path[level distance=10mm] node (root)[draw = gray, circle, inner sep=0pt, minimum size = 4mm]{$f$}
      child{
        node(0){$g$}
    child{
      node(00){$h$}
        child{
          node(000){$a$}
        }
    }
      }
      child{
        node(1){$a$}
      }
    ;
    \end{scope}
    \path[level distance=10mm] node (roota)[draw = gray, circle, inner sep=0pt, minimum size = 4mm]{$q$}
      (roota) edge [bend left=15, ->, >=stealth, below, color = gray] node {\small{$\varphi_0$}} (root)
    ;
  \end{tikzpicture}
  \caption{$c_0$}
  \end{subfigure}
  \begin{subfigure}{0.2\textwidth}
  \begin{tikzpicture}[scale=0.8,baseline=(current bounding box.base)]
    \tikzstyle{level 1}=[sibling distance=8mm]
    \tikzstyle{level 2}=[sibling distance=8mm]
    \tikzstyle{level 3}=[sibling distance=8mm]
    \tikzstyle{level 4}=[sibling distance=8mm]
    \begin{scope}[xshift = -1.5cm]
    \path[level distance=9mm] node (root){$f$}
      child{
        node(0)[draw = gray, circle, inner sep=0pt, minimum size = 4mm]{$g$}
    child{
      node(00){$h$}
        child{
          node(000){$a$}
        }
    }
      }
      child{
        node(1)[draw = gray, circle, inner sep=0pt, minimum size = 4mm]{$a$}
      }
    ;
    \end{scope}
    \path[level distance=9mm] node (root){$f$}
      child{
        node(a)[draw = gray, circle, inner sep=0pt, minimum size = 4mm]{$q$}
      }
      child{
        node(b)[draw = gray, circle, inner sep=0pt, minimum size = 4mm]{$q$}
      }
      (a) edge [bend left = 60, ->, >=stealth, below, color = gray] node {\small{$\varphi_1$}} (0)
      (b) edge [bend left = 60, ->, >=stealth, below, color = gray] node {\small{$\varphi_1$}} (1)
    ;
  \end{tikzpicture}
  \caption{$c_1$}
  \end{subfigure}
  \begin{subfigure}{0.2\textwidth}
  \begin{tikzpicture}[scale=0.8,baseline=(current bounding box.base)]
    \tikzstyle{level 1}=[sibling distance=8mm]
    \tikzstyle{level 2}=[sibling distance=8mm]
    \tikzstyle{level 3}=[sibling distance=8mm]
    \tikzstyle{level 4}=[sibling distance=8mm]
    \begin{scope}[xshift = -1.5cm]
    \path[level distance=9mm] node (root){$f$}
      child{
        node(0){$g$}
    child{
      node(00)[draw = gray, circle, inner sep=0pt, minimum size = 4mm]{$h$}
        child{
          node(000){$a$}
        }
    }
      }
      child{
        node(1)[draw = gray, circle, inner sep=0pt, minimum size = 4mm]{$a$}
      }
    ;
    \end{scope}
    \path[level distance=9mm] node (root){$f$}
      child{
        node(a)[draw = gray, circle, inner sep=0pt, minimum size = 4mm]{$q$}
      }
      child{
        node(b)[draw = gray, circle, inner sep=0pt, minimum size = 4mm]{$q$}
      }
      (a) edge [bend left, ->, >=stealth, below, color = gray] node {\small{$\varphi_2$}} (00)
      (b) edge [bend left = 60, ->, >=stealth, below, color = gray] node {\small{$\varphi_2$}} (1)
    ;
  \end{tikzpicture}
  \caption{$c_2$}
  \end{subfigure}
  \vskip 0.5em
  \begin{subfigure}{0.2\textwidth}
  \begin{tikzpicture}[scale=0.8,baseline=(current bounding box.base)]
    \tikzstyle{level 1}=[sibling distance=8mm]
    \tikzstyle{level 2}=[sibling distance=8mm]
    \tikzstyle{level 3}=[sibling distance=8mm]
    \tikzstyle{level 4}=[sibling distance=8mm]
    \begin{scope}[xshift = -1.5cm]
    \path[level distance=9mm] node (root){$f$}
      child{
        node(0){$g$}
    child{
      node(00)[draw = gray, circle, inner sep=0pt, minimum size = 4mm]{$h$}
        child{
          node(000){$a$}
        }
    }
      }
      child{
        node(1){$a$}
      }
    ;
    \end{scope}
    \path[level distance=9mm] node (root){$f$}
      child{
        node(a)[draw = gray, circle, inner sep=0pt, minimum size = 4mm]{$q$}
      }
      child{
        node(b){$a$}
      }
      (a) edge [bend left, ->, >=stealth, below, color = gray] node {\small{$\varphi_3$}} (00)
    ;
  \end{tikzpicture}
  \caption{$c_3$}
  \end{subfigure}
  \begin{subfigure}{0.2\textwidth}
  \begin{tikzpicture}[scale=0.8,baseline=(current bounding box.base)]
    \tikzstyle{level 1}=[sibling distance=8mm]
    \tikzstyle{level 2}=[sibling distance=8mm]
    \tikzstyle{level 3}=[sibling distance=8mm]
    \tikzstyle{level 4}=[sibling distance=8mm]
    \begin{scope}[xshift = -1.5cm]
    \path[level distance=9mm] node (root){$f$}
      child{
        node(0){$g$}
    child{
      node(00){$h$}
        child{
          node(000)[draw = gray, circle, inner sep=0pt, minimum size = 4mm]{$a$}
        }
    }
      }
      child{
        node(1){$a$}
      }
    ;
    \end{scope}
    \path[level distance=9mm] node (root){$f$}
      child{
        node(a){$h$}
        child{
          node(aa)[draw = gray, circle, inner sep=0pt, minimum size = 4mm]{$q$}
        }
      }
      child{
        node(b){$a$}
      }
      (aa) edge [bend left=15, ->, >=stealth, below, color = gray] node {\small{$\varphi_4$}} (000)
    ;
  \end{tikzpicture}
  \caption{$c_4$}
  \end{subfigure}
  \begin{subfigure}{0.2\textwidth}
  \begin{tikzpicture}[scale=0.8,baseline=(current bounding box.base)]
    \tikzstyle{level 1}=[sibling distance=8mm]
    \tikzstyle{level 2}=[sibling distance=8mm]
    \tikzstyle{level 3}=[sibling distance=8mm]
    \tikzstyle{level 4}=[sibling distance=8mm]
    \begin{scope}[xshift = -1.5cm]
    \path[level distance=9mm] node (root){$f$}
      child{
        node(0){$g$}
    child{
      node(00){$h$}
        child{
          node(000){$a$}
        }
    }
      }
      child{
        node(1){$a$}
      }
    ;
    \end{scope}
    \path[level distance=9mm] node (root){$f$}
      child{
        node(a){$h$}
        child{
          node(aa){$a$}
        }
      }
      child{
        node(b){$a$}
      }
    ;
  \end{tikzpicture}
  \caption{$c_5$}
  \end{subfigure}
  \begin{subfigure}{0.2\textwidth}
    \begin{tikzpicture}[scale=0.8,baseline=(current bounding box.base)]
      \tikzstyle{level 1}=[sibling distance=8mm]
      \tikzstyle{level 2}=[sibling distance=8mm]
      \tikzstyle{level 3}=[sibling distance=8mm]
      \tikzstyle{level 4}=[sibling distance=8mm]
      \begin{scope}[xshift = -1.5cm]
      \path[level distance=9mm] node (root){$f$}
        child{
          node(0){$g$}
      child{
        node(00){$h$}
          child{
            node(000){$a$}
          }
      }
        }
        child{
          node(1){$a$}
        }
      ;
      \end{scope}
      \path[level distance=9mm] node (roota){$f$}
        child{
          node(a){$h$}
          child{
            node(aa){$a$}
          }
        }
        child{
          node(b){$a$}
        }

        (roota) edge [bend left = 0, ->, >=stealth, below, color = gray] node {} (root)
        (a) edge [bend left = 30, ->, >=stealth, below, color = gray] node {} (00)
        (aa) edge [bend left = 30, ->, >=stealth, below, color = gray] node {} (000)
        (b) edge [bend left = 60, ->, >=stealth, below, color = gray] node {} (1)

      ;
    \end{tikzpicture}
    \caption{origin mapping}
    \label{fig:intro-mapping}
    \end{subfigure}

    \caption{The configuration sequence $c_0$ to $c_5$ of $\mathcal T$ on $f(g(h(a)),a)$ and resulting origin mapping from Example \ref{ex:transducer}.}
    \label{fig:configuration}
    \end{figure}
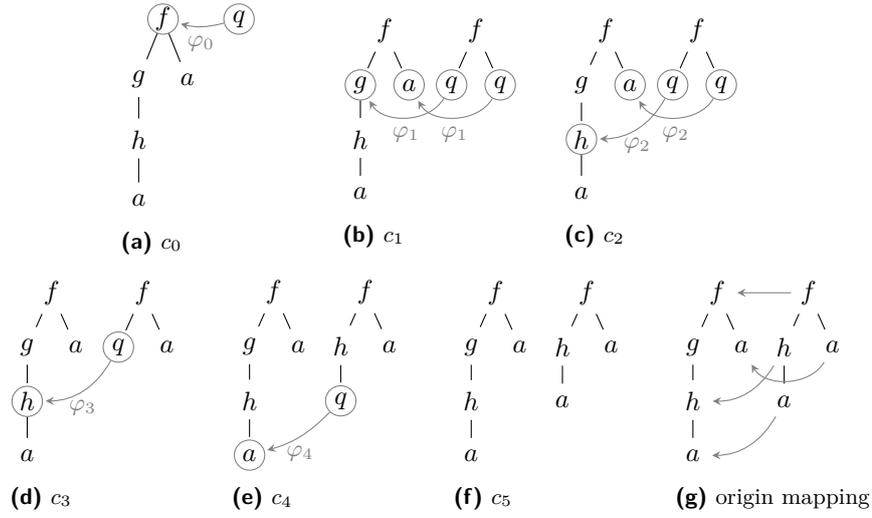

\begin{example}\label{ex:zwei}
Let $\Sigma,\Gamma$ be given by $\Sigma_2 = \{f\}$, $\Sigma_0 = \{a\}$, $\Gamma_1=\{h\}$, and $\Gamma_0 = \{b\}$.
Consider the \tdtt\ $\mathcal T$ given by $(\{q\},\Sigma,\Gamma,\{q\},\Delta)$ with $\Delta$ $=$ $\{$ $q(a) \rightarrow b$, $q(f(x_1,x_2)) \rightarrow h(q(x_1))$, $q(f(x_1,x_2)) \rightarrow h(q(x_2))$ $\}$.
Basically, when reading an $f$-labeled node, the \tdtt can non-deterministically decide whether to continue reading in left or the right subtree.
In \cref{fig:k-origin} two origin mappings $o\colon\dom{s} \to \dom{t}$ and $o'\colon\dom{s} \to \dom{t}$ are given that are result of runs of $\mathcal T$ on $t = f(f(f(a,a),a),a),f(f(a,a),a),a))$ with final output $s = h(h(h(b)))$.
\end{example}

 
\begin{figure}[t!]
    \begin{center}

      \begin{subfigure}{0.4\textwidth}
        \begin{tikzpicture}[scale=0.8,baseline=(current bounding box.base)]
          \tikzstyle{level 1}=[sibling distance=16mm]
          \tikzstyle{level 2}=[sibling distance=8mm]
          \tikzstyle{level 3}=[sibling distance=8mm]
          \tikzstyle{level 4}=[sibling distance=8mm]
          \begin{scope}[xshift = -2.5cm]
          \path[level distance=9mm] node (root){$f$}
          child{
              node(0){$f$}
            child{
              node(00){$f$}
                child{
                node(000){$a$}
                }
                child{node {$a$}}
            }
          child{node {$a$}}
          }
          child{
            node(1){$f$}
          child{
            node(10){$f$}
              child{
              node(100){$a$}
              }
              child{node {$a$}}
          }
        child{node {$a$}}
        }
          ;
          \end{scope}
          \path[level distance=9mm] node (roota){$h$}
            child{
              node(a){$h$}
              child{
                node(aa){$h$}
                 child{
                   node(aaa) {$b$}
                 }
              }
            }

  (roota) edge [bend right = 20, ->, >=stealth, below, color = gray] node {} (root)
  (a) edge [bend right = 20, ->, >=stealth, below, color = gray] node {} (0)
  (aa) edge [bend right = 20, ->, >=stealth, below, color = gray] node {} (00)
  (aaa) edge [bend right = 20, ->, >=stealth, below, color = gray] node {} (000)
    
          ;
        \end{tikzpicture}
        \caption{$o\colon\dom{s} \to \dom{t}$}
        \end{subfigure}
        \begin{subfigure}{0.4\textwidth}
          \begin{tikzpicture}[scale=0.8,baseline=(current bounding box.base)]
            \tikzstyle{level 1}=[sibling distance=16mm]
            \tikzstyle{level 2}=[sibling distance=8mm]
            \tikzstyle{level 3}=[sibling distance=8mm]
            \tikzstyle{level 4}=[sibling distance=8mm]
            \begin{scope}[xshift = -2.5cm]
            \path[level distance=9mm] node (root){$f$}
            child{
                node(0){$f$}
              child{
                node(00){$f$}
                  child{
                  node(000){$a$}
                  }
                  child{node {$a$}}
              }
            child{node {$a$}}
            }
            child{
              node(1){$f$}
            child{
              node(10){$f$}
                child{
                node(100){$a$}
                }
                child{node {$a$}}
            }
          child{node {$a$}}
          }
            ;
            \end{scope}
            \path[level distance=9mm] node (roota){$h$}
              child{
                node(a){$h$}
                child{
                  node(aa){$h$}
                   child{
                     node(aaa) {$b$}
                   }
                }
              }

    (roota) edge [bend right = 20, ->, >=stealth, below, color = gray] node {} (root)
    (a) edge [bend right = 20, ->, >=stealth, below, color = gray] node {} (1)
    (aa) edge [bend right = 20, ->, >=stealth, below, color = gray] node {} (10)
    (aaa) edge [bend right = 20, ->, >=stealth, below, color = gray] node {} (100)
      
            ;
          \end{tikzpicture}
          \caption{$o'\colon\dom{s} \to \dom{t}$}
          \end{subfigure}

    \end{center}
     \caption{Origin mappings $o,o'$. We have that $\mathit{dist}(o(111),o'(111))$, that is, the distance of the origins of the leaf node, is the length of the shortest path from node $111$ to node $211$ which is $6$.
     }
     \label{fig:k-origin}
    \end{figure}
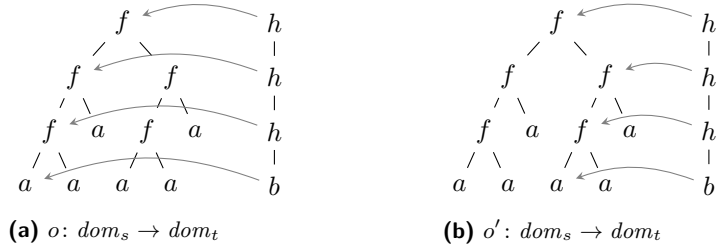

In this work, we focus on decision problems for transducers with origin semantics.
To begin with, we introduce some notations and state relevant known results in this context.

\mysubparagraph{Shorthand notations.}

Let $\mathcal C$ denote a class of transducers with origin semantics, e.g., \tdtt or \dtdtt.
Given a class $\mathcal C$ and $\mathcal T_1,\mathcal T_2 \in \mathcal C$, if $R(\mathcal T_1) \subseteq R(\mathcal T_2)$ (resp.\ $R_o(\mathcal T_1) \subseteq R_o(\mathcal T_2)$), we write $\mathcal T_1 \subseteq \mathcal T_2$ (resp.\ $\mathcal T_1 \subseteq_o \mathcal T_2$).
Furthermore, if $R(\mathcal T_1) = R(\mathcal T_2)$ (resp.\ $R_o(\mathcal T_1) = R_o(\mathcal T_2)$), we write $\mathcal T_1 = \mathcal T_2$ (resp.\ $\mathcal T_1 =_o \mathcal T_2$).
Given classes $\mathcal C_1,\mathcal C_2$, $\mathcal T_1 \in \mathcal C_1$, and $\mathcal T_2 \in \mathcal C_2$, if $\mathcal T_1$ defines a function $f$ that is a uniformization of $R(\mathcal T_2)$, we say $\mathcal T_1$ uniformizes $\mathcal T_2$, if additionally $\mathcal T_1 \subseteq_o \mathcal T_2$, we say $\mathcal T_1$ origin uniformizes $\mathcal T_2$.

\mysubparagraph{Decision problems.}

The \emph{inclusion} resp.\ \emph{origin inclusion problem} for a class $\mathcal C$ asks, given $\mathcal T_1, \mathcal T_2 \in \mathcal C$, whether $\mathcal T_1 \subseteq \mathcal T_2$ resp.\ $\mathcal T_1 \subseteq_o \mathcal T_2$.
The \emph{equivalence} resp.\ \emph{origin equivalence problem} for a class $\mathcal C$ asks, given $\mathcal T_1, \mathcal T_2 \in \mathcal C$, whether $\mathcal T_1 = \mathcal T_2$ resp.\ $\mathcal T_1 =_o \mathcal T_2$.
Lastly, the \emph{uniformization} resp.\ \emph{origin uniformization problem} for classes $\mathcal C_1,\mathcal C_2$ asks, given $\mathcal T_2 \in \mathcal C_2$, whether there exists $\mathcal T_1 \in \mathcal C_1$ such that $\mathcal T_1$ uniformizes (resp.\ origin uniformizes) $\mathcal T_2$.

As mentioned in the introduction, generally, a transduction can be defined by several transducers behaving very differently, making many problems intractable.
Adding origin semantics to transducers, i.e., seeing the transducer behavior as part of the transduction, allows to recover decidability.
The following is known for the class \ntdtt.

\begin{restatable}[\cite{esik1980decidability}]{theorem}{thmundec}\label{thm:undec}
  Inclusion and equivalence are undecidable for the class \ntdtt.
\end{restatable}

\begin{restatable}[\cite{DBLP:journals/iandc/FiliotMRT18}]{theorem}{thmorigindec}\label{thm:origin-dec}
  \mbox{Origin inclusion and origin equivalence are decidable for the class \ntdtt.}
\end{restatable}

Turning to uniformization problems, it is known that every \tdtt is uniformizable by a \dtdtt with regular lookahead (a \dtdttr), that is, the transducer can check membership of the subtrees of a node in regular tree-languages before processing the node.

\begin{restatable}[\cite{DBLP:journals/mst/Engelfriet77}]{theorem}{thmalways}\label{thm:always}
  Every \ntdtt has a \dtdttr-uniformization.
\end{restatable}

However, when requiring that the input should be transformed on-the-fly (without regular lookahead), the uniformization problem becomes undecidable.
In \cite{CarayolL14}, it was shown that it is undecidable whether a one-way (non-deterministic) word transducer has a uniformization by a sequential transducer (that is, basically, a one-way deterministic transducer).
So, we get undecidability in the tree setting for free (as stated in \cref{thm:undec-unif}).
This problem has not been investigated with origin semantics so far.
We show decidability (also for more relaxed versions), see \cref{thm:korigin}.

\begin{restatable}{theorem}{thmundecunif}\label{thm:undec-unif}
  \dtdtt-uniformization is undecidable for the class \ntdtt.
\end{restatable}

Since the origin semantics is rather rigid, in the next section, we introduce two similarity measures between transducers which are based on their behavior and re-investigate the introduced decision problems for transducers with `similar' behavior.

\section{Similarity measures for transducers}
\label{sec:similarity}

An idea that naturally comes to mind is to say that two transducers behave similarly if for two computations over the same input that yield the same output their respective origin mappings are `similar'.

The other idea is to say that two computations are similar if their output delay is small, roughly meaning that for the same prefix (for an adequate notion of prefix for trees) of the input the so-far produced output is of similar size.
Decision problems using this measure have already been investigated for (one-way) word transducers \cite{DBLP:conf/icalp/FiliotJLW16}, we lift the measure to top-down tree transducers.

\mysubparagraph{Origin distance.}

Given a tree $t$, let the distance between two nodes $u,v \in \dom{t}$, written $\mathit{dist}(u,v)$, be the shortest path between $u$ and $v$ (ignoring the edge directions), an example is given in \cref{fig:k-origin}.

Given $\mathcal T,\mathcal T_1,\mathcal T_2 \in \mathcal C$, where $\mathcal C$ is a class of transducers with origin semantics. 
We say $(t,s,o)$ is \emph{$k$-origin included} in $R_o(\mathcal T)$, written $(t,s,o) \in_k R_o(\mathcal T)$, if there is $(t,s,o') \in R_o(\mathcal T)$ such that $\mathit{dist}(o(i),o'(i)) \leq k$ for all $i \in \dom{s}$.
We say $\mathcal T_1$ is \emph{$k$-origin included} in $\mathcal T_2$, written $\mathcal T_1 \subseteq_k \mathcal T_2$, if $(s,t,o) \in_k R_o(\mathcal T_2)$ for all $(s,t,o) \in R_o(\mathcal T_1)$.
We say $\mathcal T_1$ and $\mathcal T_2$ are \emph{$k$-origin equivalent}, written $\mathcal T_1 =_k \mathcal T_2$, if $\mathcal T_1 \subseteq_k \mathcal T_2$ and $\mathcal T_2 \subseteq_k \mathcal T_1$.
We say $\mathcal T_1$ \emph{$k$-origin uniformizes} $\mathcal T_2$ if $\mathcal T_1 \subseteq_k \mathcal T_2$ and $\mathcal T_1$ uniformizes $\mathcal T_2$.
The $k$-origin decision problems are defined as expected.

We need some additional notations, before we can introduce the concept of delay.

\mysubparagraph{Partial and prefix trees.}

Let $N_{\SigmaI}$ be the set of all trees over $\SigmaI$ which can have symbols from $\SigmaI$, that is, symbols with rank $\geq 0$, at their leaves.
The set $N_\Sigma$ is the set of all \emph{partial trees} over $\Sigma$.
Note that $N_\Sigma$ includes $T_\Sigma$.  
We say a tree $t' \in N_\Sigma$ is a \emph{prefix tree} of a tree $t \in N_\Sigma$, written $t' \sqsubseteq t$, if $\dom{t'} \subseteq \dom{t}$, and $\val{t'}(u) = \val{t}(u) \text{ for all } u\in\dom{t'}.$
Given $t_1,t_2 \in N_\Sigma$, its \emph{greatest common prefix}, written $t_1 \wedge t_2$, is the tree $t \in N_\Sigma$ such that $\dom{t}$ is the largest subset of $\dom{t_1} \cap \dom{t_2}$ such that $t \sqsubseteq t_1$ and $t \sqsubseteq t_2$.
Removing $t_1 \wedge t_2$ from $t_1$ and $t_2$ naturally yields a set of partial trees (we omit a formal definition) called \emph{difference trees}.
These notions are visualized in \cref{fig:k-delay}.
 
 
\mysubparagraph{Delay.}

Given words $w_1,w_2$, to compute their delay, we remove their greatest common prefix $w = w_1 \wedge w_2$, say $w_1 = wv_1$ and $w_2 = wv_2$, and their delay is the maximum of the length of their respective reminders, i.e., $\mathrm{max} \{|v_1|,|v_2|\}$.
We lift this to trees, given (partial) trees $t_1,t_2$, we remove their greatest common prefix $t_1 \wedge t_2$ from $t_1$ and $t_2$ which yields a set $S$ of partial trees, we define their delay as $\mathit{delay}(t_1,t_2) = \mathrm{max} \{ h(t)+1 \mid t \in S\}$.
An example is given in \cref{fig:k-delay}.
Note that for trees over unary and leaf symbols (a way to see words) the definitions for words and trees are equal.
Recall that the length of the word $a$ is one, but the height of the tree $a$ is zero.

In order to define a similarity measure between transducers using delay, we take two transducer runs on the same input and compute the delay between their produced outputs throughout their runs.
Although we have defined delay between words and trees, we only provide a formal definition for top-down tree transducers.
However, for word transducers, examples are given in \cref{ex:incomparable}, and a formal definition can be found in \cite{DBLP:conf/icalp/FiliotJLW16}. 

Now, let $\mathcal T_1$ and $\mathcal T_2$ be {\tdtt}s, and $\rho_1$ and $\rho_2$ be runs of $\mathcal T_1$ and $\mathcal T_2$, respectively, over the same input tree $t \in T_\Sigma$ such that $\rho_1\colon (t,q_0^{\mathcal T_1},\varphi_0) \rightarrow_{\mathcal T_1}^* (t,t_1,\varphi) \text{ with } t_1 \in T_\Gamma$, and $\rho_2\colon (t,q_0^{\mathcal T_2},\varphi_0) \rightarrow_{\mathcal T_2}^* (t,t_2,\varphi') \text{ with } t_2 \in T_\Gamma$.

Basically, we take a look at all configurations that occur in the runs and compute the delay between the output trees of compatible configurations where compatible means in both configurations the same prefix (level-wise, see below) of the input tree has been processed.

Let us be a bit more clear what we mean with compatible.
Note that when comparing two configuration sequences (i.e., runs) of word transducers the notion of `have processed the same input so far' is clear.
For tree transducers, in one configuration sequence, a left-hand subtree might be processed before the right-hand subtree, and in another configuration sequence vice versa.
Since these computation steps are done in a parallel fashion (just written down in an arbitrary order in the configuration sequence), we need to make sure to compare configurations where the subtrees have been processed equally far (we call this level-wise).
Also, a tree transducer might not even read the whole input tree, as, e.g., in \cref{ex:zwei}.
We also (implicitly) take care of this in our definition.

The result is the maximum of the delay between output trees of compatible configurations.
Given $t \in T_\Sigma$, let $\mathit{Prefs}_{\mathit{level}}(t)$ denote the set of all prefix trees of $t$ such that if a node at level $i$ is kept, then all other nodes at level $i$ are kept, i.e., for $t = f(h(a),h(a))$, $\mathit{Prefs}_{\mathit{level}}(t)$ contains $f(h,h)$, but not $f(h(a),h)$.
Given an intermediate configuration $(t,t'_i,\varphi')$ of the run $\rho_i$, we recall that $t'_i \in T_{\Gamma \cup Q_{\mathcal T_i}}$ meaning $t'_i$ contains states of $\mathcal T_i$ as leaves.
 Let $t'_i|\Gamma$ denote the partial tree obtained from $t_i$ by removing all non-$\Gamma$-labelled nodes.
 We define $\mathit{delay}(\rho_1,\rho_2)$ as
 \begin{align*}
  \hspace{-5pt}\max\{ \mathit{delay}(t'_1|\Gamma,t'_2|\Gamma) \mid \ & \text{there is } t' \in \mathit{Prefs}_{\mathit{level}}(t), \text{ there is } (t,t_1',\varphi') \text{ in } \rho_1 \text{ with } t_1' \in \mathcal T_1(t'),\\
  & \text{and, there is } (t,t_2',\varphi'') \text{ in } \rho_2 \text{ with } t_2' \in \mathcal T_2(t')\ \}.
 \end{align*}
 The conditions $t_1' \in \mathcal T_1(t')$ and $t_2' \in \mathcal T_2(t')$ are introduced to make sure that all input nodes that can be processed from $t'$ are processed in the selected configurations.

 
  
\begin{figure}[t!]
    \begin{center}

      \begin{subfigure}{0.3\textwidth}
        \begin{tikzpicture}[thick,baseline=(current bounding box.north)]
        \tikzstyle{level 1}=[sibling distance=16mm]
         \tikzstyle{level 2}=[sibling distance=10mm]
      
         \path[level distance=8mm] node[draw = black, circle, inner sep=0pt, minimum size = 4mm] (root){$g$}
          child{
            node[draw = black, circle, inner sep=0pt, minimum size = 4mm](0){$f$}
              child{
                node[draw = black, circle, inner sep=0pt, minimum size = 4mm](00){$a$}
              }
              child{
                node[draw = black, circle, inner sep=0pt, minimum size = 4mm](01){$f$}
                  child{
                    node(010){$b$}
                  }
                  child{
                    node(011){$h$}
                     child{
                      node(0110){$a$}
                     }
                  }
              }
          }
          child{
            node[draw = black, circle, inner sep=0pt, minimum size = 4mm](1){$f$}
          }
        ;
        
        \end{tikzpicture}
        \caption{partial tree $t_1$}
        \end{subfigure}
        \begin{subfigure}{0.3\textwidth}
          \begin{tikzpicture}[thick,baseline=(current bounding box.north)]
            \tikzstyle{level 1}=[sibling distance=16mm]
             \tikzstyle{level 2}=[sibling distance=10mm]

            \path[level distance=8mm] node[draw = black, circle, inner sep=0pt, minimum size = 4mm] (root){$g$}
              child{
                node[draw = black, circle, inner sep=0pt, minimum size = 4mm](0){$f$}
                  child{
                    node[draw = black, circle, inner sep=0pt, minimum size = 4mm](00){$a$}
                  }
                  child{
                    node[draw = black, circle, inner sep=0pt, minimum size = 4mm](01){$f$}
                  }
              }
              child{
                node[draw = black, circle, inner sep=0pt, minimum size = 4mm](1){$f$}
                child{
                  node(10){$a$}
                }
                child{
                 node(11){$b$}
                }
              }
            ;
            
            \end{tikzpicture}
          \caption{partial tree $t_2$}
          \end{subfigure}
        \begin{subfigure}{0.3\textwidth}
          \quad
        \begin{tikzpicture}[thick,baseline=(current bounding box.north)]
          \tikzstyle{level 1}=[sibling distance=16mm]
          \tikzstyle{level 2}=[sibling distance=10mm]
            \path[level distance=8mm] node (root){$b$}
            ;
            \begin{scope}[xshift=0.75cm]
            \path[level distance=10mm] node (root){$h$}
            child{
            node(0){$a$}
            }
            ;
            \end{scope}
            \begin{scope}[xshift=1.5cm]
            \path[level distance=10mm] node (root){$a$}
            ;
            \end{scope}
            \begin{scope}[xshift=2.25cm]
            \path[level distance=10mm] node (root){$b$}
            ;
            \end{scope}
          \end{tikzpicture}
        \caption{partial trees resulting from removing $t_1 \wedge t_2$ from $t_1$ and $t_2$}
        \end{subfigure}

    \end{center}
    \caption{The greatest common prefix of the partial trees $t_1$ and $t_2$, $t_1 \wedge t_2$, is marked with circles in $t_1$ and $t_2$.
    The delay between $t_1$ and $t_2$ is computed from their non-common parts as $\mathit{delay}(t_1,t_2) = \mathrm{max} \{ h(t) + 1 \mid t \in \{b,h(a),a\}\} = 2$. 
    }
    \label{fig:k-delay}
    \end{figure}
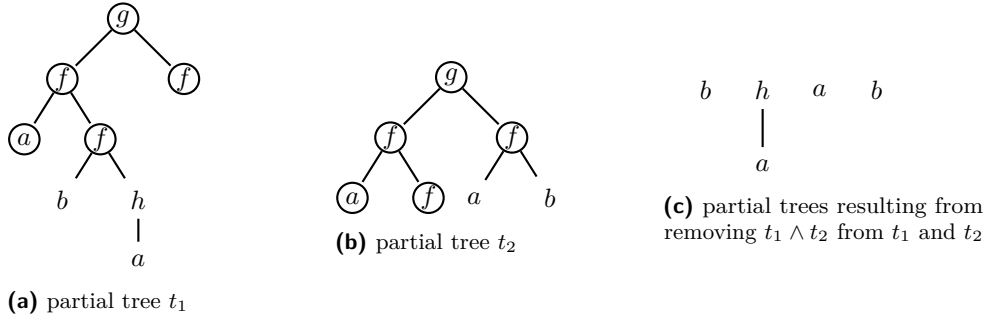

We introduce (shorthand) notations.
Let $\mathcal T_1,\mathcal T_2 \in \mathcal C$, where $\mathcal C$ is a class of transducers.
Given $(t,s) \in R(\mathcal T_1)$, we say $(t,s)$ is \emph{$k$-delay included} in $R(\mathcal T_2)$, written $(t,s) \in_{\mathbb D_k} R(\mathcal T_2)$, if there are runs $\rho$ and $\rho'$ of $\mathcal T_1$ and $\mathcal T_2$, respectively, with input $t$ and output $s$ such that $\mathit{delay}(\rho,\rho') \leq k$.
We say $\mathcal T_1$ is \emph{$k$-delay included} in $\mathcal T_2$, written $\mathcal T_1 \subseteq_{\mathbb D_k} \mathcal T_2$, if $(t,s) \in_{\mathbb D_k} R(\mathcal T_2)$ for all $(t,s) \in R(\mathcal T_1)$.
We say $\mathcal T_1$ and $\mathcal T_2$ are \emph{$k$-delay equivalent}, written $\mathcal T_1 =_{\mathbb D_k} \mathcal T_2$, if $\mathcal T_1 \subseteq_{\mathbb D_k} \mathcal T_2$ and $\mathcal T_2 \subseteq_{\mathbb D_k} \mathcal T_1$.
We say $\mathcal T_1$ \emph{$k$-delay uniformizes} $\mathcal T_2$ if $\mathcal T_1 \subseteq_{\mathbb D_k} \mathcal T_2$ and $\mathcal T_1$ uniformizes $\mathcal T_2$.
The $k$-delay decision problems are defined as expected.




In order to get a better understanding of the expressiveness and differences between the two similarity measures, we first explore their properties on word transductions since words are a particular case of trees (i.e., monadic trees).

\mysubparagraph{Word transducers.}
\label{sec:similarword}

We denote by \fst a \emph{finite state transducer}, the class of word transductions recognized by \fst{s} is the class of \emph{rational transductions}, conveniently also denoted by \fst.
We omit a formal definition of \fst{s}, because they are not considered outside of this section.
An \fst is \emph{sequential} if its transitions are input-deterministic, more formally, a \dtdtt over ranked alphabets with only unary and leaf symbols can be seen as an \fst.

The results below concern origin distance, the same results were proven for delay in \cite{DBLP:conf/icalp/FiliotJLW16}.

\begin{restatable}{proposition}{proporigin}\label{prop:origin}
  \begin{enumerate}
    \item There exist \fst{s} $\mathcal T_1$, $\mathcal T_2$ such that $\mathcal T_1 \subseteq \mathcal T_2$, but $\mathcal T_1 \not\subseteq_{k} \mathcal T_2$ for all $k \geq 0$.
    \item There exist \fst{s} $\mathcal T_1$, $\mathcal T_2$ such that $\mathcal T_1 = \mathcal T_2$, but $\mathcal T_1 \neq_{k} \mathcal T_2$ for all $k \geq 0$.
    \item There exists an \fst $\mathcal T$ such that $\mathcal T$ is sequentially uniformizable, but $\mathcal T$ is not $k$-origin sequentially uniformizable for all $k \geq 0$.
  \end{enumerate}
\end{restatable}

\begin{proof}
First, consider the \fst{s} $\mathcal T_1, \mathcal T_2$ depicted in \cref{fig:example-a}.
Both recognize the same function $f\colon \{a,b,c\}^* \to \{a\}^*$ defined as $f(ab^*c) = a$.
Clearly, $\mathcal T_1 \subseteq \mathcal T_2$ and $\mathcal T_1 = \mathcal T_2$.
However, the origin distance between $\mathcal T_1$ and $\mathcal T_2$ is unbounded.
In $\mathcal T_1$, the origin of the single output letter $a$ is always the first input letter.
In contrast, in $\mathcal T_2$, the origin of the output letter $a$ is always the last input letter.
Secondly, consider the \fst{s} $\mathcal T, \mathcal T'$ depicted in \cref{fig:example-a}.
The recognized relation $\mathcal R(\mathcal T) \subseteq \{a,b,A,B\}^* \times \{a,b\}^*$ consists of $\{(ab^nA,ab^n) \mid n \in \mathbbm N \}$ and $\{(ab^nB,ab^m) \mid 0 \leq m \leq 2n-1, n \in \mathbbm N\}$.
The sequential transducer $\mathcal T'$ recognizes the function $f\colon \{a,b,A,B\}^* \to \{a,b\}^*$ defined by $f(ab^nX) = ab^n$ for $X \in \{A,B\}$ and all $n\in\mathbbm N$.
Clearly, $\mathcal T'$ is a sequential uniformization of $\mathcal T$.
However, no sequential uniformization with bounded origin distance exists, see \cref{app:body}.
\end{proof}

\begin{figure}[t!]

\begin{subfigure}{\textwidth}
  \begin{tikzpicture}[thick,node distance=2.5em]
    \tikzstyle{every state}+=[inner sep=4pt, minimum size=3pt];
    \node[state, initial, initial text=$\mathcal T_1$] (0) {};
    \node[state, right of= 0]                          (1) {};
    \node[state, accepting, right of= 1]               (2) {};

    \draw[->] (0) edge               node  {$a|a$} (1);
    \draw[->] (1) edge[loop above]   node  {$b|\varepsilon$} ();       
    \draw[->] (1) edge               node  {$c|\varepsilon$} (2);

    \begin{scope}[xshift=10em]
    \tikzstyle{every state}+=[inner sep=4pt, minimum size=3pt];
    \node[state, initial, initial text=$\mathcal T_2$] (0) {};
    \node[state, right of= 0]                          (1) {};
    \node[state, accepting, right of= 1]               (2) {};

    \draw[->] (0) edge               node  {$a|\varepsilon$} (1);
    \draw[->] (1) edge[loop above]   node  {$b|\varepsilon$} ();       
    \draw[->] (1) edge               node  {$c|a$} (2);
    \end{scope}

    \begin{scope}[xshift=20em]
    \tikzstyle{every state}+=[inner sep=4pt, minimum size=3pt];
    \node[state, initial, initial text=$\mathcal T$] (0) {};
    \node[state, right of= 0]                          (1) {};
    \node[state, accepting, right of= 1]               (2) {};

    \node[state, below of= 0] (00) {};
    \node[state, right of= 00]                          (01) {};
    \node[state, accepting, right of= 01]               (02) {};

    \draw[->] (0) edge               node  {$a|a$} (1);
    \draw[->] (1) edge[loop above]   node  {$b|b$} ();       
    \draw[->] (1) edge               node  {$A|\varepsilon$} (2);

    \draw[->] (0) edge               node[swap]  {$a|a$} (00);

    \draw[->] (00) edge               node  {$b|b$} (01);
    \draw[->] (00) edge[swap]         node  {$b|\varepsilon$} (01);
    \draw[->] (00) edge[loop below]   node  {$b|\varepsilon$} ();  
    \draw[->] (01) edge[loop below]   node  {$b|bb$} ();       
    \draw[->] (01) edge               node  {$B|\varepsilon$} (02);
    \end{scope}

    \begin{scope}[xshift=30em]
    \tikzstyle{every state}+=[inner sep=4pt, minimum size=3pt];
    \node[state, initial, initial text=$\mathcal T'$] (0) {};
    \node[state, right of= 0]                          (1) {};
    \node[state, accepting, right of= 1]               (2) {};

    \draw[->] (0) edge               node  {$a|a$} (1);
    \draw[->] (1) edge[loop above]   node  {$b|b$} ();       
    \draw[->] (1) edge               node  {$A|\varepsilon$} (2);
    \draw[->] (1) edge               node[swap]  {$B|\varepsilon$} (2);
    \end{scope}
  \end{tikzpicture}
\caption{We have $\mathcal T_1 = \mathcal T_2$, and $\mathcal T_1 =_{\mathbb D_1} \mathcal T_2$, but $\mathcal T_1 \neq_{k} \mathcal T_2$ for all $k \geq 0$; $\mathcal T'$ is a sequential uniformizer of $\mathcal T$.}
\label{fig:example-a}
\end{subfigure}

\begin{subfigure}{\textwidth}
  \begin{tikzpicture}[thick,node distance=2.5em]
    \tikzstyle{every state}+=[inner sep=4pt, minimum size=3pt];
    \node[state, initial, initial text=$\mathcal T_3$] (0) {};
    \node[state, right of= 0]                          (1) {};
    \node[state, accepting, right of= 1]               (2) {};

    \draw[->] (0) edge               node  {$a|a$} (1);
    \draw[->] (1) edge[loop above]   node  {$\varepsilon|c$} ();       
    \draw[->] (1) edge               node  {$b|\varepsilon$} (2);

    \begin{scope}[xshift=10em]
      \tikzstyle{every state}+=[inner sep=4pt, minimum size=3pt];
      \node[state, initial, initial text=$\mathcal T_4$] (0) {};
      \node[state, right of= 0]                          (1) {};
      \node[state, accepting, right of= 1]               (2) {};

      \draw[->] (0) edge               node  {$a|a$} (1);
      \draw[->] (2) edge[loop above]   node  {$\varepsilon|c$} ();       
      \draw[->] (1) edge               node  {$b|\varepsilon$} (2);
    \end{scope}
    \end{tikzpicture}
\caption{We have $\mathcal T_3 = \mathcal T_4$, and $\mathcal T_3 =_1 \mathcal T_4$, but $\mathcal T_3 \neq_{\mathbb D_k} \mathcal T_4$ for all $k \geq 0$.}
\label{fig:example-b}
\end{subfigure}

\caption{Comparing origin distance and delay for word transducers, see the proof of \cref{prop:origin} and \cref{ex:incomparable}.}
\label{fig:example}
\end{figure}

We give an example (depicted in \cref{fig:example} and described in detail in \cref{ex:incomparable}) that shows that the two notions are orthogonal to each other.
However, if we restrict the class \fst to \emph{real-time}\footnote{$\varepsilon$-transitions (as, e.g., the loop in $\mathcal T_3$ from \cref{fig:example}, not be confused with non-producing transitions) are standard for \fst, and non-standard for \tdtt in the literature. We consider `real-time' \tdtt by default. } \fst, that is, word transducers such that in every transition exactly one input symbol is read, the notion of delay is more powerful than origin distance, see below.
It is important to note that we have proven \cref{prop:origin} for real-time \fst{s} which are equivalent to \tdtt{s} on monadic trees, i.e., \cref{prop:origin} is true for the class \tdtt.

\begin{restatable}{proposition}{lemmaoriginincludedelay}\label{lemma:origin-include-delay}
Let $\mathcal T_1, \mathcal T_2$ be real-time \fst{s}, if $\mathcal T_1 \subseteq_{i} \mathcal T_2$ for some $i \geq 0$, then $\mathcal T_1 \subseteq_{\mathbb D_{j}} \mathcal T_2$ for some $j \geq 0$.
\end{restatable}

The notion of bounded delay is suitable to regain decidability.

\begin{restatable}[\cite{DBLP:conf/icalp/FiliotJLW16}]{theorem}{thmkdelayword}\label{thm:kdelayword}
  Given $k \geq 0$, $k$-delay inclusion, $k$-delay equivalence and $k$-delay sequential uniformization are decidable for the class \nfst.
\end{restatable}

%
Ideally, we would like to lift \cref{thm:kdelayword} from word to tree transducers, but it turns out that the notion of delay is too expressive to yield decidability results for tree transducers as shown in the next paragraph.

\mysubparagraph{Tree transducers.}
\label{sec:similartree}

It is undecidable whether a given tree-automatic relation has a uniformization by a synchronous \dtdtt \cite{DBLP:conf/mfcs/LodingW16}.
A \tdtt is called synchronous if for one processed input node one output node is produced as, e.g., in \cref{ex:zwei}.
Tree-automatic relations are a subclass of the relations that are recognizable by synchronous \tdtt.
To prove the result, the authors showed that

\begin{restatable}[\cite{DBLP:conf/mfcs/LodingW16}]{lemma}{lemmadelayundec}\label{lemma:delayundec}
  There exists a synchronous \tdtt $\mathcal T_M$, based on a Turing machine $M$, that is 0-delay \dtdtt-uniformizable iff $M$ halts on the empty input.
\end{restatable}

In the proof, for a TM $M$, a \dtdtt $\mathcal T_M'$ is constructed such that $\mathcal T_M'$ $0$-delay uniformizes $\mathcal T_M$ iff $M$ halts on the empty input.
Recall that this implies that $\mathcal T_M' \subseteq_{\mathbb D_0} \mathcal T_M$ iff $M$ halts on the empty input.
Consequently, we obtain that

\begin{restatable}{theorem}{thmkdelayundec}\label{thm:k-delay-undec}
Given $k \geq 0$, $k$-delay inclusion and $k$-delay \dtdtt-uniformization are undecidable for the class \ntdtt (even for $k = 0$).
\end{restatable}

We do not know whether $k$-equivalence is decidable for a given $k \geq 0$.
Note that \cref{thm:k-delay-undec} does not imply that $0$-origin inclusion and $0$-origin \dtdtt-uniformization is undecidable for the class \tdtt.
For the class \fst, the notions of $0$-origin and $0$-delay fall together, but for \tdtt this is no longer the case.
Recall the $\tdtt$ given in \cref{ex:zwei} and its unique runs that yield the origin mappings depicted in \cref{fig:k-origin}.
The delay between these runs is zero, but their origin mappings are different.
An analysis of the (un)decidability proof(s) in \cite{DBLP:conf/mfcs/LodingW16} pins the problems down to the fact that in the specification and in the possible implementations the origins for the same output node lie on different paths in the input tree.
For trees, this fact has no influence when measuring the delay between computations (as seen in \cref{ex:zwei}).
However, it is recognizable using the origin distance as measure.
Since the notion of delay is so powerful that the decision problems under bounded delay become undecidable for tree transducers (see \cref{thm:k-delay-undec}) in contrast to word transducers (see \cref{thm:kdelayword}), in the next section, we focus on bounded origin distance.

\section{Decision problems for origin-close transducers}
\label{sec:origin-close}

We show that the decision problems become decidable for top-down tree transducers with bounded origin distance, see \cref{thm:korigin}.
The next part is devoted to explaining our proof ideas and introducing our main technical lemma (\cref{lemma:regular}) which is used in all proofs.

\mysubparagraph{Origin-close transductions are representable as regular tree languages.}
\label{sec:regular}

 
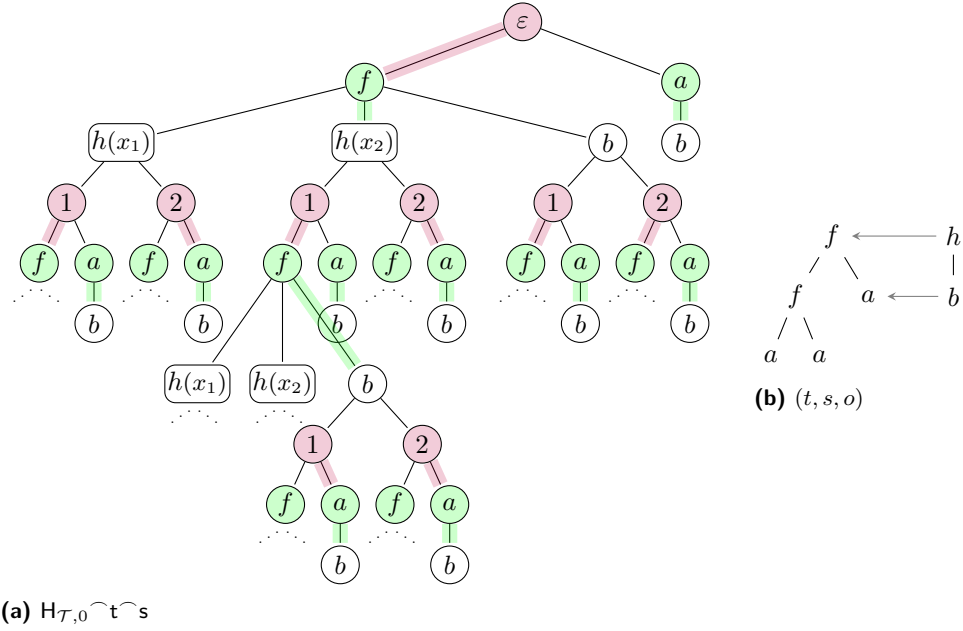
\begin{figure}[t!]
  \begin{center}

    \begin{subfigure}{0.7\textwidth}
      \begin{tikzpicture}[scale=0.8,baseline=(current bounding box.base)]
        \tikzstyle{level 1}=[sibling distance=52mm]
        \tikzstyle{level 2}=[sibling distance=40mm]
        \tikzstyle{level 3}=[sibling distance=18mm]
        \tikzstyle{level 4}=[sibling distance=9mm]
        \tikzstyle{level 6}=[sibling distance=18mm]
        \tikzstyle{level 7}=[sibling distance=9mm]

        \tikzstyle{myblack}=[circle,draw=black,inner sep=1pt, minimum size = 5mm, ]
        \tikzstyle{mybblack}=[rectangle,rounded corners,draw=black,inner sep=1pt, minimum size = 5mm, align=center,]
        \tikzstyle{myred}=[circle,draw=black,fill=purple!20,inner sep=1pt, minimum size = 5mm]
        \tikzstyle{mygreen}=[circle,draw=black,fill=green!20,inner sep=1pt, minimum size = 5mm]

        \tikzstyle{myredline}=[line width=2mm, purple, draw opacity=.2]
        \tikzstyle{mygreenline}=[line width=2mm, green, draw opacity=.2]

        \path[level distance=10mm] node[myred] (root) {$\varepsilon$}
        child{
          node[mygreen](0) {$f$}
          child{
            node[mybblack](00) {$h$\small{$(x_1)$}}
              child{
                node[myred](000) {$1$}
                child{
                  node (pp) [mygreen] {$f$}
                  child[opacity=0,yshift=0.55cm,xshift=-0.5mm]{
                    node[opacity=1,text = black] {\tiny$\iddots\ddots$}
                  }
                }
                child{
                  node(a)[mygreen] {$a$}
                  child{
                    node(b)[myblack] {$b$}
                  }
                }
              }
              child{
                node[myred](001) {$2$}
                child{
                  node[mygreen] {$f$}
                  child[opacity=0,yshift=0.55cm,xshift=-0.5mm]{
                    node[opacity=1,text = black] {\tiny$\iddots\ddots$}
                  }
                }
                child{
                  node(aa)[mygreen] {$a$}
                  child{
                    node(bb)[myblack] {$b$}
                  }
                }
              }
          }
          child{
            node[mybblack](01) {$h$\small{$(x_2)$}}
            child{
              node[myred](010) {$1$}
              child{
                  node[mygreen] (i) {$f$}
                  child[yshift=-1cm,xshift=-0.5cm]{
                    node[mybblack] {$h$\small{$(x_1)$}}
                    child[opacity=0,yshift=0.55cm,xshift=-0.5mm]{
                      node[opacity=1,text = black] {\tiny$\iddots\ddots$}
                    }
                  }
                  child[yshift=-1cm]{
                    node[mybblack] {$h$\small{$(x_2)$}}
                    child[opacity=0,yshift=0.55cm,xshift=-0.5mm]{
                      node[opacity=1,text = black] {\tiny$\iddots\ddots$}
                    }
                  }
                  child[yshift=-1cm,xshift=0.5cm]{
                    node[myblack] (ii) {$b$}
                    child{
                      node[myred] (j) {$1$}
                      child{
                        node[mygreen] {$f$}
                        child[opacity=0,yshift=0.55cm,xshift=-0.5mm]{
                          node[opacity=1,text = black] {\tiny$\iddots\ddots$}
                        }
                      }
                      child{
                        node(rr)[mygreen] {$a$}
                        child{
                          node(ss)[myblack] {$b$}
                        }
                      }
                    }
                    child{
                      node(p)[myred] {$2$}
                      child{
                        node[mygreen] {$f$}
                        child[opacity=0,yshift=0.55cm,xshift=-0.5mm]{
                          node[opacity=1,text = black] {\tiny$\iddots\ddots$}
                        }
                      }
                      child{
                        node(rrr)[mygreen] {$a$}
                        child{
                          node(sss)[myblack] {$b$}
                        }
                      }
                    }
                  }
                }
                child{
                  node(aaa)[mygreen] {$a$}
                  child{
                    node(bbb)[myblack] {$b$}
                  }
                }
            }
            child{
              node[myred](011) {$2$}
              child{
                  node[mygreen] {$f$}
                  child[opacity=0,yshift=0.55cm,xshift=-0.5mm]{
                    node[opacity=1,text = black] {\tiny$\iddots\ddots$}
                  }
                }
                child{
                  node(aaaa)[mygreen] {$a$}
                  child{
                    node(bbbb)[myblack] {$b$}
                  }
                }
            }
          }
          child{
            node[myblack](02) {$b$}
            child{
              node[myred](020) {$1$}
              child{
                  node[mygreen](q) {$f$}
                  child[opacity=0,yshift=0.55cm,xshift=-0.5mm]{
                    node[opacity=1,text = black] {\tiny$\iddots\ddots$}
                  }
                }
                child{
                  node(c)[mygreen] {$a$}
                  child{
                    node(d)[myblack] {$b$}
                  }
                }
            }
            child{
              node[myred](021) {$2$}
              child{
                  node[mygreen](qq) {$f$}
                  child[opacity=0,yshift=0.55cm,xshift=-0.5mm]{
                    node[opacity=1,text = black] {\tiny$\iddots\ddots$}
                  }
                }
                child{
                  node(cc)[mygreen] {$a$}
                  child{
                    node(dd)[myblack] {$b$}
                  }
                }
            }
          }
        }
        child{
          node[mygreen](1) {$a$}
          child{
            node[myblack](10) {$b$}
          }
        }
        ;

        \draw[myredline]  (root) -- (0);
        \draw[myredline]  (j) -- (rr);
        \draw[myredline]  (011) -- (aaaa);

        \draw[myredline]  (000) -- (pp);
        \draw[myredline]  (001) -- (aa);
        \draw[myredline]  (010) -- (i);
        \draw[myredline]  (p) -- (rrr);
        \draw[myredline]  (020) -- (q);
        \draw[myredline]  (021) -- (qq);

        \draw[mygreenline] (1) -- (10);
        \draw[mygreenline] (0) -- (01);
        \draw[mygreenline] (a) -- (b);
        \draw[mygreenline] (aa) -- (bb);
        \draw[mygreenline] (aaa) -- (bbb);
        \draw[mygreenline] (aaaa) -- (bbbb);
        \draw[mygreenline] (c) -- (d);
        \draw[mygreenline] (cc) -- (dd);
        \draw[mygreenline] (i) -- (ii);
        \draw[mygreenline] (rr) -- (ss);
        \draw[mygreenline] (rrr) -- (sss);

      \end{tikzpicture}
      \caption{${\mathsf H_{\mathcal T,0}}^{\mathsf\frown}{\mathsf t}^{\mathsf\frown}{\mathsf s}$}
      \end{subfigure}
      \begin{subfigure}{0.25\textwidth}
        \begin{tikzpicture}[scale=0.8,baseline=(current bounding box.base)]
          \tikzstyle{level 1}=[sibling distance=12mm]
          \tikzstyle{level 2}=[sibling distance=8mm]
          \tikzstyle{level 3}=[sibling distance=8mm]
          \tikzstyle{level 4}=[sibling distance=8mm]
          \begin{scope}[xshift = -2cm]
          \path[level distance=10mm] node (root){$f$}
          child{
            node(0){$f$}
            child{
                node {$a$}
            }
            child{
              node {$a$}
            }
          }
          child{
            node(1){$a$}
          }
          ;
          \end{scope}
          \path[level distance=10mm] node (roota){$h$}
            child{
              node(a){$b$}
              }

  (roota) edge [bend right = 0, ->, >=stealth, below, color = gray] node {} (root)
  (a) edge [bend right = 0, ->, >=stealth, below, color = gray] node {} (1)
          ;
        \end{tikzpicture}
        \caption{$(t,s,o)$}
        \end{subfigure}

  \end{center}
   \caption{Infinite tree ${\mathsf H_{\mathcal T,0}}$ based on $\mathcal T$ from \cref{ex:zwei}.
   On red nodes \In{} must make a choice, on green nodes \Out{} must make a choice. Their respective strategies $\mathsf{t}$ and $\mathsf{s}$ which define their choices are highlighted on the edges in red and green, respectively. Together, $\mathsf{t}$ and $\mathsf{s}$ encode the input tree $t = f(f(a,a),a)$, the output tree $s = h(b)$ and origin mapping $o\colon \dom{s} \to \dom{t}$ as depicted.
   Note that since $\mathsf t$ and $\mathsf s$ are strategies, choices are made whatever the other player does, that is why in ${\mathsf H_{\mathcal T,0}}^{\mathsf\frown}{\mathsf t}^{\mathsf\frown}{\mathsf s}$, we also have, e.g., a green annotation at node 1112 even though {\In} picked node 1111.
   }
   \label{fig:proofidea}
  \end{figure}

Given $k \geq 0$, and a \tdtt $\mathcal T$, we construct an infinite tree ${\mathsf H_{\mathcal T,k}}$, given as the unfolding of a finite graph ${\mathsf G_{\mathcal T,k}}$, such that a node in this infinite tree represents an input sequence from a finite input tree and an output sequence (where the intuition is that this output sequence was produced while processing this input sequence).
The idea is that in ${\mathsf H_{\mathcal T,k}}$, we define choices (aka.\ strategies) of two so-called players \In{} and \Out{}, where a strategy $\mathsf{t}$ of \In{} together with a strategy $\mathsf{s}$ of \Out{} defines an input tree $t$, an output tree $s$, and an origin mapping $o\colon \dom{s} \to \dom{t}$ of $s$ in $t$.
We can annotate the tree ${\mathsf H_{\mathcal T,k}}$ with the strategies $\mathsf{t}$ and $\mathsf{s}$ which yields a tree ${\mathsf H_{\mathcal T,k}}^{\mathsf\frown}{\mathsf t}^{\mathsf\frown}{\mathsf s}$.
We use this game-like view for all considered decision problems.
We illustrate this view.

\begin{example}\label{ex:proof}
Recall the \tdtt $\mathcal T$ over $\Sigma$ and $\Gamma$ given in \cref{ex:zwei}.
First, we explain how the graph ${\mathsf G_{\mathcal T,0}}$ looks like.
Its unfolding is the infinite tree ${\mathsf H_{\mathcal T,0}}$ (with annotations $\mathsf{t}$ and $\mathsf{s}$) depicted in \cref{fig:proofidea}.
We have three types of nodes: $\{\varepsilon,1,2\}$ to indicate that the current node is the root, a first or a second child. The maximum rank of $\Sigma$ is two, hence $\{\varepsilon,1,2\}$.
These nodes belong to \In{} who can choose the input label, represented by nodes $\{f,a\}$.
Then \Out{} chooses which output (from $T_\Gamma(X)$) should be produced while processing a node.
Since $k=0$, and all right-hand sides of rules in $\mathcal T$ have height at most one, only outputs of height at most one are suitable to maintain origin distance $k=0$.
For input $f$ possible choices are $h(x_1)$ and $h(x_2)$, indicating whether to continue to process the left or the right subtree, or $b$.
For input $a$ only output $b$ is possible.
After the output, edges to $\{1,\cdots,\mathit{rk}(\sigma)\}$ exist, where $\sigma$ is the last seen input letter.
Further explanation is given in \cref{fig:proofidea}.
\end{example}

We present our main technical lemma which states that origin-close transductions are representable as tree language recognizable by a parity tree automaton (a \pta).

\begin{restatable}{lemma}{lemmaregular}\label{lemma:regular}
Given $k \geq 0$ and a \tdtt $\mathcal T$, there exists a \pta that recognizes the tree language $\{ {\mathsf H_{\mathcal T,k}}^{\mathsf\frown}{\mathsf t}^{\mathsf\frown}{\mathsf s} \mid (t,s,o) \in_k R_o(\mathcal T)\}$. 
\end{restatable}

\begin{proof}[Proof sketch.]
  The infinite tree ${\mathsf H_{\mathcal T,k}}^{\mathsf\frown}{\mathsf t}^{\mathsf\frown}{\mathsf s}$ encodes a triple $(t,s,o)$.
  We construct a \pta (which has in fact a safety acceptance condition) that guesses a run of $\mathcal T$ over the input tree $t$ with output tree $s$ that yields an origin mapping $o'$ such that $(t,s,o') \in R_o(\mathcal T)$ and $\mathit{dist}(o(i),o'(i)) \leq k$ which implies that $(t,s,o) \in_k R_o(\mathcal T)$.

  Checking whether $\mathit{dist}(o(i),o'(i)) \leq k$ can be done on-the-fly because the origin distance is bounded which implies that the difference trees of so-far produced output by the guessed run and the productions encoded by the annotations are of bounded size.
  Thus, they can be stored in the state space of the \pta.
  Even tough the construction idea is rather simple, the implementation and correctness proof are non-trivial.
  We face two difficulties.
  Firstly, we have to account for the fact that in $o$ and $o'$ origins for the same output node can lie on different paths of the input tree.
  However, since their distance is bounded, the amount of shared information that the \pta has to check on different paths is also bounded.
  Secondly, it is possible to have non-linear transformation rules (that is, rules with copy, e.g., $q(f(x_1,x_2)) \to f(q_1(x_2),q_2(x_2)))$ which adds another layer of complication.
  This causes that an unbounded number of output nodes can have the same input node as origin.
  We require that ${\mathsf H_{\mathcal T,k}}^{\mathsf\frown}{\mathsf t}^{\mathsf\frown}{\mathsf s}$ is a tree over a ranked alphabet, hence we have to bound the number of output choices that can be made at an input node.
  We show that it suffices to only make a bounded number of output choices for each input node.
  The main insight is that when two continuations of the output tree depend on the same continuation of the input tree, then it suffices to only consider one of them (because the other one can be continued in the same way) if they share the same relevant information where relevant basically means that the state that $\mathcal T$ has reached (guessed by the \pta) at these two output nodes and the output difference trees compared to \Out{}'s choices (given by $\mathsf s$) are the same.
\end{proof}

\mysubparagraph{Solving decision problems for origin-close transducers.}
\label{sec:decision-problems}

We show that deciding $k$-origin inclusion and equivalence for \tdtt{s} reduces to deciding language inclusion for \pta{s}.

\begin{restatable}{proposition}{propkoriginequiv}\label{thm:korginequiv}
Given $k \geq 0$, $k$-origin inclusion and $k$-origin equivalence are decidable for the class \ntdtt.
\end{restatable}

\begin{proof}
Let $\mathcal T_1,\mathcal T_2$ be \tdtt{s} over the same input and output alphabet.
If $\mathcal T_1 \subseteq_k \mathcal T_2$, then $(t,s,o) \in R_o(\mathcal T_1)$ implies that $(t,s,o) \in_k R_o(\mathcal T_1)$ for all $(t,s,o) \in R_o(\mathcal T_1)$.
\cref{lemma:regular} yields that there are \pta{s} $\mathcal A_1,\mathcal A_2$ that recognize 
$\{ {\mathsf H_{\mathcal T_1,0}}^{\mathsf\frown}{\mathsf t}^{\mathsf\frown}{\mathsf s} \mid (t,s,o) \in R_o(\mathcal T)\}$
 and
 $\{ {\mathsf H_{\mathcal T_2,k}}^{\mathsf\frown}{\mathsf t}^{\mathsf\frown}{\mathsf s} \mid (t,s,o) \in_k R_o(\mathcal T_2)\}$,
respectively.
Basically, we want to check that $L(\mathcal A_1) \subseteq L(\mathcal A_2)$.
However, we have to overcome a slight technical difficulty.
If there are trees ${\mathsf H_{\mathcal T_1,0}}^{\mathsf\frown}{\mathsf t_1}^{\mathsf\frown}{\mathsf s_1} \in L(\mathcal A_1)$ and ${\mathsf H_{\mathcal T_2,k}}^{\mathsf\frown}{\mathsf t_2}^{\mathsf\frown}{\mathsf s_2} \in L(\mathcal A_2)$ such that for their encoded triples $(t_1,s_1,o_1)$ and $(t_2,s_2,o_2)$ holds that $t_1 = t_2$, $s_1 = s_2$ and $o_1$ and $o_2$ have an origin difference of at most $k$, i.e., $(t_1,s_1,o_1) \in_k \mathcal R_0(\mathcal T_2)$, it not necessarily holds that ${\mathsf H_{\mathcal T_1,0}}^{\mathsf\frown}{\mathsf t_1}^{\mathsf\frown}{\mathsf s_1} \in L(\mathcal A_2)$.
This is due to the fact that the base trees ${\mathsf H_{\mathcal T_1,0}}$ and ${\mathsf H_{\mathcal T_2,k}}$ look different in general because choices for \Out{} in the first tree are based on the rules of $\mathcal T_1$ and without origin distance and in the latter tree based on the rules of $\mathcal T_2$ with $k$-origin distance.
We only care whether the paths reachable by following the annotations $\mathsf{t_1}$ and $\mathsf{s_1}$ through ${\mathsf H_{\mathcal T_1,0}}$ and the paths reachability by following the annotations $\mathsf{t_2}$ and $\mathsf{s_2}$ through ${\mathsf H_{\mathcal T_2,k}}$ are the same.
Thus, we introduce the operation $\mathit{purge}$ which applied to a tree annotated with strategies of \In{} and \Out{} removes all non-strategy paths.
It is not difficulty to see that the sets 
  $L_1 := \{ \mathit{purge}\left({\mathsf H_{\mathcal T_1,0}}^{\mathsf\frown}{\mathsf t}^{\mathsf\frown}{\mathsf s}\right) \mid (t,s,o) \in R_o(\mathcal T)\}$ and $L_2 := \{ \mathit{purge}\left({\mathsf H_{\mathcal T_2,k}}^{\mathsf\frown}{\mathsf t}^{\mathsf\frown}{\mathsf s}\right) \mid (t,s,o) \in_k R_o(\mathcal T_2)\}$
are also \pta-recognizable.
Hence, in order to check whether $\mathcal T_1 \subseteq_k \mathcal T_2$, we have to check whether $L_1 \subseteq L_2$, which is decidable.
We have shown that $k$-origin inclusion for \tdtt{s} is decidable, consequently, $k$-origin equivalence for \tdtt{s} is decidable for all $k\geq 0$.
\end{proof}

We show that checking whether a \tdtt is $k$-origin \dtdtt-uniformizable reduces to deciding emptiness of \pta{s}.

\begin{restatable}{proposition}{propkoriginlunif}\label{thm:korginunif}
\mbox{Given $k \geq 0$, $k$-origin \dtdtt-uniformization is decidable for the class \ntdtt.}
\end{restatable}
  
\begin{proof}
  \noindent Given a \tdtt $\mathcal T$, by \cref{lemma:regular}, there is a \pta that recognizes 
 
  $\{ {\mathsf H_{\mathcal T,k}}^{\mathsf\frown}{\mathsf t}^{\mathsf\frown}{\mathsf s} \mid (t,s,o) \in_k R_o(\mathcal T)\}.$
  
  \noindent By closure under complementation and intersection, there is a \pta that recognizes

  $\{ {\mathsf H_{\mathcal T,k}}^{\mathsf\frown}{\mathsf t}^{\mathsf\frown}{\mathsf s} \mid (t,s,o) \notin_k R_o(\mathcal T)\}.$
   
  \noindent By closure under projection, there is a \pta that recognizes

  $\{ {\mathsf H_{\mathcal T,k}}^{\mathsf\frown}{\mathsf s} \mid \exists\thinspace \mathsf{t}: (t,s,o) \notin_k R_o(\mathcal T)\}.$
   
  \noindent By closure under complementation and intersection, there is a \pta that recognizes

  $\{ {\mathsf H_{\mathcal T,k}}^{\mathsf\frown}{\mathsf s} \mid \forall\thinspace \mathsf{t}: (t,s,o) \in_k R_o(\mathcal T)\}.$
   
  \noindent By closure under projection, there is a \pta that recognizes

  $\{ {\mathsf H_{\mathcal T,k}} \mid \exists\thinspace \mathsf{s}: \forall\thinspace \mathsf{t}: (t,s,o) \in_k R_o(\mathcal T)\}.$\\

  Let $\mathcal A$ denote the \pta obtained in the last construction step.
  We show that $\mathcal T$ is $k$-origin $\dtdtt$-uniformizable iff $L(\mathcal A) \neq \emptyset$.
We have that $L(\mathcal A) = \{ {\mathsf H_{\mathcal T,k}} \mid \exists\thinspace \mathsf{s}: \forall\thinspace \mathsf{t}: (t,s,o) \in_k R_o(\mathcal T)\}.$
Colloquially, this means that we can fix output choices that only depend on the previously seen input choices, which exactly describes $\dtdtt$-uniformizability.

Assume $\mathcal T$ is $k$-origin $\dtdtt$-uniformizable, say by a $\dtdtt$ $\mathcal T'$.
There exists a strategy of {\Out} in ${\mathsf H_{\mathcal T,k}}$ that copies the computations of $\mathcal T'$.
Clearly, since $\mathcal T'$ is deterministic, we obtain that $\exists\thinspace \mathsf{s}:\forall\thinspace \mathsf{t}: (t,s,o) \in_k R_o(\mathcal T)$, $\mathsf s$ can be chosen to be the strategy that copies $\mathcal T'$.
Thus, $L(\mathcal A) \neq \emptyset$.

For the other direction, assume that $L(\mathcal A) \neq \emptyset$.
This implies that also the set $\{ {\mathsf H_{\mathcal T,k}}^{\mathsf\frown}{\mathsf s} \mid \forall\thinspace \mathsf{t}: (t,s,o) \in_k R_o(\mathcal T)\}$ is non-empty and \pta recognizable.
Since the set is \pta recognizable, it contains a regular infinite tree (meaning the tree has a finite representation).
This tree implicitly contains a finite representation of some strategy $\mathsf s$ such that $\forall\thinspace \mathsf{t}: (t,s,o) \in_k R_o(\mathcal T)$.
Hence, the strategy $\mathsf s$ can be translated into a finite-state \dtdtt that $k$-origin uniformizes~$\mathcal T$.
\end{proof}

Finally, combining \cref{thm:korginequiv,thm:korginunif}, we obtain our main result.

\begin{restatable}{theorem}{thmkorigin}\label{thm:korigin}
Given $k \geq 0$, $k$-origin inclusion, $k$-origin equivalence, and $k$-origin \dtdtt-uniformization are decidable for the class \ntdtt.
\end{restatable}

\section{Conclusion}
\label{sec:conclusion}

We introduced two similarity measures for \tdtt{s} based on their behavior and studied decision problems for similar \tdtt{s}.
For \tdtt{s} with bounded delay, the decision problems remain undecidable.
For origin-close \tdtt{s} they become decidable.
For future work, we plan to consider other tree transducer models.
In \cite{DBLP:journals/iandc/FiliotMRT18}, it was shown that origin inclusion and origin equivalence are decidable for MSO tree transducers and macro tree transducers.



\bibliography{content/biblio}

\appendix

\section{Missing example}

\begin{example}\label{ex:incomparable}
  To begin with, consider the \fst{s} $\mathcal T_1, \mathcal T_2$ depicted in \cref{fig:example-a}, we already explained in the proof of \cref{prop:origin} that $\mathcal T_1 = \mathcal T_2$, but $\mathcal T_1 \neq_{k} \mathcal T_2$ for all $k \geq 0$, i.e., their origin distance is unbounded.
  However, their delay is bounded by 1.
  It is easy to see that $\mathcal T_1 =_{\mathbb D_1} \mathcal T_2$, because their difference in the length of their outputs for the same input is at most one letter.
  Now, consider the \fst{s} $\mathcal T_3, \mathcal T_4$ depicted in \cref{fig:example-b}.
  Both recognize the relation $\{ (ab,c^n) \mid n\in\mathbbm N \}$, hence, $\mathcal T_3 = \mathcal T_4$.
  Clearly, their origin distance is bounded by 1.
  The whole output either has the first or the second letter as origin.
  However, $\mathcal T_3 \neq_{\mathbb D_k} \mathcal T_4$ for all $k \geq 0$, i.e., their delay is unbounded. 
  For any $k$, take the consider the unique runs that admit output $c^{k+1}$ in $\mathcal T_3$ and $\mathcal T_4$, respectively.
  We compare these runs for the input prefix $a$, $\mathcal T_3$, already has produced $c^{k+1}$, and $\mathcal T_4$ no output so far.
  Their delay is $k+1$.
\end{example}

\section{Missing proofs of \texorpdfstring{\cref{prop:origin,lemma:origin-include-delay}}{Propositions 7 and 9}}
\label{app:body}

\proporigin*

\begin{proof}
  Secondly, consider the \fst{s} $\mathcal T, \mathcal T'$ depicted in \cref{fig:example-a}.
  The recognized relation $\mathcal R(\mathcal T) \subseteq \{a,b,A,B\}^* \times \{a,b\}^*$ consists of $\{(ab^nA,ab^n) \mid n \in \mathbbm N \}$ and $\{(ab^nB,ab^m) \mid 0 \leq m \leq 2n-1, n \in \mathbbm N\}$.
  The sequential transducer $\mathcal T'$ recognizes the function $f\colon \{a,b,A,B\}^* \to \{a,b\}^*$ defined by $f(ab^nX) = ab^n$ for $X \in \{A,B\}$ and all $n\in\mathbbm N$.
  Clearly, $\mathcal T'$ is a sequential uniformization of $\mathcal T$.
  However, no sequential uniformization with bounded origin distance exists.
  Towards a contradiction, assume there is sequential transducer $\mathcal T''$ that uniformizes $\mathcal T$ such that $\mathcal T'' \subseteq_k \mathcal T$ for some $k \geq 0$.
  Consider the input word $ab^{2k}A$, in $\mathcal T$ there is only one run with the input which yields the output $ab^{2k}$ and the origin of the $i$th output letter is the $i$th input letter for all $i$.
  Since $\mathcal T'' \subseteq_k \mathcal T$ there exists a run of $\mathcal T''$ on $ab^{2k}A$ that yields $ab^{2k}$ and the origin of the first $b$ in the output is at latest the $k$th $b$ in the input.
  Now, consider the input $ab^{6k}B$, the output of $\mathcal T''$ on $ab^{6k}B$ is $ab^{6k}$.
  Since $\mathcal T''$ is sequential, the the runs of $\mathcal T''$ on $ab^{2k}A$ and $ab^{6k}B$ are the same up to the input $ab^{2k}$, thus, also for the output $ab^{6k}$ the origin first output $b$ is at latest the $k$th $b$ in the input.
  Now we compare this with all possible runs in $\mathcal T$ on $ab^{6k}B$ that also yield $ab^{6k}$.
  Note that $\mathcal T$ (after producing the first $b$) must always produce two $b$ at once, thus in order to produce $ab^{6k}$ for the input $ab^{6k}B$, the production of $b$ can only start after while reading the second half of the input.
  This implies that the first output has an origin in the second half of input which has a distance of more than $k$ (at least $2k$) to the $k$th $b$ in the input. 
\end{proof}

\lemmaoriginincludedelay*

\begin{proof}
Let $\mathcal T_1, \mathcal T_2$ be real-time \fst{s} such that $\mathcal T_1 \subseteq_{i} \mathcal T_2$ for some $i \geq 0$.
Let $\ell$ be the maximum number of output letters that $\mathcal T_1$ produces in a computation step.
Consider any $(u,v,o_1) \in R_o(\mathcal T_1)$, since $\mathcal T_1 \subseteq_{i} \mathcal T_2$, there is $(u,v,o_2) \in R_o(\mathcal T_2)$ such that $\mathit{dist}(o_1(d),o_2(d)) \leq i$ for all $d \in \dom{v}$.
Let $\rho_1$ and $\rho_2$ be the corresponding runs of $\mathcal T_1$ and $\mathcal T_2$, respectively.
We show that $\mathit{delay}(\rho_1,\rho_2) \leq \ell \cdot i$ which implies that $\mathcal T_1 \subseteq_{\mathbb D_{\ell \cdot i}} \mathcal T_2$.
Let $u = a_1\cdots a_n$ and $v = b_1\cdots b_m$.
Pick any prefix of $u$, say $a_1\cdots a_k$, and consider the prefixes of the runs $\rho_1$ and $\rho_2$ such that the input $a_1\cdots a_k$ has been processed.
Let $b_1\cdots b_{k_1}$ and $b_1\cdots b_{k_2}$ be the respective produced outputs.
Wlog., let $k_1 \leq k_2$.
If $k_1 = k_1$, then the output delay for the prefix $a_1\cdots a_k$ is zero.
So assume $k_1 < k_2$.
We have to show that $|b_{k_1+1}\cdots b_{k_2}|$ is less than $\ell \cdot i$.
Since the origin mappings of $\rho_1$ and $\rho_2$, that is, $o_1$ and $o_2$, have a distance of at most $i$, we know that the origin of $b_{k_1+1}\cdots b_{k_2}$ in $\rho_1$ is no later than at the letter $a_{k+i}$.
On $a_{k+1}\cdots a_{k+1}$, $\mathcal T_1$ can produce at most $\ell \cdot i$ output letters.
Consequently,  $|b_{k_1+1}\cdots b_{k_2}| \leq \ell \cdot i$.
\end{proof}

\section{Organization of the appendix}

The remainder of the appendix is devoted to the proof of \cref{lemma:regular}.

We first prove a simpler variant, \cref{lemma:linear}, for linear top-down tree transductions.
We present the formal setup in \cref{app:linear}.
In \cref{app:notations}, we introduce several auxiliary notations and definitions.
Subsequently, we construct the desired \pta in \cref{app:construction} and prove \cref{lemma:linear} in \cref{app:correctness}.
Finally, in \cref{app:fullcase}, we lift the construction and proof to the full set of top-down tree transductions, i.e., we prove \cref{lemma:regular}.

\section{Linear top-down tree transductions}
\label{app:linear}

We fix a linear \tdtt and some values.

\begin{assumption}\label{ass:linear}
  \begin{itemize}
    \item Let $\Sigma,\Gamma$ be ranked alphabets, and let $m$ be the maximal rank of $\Sigma$.
    \item Let $\mathcal T$ be a \emph{linear} {\tdtt} of the form $(Q,\Sigma,\Gamma,q_0,\Delta)$.
    \item Let $M$ be the maximal height of a tree appearing on the right-hand side of a transition rule in $\Delta$.
  \end{itemize}
\end{assumption}

As mentioned in \cref{sec:regular}, given $k \geq 0$, and a \tdtt $\mathcal T$, we construct an infinite tree ${\mathsf H_{\mathcal T,k}}$, given as the unfolding of a finite graph ${\mathsf G_{\mathcal T,k}}$, such that a node in this infinite tree represents an input sequence from a finite input tree and an output sequence (where the intuition is that this output sequence was produced while processing this input sequence).
The idea is that in ${\mathsf H_{\mathcal T,k}}$, we define choices (aka.\ strategies) of two so-called players \In{} and \Out{}, where a strategy $\mathsf{t}$ of \In{} together with a strategy $\mathsf{s}$ of \Out{} defines an input tree $t$, an output tree $s$, and an origin mapping $o\colon \dom{s} \to \dom{t}$ of $s$ in $t$.
We can annotate the tree ${\mathsf H_{\mathcal T,k}}$ with the strategies $\mathsf{t}$ and $\mathsf{s}$ which yields a tree ${\mathsf H_{\mathcal T,k}}^{\mathsf\frown}{\mathsf t}^{\mathsf\frown}{\mathsf s}$.

Recall the given example in \cref{ex:proof} and its visualization depicted in \cref{fig:proofidea}.

\begin{definition}[$\mathcal G_{\mathcal T}^k$]\label{def:graphG}
Given $k \geq 0$, we define the graph $\mathcal G_{\mathcal T}^k$ with edges
 \begin{itemize}
  \item $(\varepsilon \times \SigmaI) \cup (\{1,\dots,m\} \times \SigmaI)$, \hfill {\small Edges of Player \In}
  \item $\bigcup_{i=0}^m (\SigmaIrk{i} \times T_{out,i}$), \hfill {\small Edges of Player \Out}
  
 where $T_{out,i} = \{ t \in T_{\SigmaO}(X_i)\mid t \text{ is \emph{linear} and } h(t) \leq M + 2kM\}$,
  \item $\bigcup_{i=0}^m (T_{out,i} \times \{1,\dots,i\})$, and \hfill {\small Edges to all directions}
  \item $\varepsilon$ is the initial vertex.
 \end{itemize}
\end{definition}

The height of the output choices depends on $\mathcal T$ and the given parameter $k$; their height is bounded by $M + 2kM$, where the summand $M$ covers the case $k=0$.
The intuition behind this is that the output choices made by {\Out} have to mimic the outputs in the transitions of $\mathcal T$ which are of height at most $M$.
Since we are interested in $k$-origin close computations, the output produced by {\Out} can be made up to $k$ computation steps earlier (or later) compared to the computation of $\mathcal T$, meaning the height of the output difference is at most $kM$.
Then, the next choice of {\Out} might be such that the situation reverses.
Hence, output choices with a height of at most $2 \cdot kM$ suffice.

Our desired infinite tree is the following regular tree.

\begin{definition}[${\mathsf H_{\mathcal T,k}}$]
 Given $k \geq 0$, let ${\mathsf H_{\mathcal T,k}}$ be the unfolding of the graph $\mathsf G_{\mathcal T}^k$ with root node $\varepsilon$.
\end{definition}

We need the notion of strategy annotations and resulting encoding.

\begin{definition}[Strategy annotations]
 Given $k \geq 0$ and ${\mathsf H_{\mathcal T,k}}$, let $\mathsf t$ be an encoding of a strategy for \In, let $\mathsf s$ be an encoding of a positional strategy for \Out.
Note that the strategies are positional by default, because ${\mathsf H_{\mathcal T,k}}$ is a tree.
 We denote by ${\mathsf H_{\mathcal T,k}}^{\mathsf\frown}{\mathsf t}^{\mathsf\frown}{\mathsf s}$ the tree ${\mathsf H_{\mathcal T,k}}$ with annotated strategies $\mathsf s$ and $\mathsf t$.
 
 An annotated tree ${\mathsf H_{\mathcal T,k}}^{\mathsf\frown}{\mathsf t}^{\mathsf\frown}{\mathsf s}$ uniquely identifies an input tree $t$, an output tree $s$ and an origin function $o\colon \dom{s} \to \dom{t}$. 
\end{definition}

See \cref{ex:proof,fig:proofidea} for an example of the above definition.

The next lemma is our key lemma for linear top-down tree transductions, it states that the set of origin-close linear top-down tree transductions is a regular tree language recognizable by a parity tree automaton.

\begin{restatable}{lemma}{lemmawichtig}\label{lemma:linear}
  Given $k \geq 0$ and a linear \tdtt $\mathcal T$, there exists a \pta that recognizes the tree language $\{ {\mathsf H_{\mathcal T,k}}^{\mathsf\frown}{\mathsf t}^{\mathsf\frown}{\mathsf s} \mid (t,s,o) \in_k R_o(\mathcal T) \text{ and } o \text{ is a linear transduction}\}$.
\end{restatable}

For ease of presentation, in the next sections, we leave out the subscripts $\mathcal T$ and $k$ and simply write ${\mathsf H}$ instead of ${\mathsf H_{\mathcal T,k}}$.

\section{Notations and definitions}
\label{app:notations}

Towards the construction of our desired parity tree automaton we need further notations and definitions.

\begin{definition}[$k$-neighborhood]
For $k \geq 0$, let $N_{\SigmaI}^k$ be the subset of trees from  $N_{\SigmaI}$ that have at most height $k$.
\end{definition}

\begin{definition}[Annotations]
For a set $A$ and $k \geq 0$, let ${A}^{(k)}$ be the set that is $A$ extended by $a^{(1)},\dots,a^{(k)}$ for each $a \in A$.

For $t \in T_{A}$, let $t^{(j)}$ denote the tree that results from $t$ by replacing each occurrence of some $a \in A$ by $a^{(j)}$ for $j \in \{1,\dots,k\}$.

For $t \in T_{A^{(j)}}$, $\mathit{strip}(t)$ is $t$ without its annotations.

For $f^{(j)} \in A^{(k)}$, $\mathit{ann}(f^{(j)})$ is its annotation $j$.
\end{definition}

The introduced annotations will be used to mark symbols from the output alphabet $\Gamma$ and states $Q$ of $\mathcal T$ that are used as labels in trees.

Now, for the construction of the parity tree automaton, we compare computations of $\mathcal T$ and {\Out} on an input tree (given by the strategy annotations in $\mathsf H$).
Since the origins of these computations can have a distance of at most $k$, we are going to represent the difference in the state space of our parity tree automaton.

Towards this, we introduce two types of output information trees, one for representing the situation where $\mathcal T$ is ahead (see \cref{def:tdtt-info-tree}), and the other one where \Out\ is ahead (see \cref{def:out-info-tree}).

We begin with a representation for the case where $\mathcal T$ is ahead (or on par) compared to {\Out}.

\begin{definition}[Output information tree wrt.\ $\mathcal T$]\label{def:tdtt-info-tree}
 A tree $s \in N_{\SigmaO^{(k)} \cup \SigmaI \cup Q^{(k)} \cup Q}$ is said to be an \emph{output information tree} if there is some $n$ such that $s$ is of the form $C[s_1,\dots,s_n]$, where
 \begin{itemize}
  \item $C \in T_{\Gamma^{(k)}\setminus\Gamma}(X_n)$ is an $n$-context,
  \item $s_1,\dots,s_n \in N_{Q^{(k)} \cup \SigmaI}$ are trees such that for each $s_i$ either $s_i \in Q$ or $s_i = q^{(j)}(t)$ for some $q^{(j)} \in Q^{(k)}$ and some $t \in N_\Sigma$, and
  \item there is at most one such $s_i \in Q$.
 \end{itemize}
\end{definition}

Consider $\Sigma$ given by $\Sigma_2 = \{f\}$ and $\Sigma_0 = \{a,b\}$, and $\Gamma$ given by $\Gamma_2 = \{g\}$, $\Gamma_1 = \{h\}$ and $\Gamma_0 = \{c\}$.
Examples of output information trees over $\Sigma$ and $\Gamma$ are
  
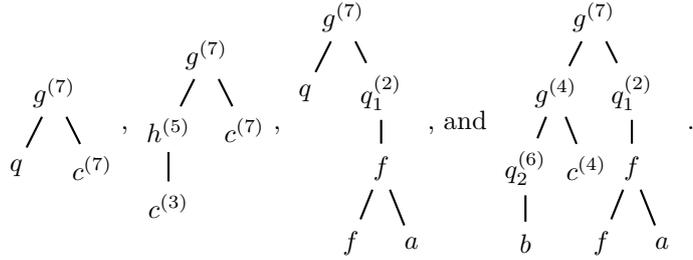
\begin{tikzpicture}[thick,baseline=(current bounding box.center)]
  \tikzstyle{level 1}=[sibling distance=10mm]
   \tikzstyle{level 2}=[sibling distance=8mm]
   \tikzstyle{level 3}=[sibling distance=8mm]

  \path[level distance=10mm,]
  node (root1) at (0,0) {$g^{(7)}$}
   child{
     node(0){$q$}
   }
   child{
       node(1){$c^{(7)}$}
   }
  ;
\end{tikzpicture},
  \begin{tikzpicture}[thick,baseline=(current bounding box.center)]
  \tikzstyle{level 1}=[sibling distance=10mm]
   \tikzstyle{level 2}=[sibling distance=8mm]
   \tikzstyle{level 3}=[sibling distance=8mm]

  \path[level distance=10mm,]
  node (root1) at (0,0) {$g^{(7)}$}
   child{
     node(0){$h^{(5)}$}
      child{
       node(00){$c^{(3)}$}
      }
   }
   child{
       node(1){$c^{(7)}$}
   }
  ;
\end{tikzpicture},
  \begin{tikzpicture}[thick,baseline=(current bounding box.center)]
  \tikzstyle{level 1}=[sibling distance=10mm]
   \tikzstyle{level 2}=[sibling distance=8mm]
   \tikzstyle{level 3}=[sibling distance=8mm]

  \path[level distance=10mm,]
  node (root1) at (0,0) {$g^{(7)}$}
   child{
     node(0){$q$}
   }
   child{
       node(1){$q_1^{(2)}$}
       child{
        node{$f$}
        child{
        node{$f$}
        }
        child{
        node{$a$}
        }
       }
   }
  ;
\end{tikzpicture}, and
  \begin{tikzpicture}[thick,baseline=(current bounding box.center)]
  \tikzstyle{level 1}=[sibling distance=10mm]
   \tikzstyle{level 2}=[sibling distance=8mm]
   \tikzstyle{level 3}=[sibling distance=8mm]

  \path[level distance=10mm,]
  node (root1) at (0,0) {$g^{(7)}$}
   child{
     node(0){$g^{(4)}$}
     child{
      node{$q_2^{(6)}$}
      child{
      node{$b$}
      }
     }
     child{
     node{$c^{(4)}$}
     }
   }
   child{
       node(1){$q_1^{(2)}$}
       child{
        node{$f$}
        child{
        node{$f$}
        }
        child{
        node{$a$}
        }
       }
   }
  ;
\end{tikzpicture}
.

We explain the intuition between (the annotations of) output information trees wrt.\ $\mathcal T$.
The upper part of such a tree with labels from ${\Gamma^{(k)}\setminus\Gamma}$ describes the part of the output that $\mathcal T$ is currently ahead of {\Out}.
The annotation $a$ of an output symbol $g^{(a)} \in \Gamma^{(k)}$ indicates that {\Out} has to produce this output in at most $a$ computation steps, otherwise the bound on the origin distance will be violated.
The occurrence of a $q \in Q$, i.e., a state without annotation, indicates that {\Out} and $\mathcal T$ currently read the same node of the input tree and $\mathcal T$ is in $q$.
Furthermore, the annotation $a$ of a state symbol $p^{(a)} \in Q^{(k)}$ indicates that the current input node of {\Out} has a distance of $a$ to a node of the input tree that $\mathcal T$ processes while in $p$;
it is assumed that these vertices lie on divergent paths.
The subtree below $p^{(a)}$ indicates the (partial) input tree that $\mathcal T$ processes from this divergent vertex.

Now, given such an output information tree our goal is to define an update of the tree according to transitions of $\mathcal T$.
We give two auxiliary definitions first.

For the first auxiliary definition, consider the case that $\mathcal T$ and {\Out} process the same node of the input tree.
We define the application of a transition from $\mathcal T$ assuming $\mathcal T$ is in the state $q$ and the next input symbol is $f$.

An intuition is given in \cref{fig:expansion}.

\begin{figure}[p]
\begin{center}
  \begin{tikzpicture}[thick,baseline=(current bounding box.north)]
  \tikzstyle{level 1}=[sibling distance=16mm]
   \tikzstyle{level 2}=[sibling distance=8mm]
   \tikzstyle{level 3}=[sibling distance=8mm]

 \node[] at (-1.5,0) {$q(f):$};
 
  \path[level distance=10mm,]
  node (root) {$q$}
  child{
  node (root2){$f$}
   }
  ;
  
  \begin{scope}[xshift=5cm]
   \tikzstyle{level 1}=[sibling distance=16mm]
   \tikzstyle{level 2}=[sibling distance=8mm]

 \node[] at (-2,0) {$s \in \mathrm{ext}(q(f)):$};
 
  \path[level distance=10mm,] node (root){$g$}
    child{
      node(0){$q_1$}
        child{
          node(00){$x_1$}
        }
    }
    child{
      node(1){$q_2$}
       child{
        node(10){$x_2$}
       }
    }
  ;  
  \end{scope}
  \end{tikzpicture}
\end{center}
\caption{An example of the extension for $q$ and $f$ wrt.\ the rule $q(f(x_1,x_2)) \to g(q_1(x_2),q_2(x_1))$ according to \cref{def:expansion}.}
\label{fig:expansion}
  \vspace*{\floatsep}
\begin{center}
  \begin{tikzpicture}[thick,baseline=(current bounding box.north)]
  \tikzstyle{level 1}=[sibling distance=16mm]
   \tikzstyle{level 2}=[sibling distance=8mm]
   \tikzstyle{level 3}=[sibling distance=8mm]

 \node[] at (-1.5,0) {$q^{(2)}(t):$};
 
  \path[level distance=10mm,]
  node (root) {$q^{(2)}$}
  child{
  node (root2){$f$}
    child{
      node(0){$f$}
        child{
          node(00){$a$}
        }
        child{
          node(01){$b$}
        }
    }
    child{
      node(1){$a$}
    }
   }
  ;
  
  \begin{scope}[xshift=5cm]
   \tikzstyle{level 1}=[sibling distance=16mm]
   \tikzstyle{level 2}=[sibling distance=8mm]

 \node[] at (-2.5,0) {$s \in \mathrm{ext_k}(q^{(2)}(t)):$};
 
  \path[level distance=10mm,] node (root){$g^{(k-2)}$}
    child{
      node(0){$q_1^{(4)}$}
        child{
          node(00){$a$}
        }
    }
    child{
      node(1){$q_2^{(4)}$}
       child{
        node(10){$f$}
         child{
          node{$a$}
         }
         child{
          node{$b$}
         }
       }
    }
  ;  
  \end{scope}

  \end{tikzpicture}
\end{center}
\caption{An example of the extension for $q^{(2)}$ and $t = f(f(a,b),a)$ with distance update wrt.\ the rule $q(f(x_1,x_2)) \to g(q_1(x_2),q_2(x_1))$ according to \cref{def:expansion-dist}.}
\label{fig:expansion-dist}
  \vspace*{\floatsep}
   
\begin{center}
  \begin{tikzpicture}[thick,baseline=(current bounding box.north)]
  \tikzstyle{level 1}=[sibling distance=16mm]
   \tikzstyle{level 2}=[sibling distance=8mm]
   \tikzstyle{level 3}=[sibling distance=8mm]

 \node[] at (-1.5,0) {$s(f'):$};
 
  \path[level distance=10mm,]
  node (root){$g^{(7)}$}
    child{
      node(0){$q$}
       child{
        node{$f'$}
       }
    }
    child{
      node(1){$q_1^{(4)}$}
        child{
          node(10){$f$}
          child{
            node{$a$}
          }
          child{
            node{$b$}
          }
        }
    }
  ;
  
  \begin{scope}[xshift=5cm]
   \tikzstyle{level 1}=[sibling distance=16mm]
   \tikzstyle{level 2}=[sibling distance=8mm]

 \node[] at (-2,0) {$s' \in \mathrm{EXT}_k(s(f')):$};
 
  \path[level distance=10mm,] node (root){$g^{(7)}$}
    child{
      node(0){$h^{(k)}$}
        child{
          node(00){$q_1$}
           child{
            node{$x_1$}
           }
        }
    }
    child{
      node(1){$g^{(k-4)}$}
       child{
        node{$q_2^{(6)}$}
          child{
           node{$a$}
          }
       }
       child{
        node{$q_3^{(6)}$}
          child{
           node{$b$}
          }
       }
    }
  ;  
  \end{scope}

  \end{tikzpicture}
\end{center}
\caption{An example of an extension for the output information tree $s = f^{(7)}(q,q_1^{(4)}(g(a,b))$ and input symbol $f'$ to \cref{def:expansion-max}.
To compute the extension the rules $q(f'(x_1,x_2)) \to h(q_1(x_1))$ and $q_1(f(x_1,x_2)) \to g(q_2(x_1),q_3(x_2))$ were used.
}
\label{fig:expansion-max}
 \end{figure}



\begin{definition}[Extension]\label{def:expansion}
 For $f \in \SigmaI$ and $q \in Q$, let $\mathrm{ext}(q(f)) \subseteq  N_{\SigmaI \cup \SigmaO \cup Q}(X)$ be the \emph{set of extensions} such that an \emph{extension}
 \begin{center}\mbox{
 $s \in \mathrm{ext}(q(f) :\Leftrightarrow \left \{
 \begin{aligned}
  & \exists\thinspace q\big(f(x_1,\dots,x_i)\big) \rightarrow w[q_1(x_{j_1}),\dots,q_n(x_{j_n})] \in \Delta\\
  & s = w[q_1(x_{j_1}),\dots,q_n(x_{j_n})].
 \end{aligned} \right.$}
 \end{center}
\end{definition}

For the second auxiliary definition, consider the case that $\mathcal T$ and {\Out} process different nodes of the input tree and the distance between these nodes is given.
Assume {\Out} is currently at a node $u$ of the input tree and $\mathcal T$ is at another node $u'$ of the input tree such that $u$ and $u'$ lie on divergent paths.
Recall that the idea is that, given a state $q \in Q^{(k)}$ with annotation $a$, it should denote that $u$ has distance of $a$ to $u'$.
The result of an application of a rule is then defined under the assumption that {\Out} advances to a successor of $u$ and $\mathcal T$ advances to a successor of $u'$ which results in an increase of the distance by two, because $u$ and $u'$ lie on divergent paths.
Since the origin of the extended output is $u'$ with a distance of $a$ to $u$ it is implied that {\Out} has to recover the produced output at most $k-a$ steps later.
We define the application of a transition from $\mathcal T$ assuming $\mathcal T$ is in the state $q$, the distance between the nodes $u$ and $u'$ that $\mathcal T$ and {\Out} read, respectively, is $a$, and the next input(s) for $\mathcal T$ are given by a partial tree $t$.
For an intuition see \cref{fig:expansion-dist}.

\begin{definition}[Extension with distance update]\label{def:expansion-dist}
 For $q^{(a)} \in Q^{(k)}$ and a partial input tree $t \in N_{\SigmaI}$, let $\mathrm{ext_k}(q^{(a)}(t))$ be the \emph{set of extensions with distance update} such that an extension 
 \begin{center}\mbox{
 $s \in \mathrm{ext_k}(q^{(a)}(t))  :\Leftrightarrow \left \{
 \begin{aligned}
  & \exists\thinspace q\big(\val{t}(\varepsilon)(x_1,\dots,x_i)\big) \rightarrow w[q_1(x_{j_1}),\dots,q_n(x_{j_n})] \in \Delta\\
  & s = w^{(k-a)}[q_1^{(a+2)}(t|_{j_1}),\dots,q_n^{(a+2)}(t|_{j_n})].
 \end{aligned}\right.$}
 \end{center}
\end{definition}


Now, we are ready to define an update of an output information tree according to the transitions of $\mathcal T$, called extension.
The extension of an output information tree is obtained by extending all positions where states occur according to \cref{def:expansion}~and ~\cref{def:expansion-dist}.
See \cref{fig:expansion-max} for an intuition.


\begin{definition}[Extension of an OIT wrt.\ $\mathcal T$]\label{def:expansion-max}
 Let $s \in N_{\SigmaO^{(k)} \cup \SigmaI \cup Q^{(k)}}$ be an output information tree wrt.\ $\mathcal T$ and $f \in \SigmaI$.
 By definition of output information trees wrt.\ $\mathcal T$, $s$ contains at most one node whose label is in $Q$.
 
 If $s$ contains a node whose label is in $Q$, say $q$, we let $\mathrm{EXT}_k(s(f)) \subseteq N_{\SigmaO^{(k)} \cup \SigmaI \cup Q^{(k)}}(X)$ be the set such that $s' \in \mathrm{EXT}_k(s(f))$
 
 \begin{center}
 $:\Leftrightarrow$
 \mbox{
 $ \left \{
 \begin{aligned}
  & \exists\thinspace u,u_1,\ldots,u_n \in \dom{s} \text{ with } \val{s}(u)=q,\val{s}(u_1)=q_1^{(a_1)},\ldots,\val{s}(u_n)=q_n^{(a_n)}\\
 & \exists\thinspace s_0\in \mathrm{ext}(q(f)), s_1\in\mathrm{ext_k}\left(q_1^{(a_1)}(s|_{u_11})\right),\ldots,s_n\in\mathrm{ext_k}\left(q_n^{(a_n)}(s|_{u_n1})\right)\\
 & \forall\thinspace u' \in \dom{s}\setminus\{u,u_1,\ldots,u_n\}: \val{s}(u')\notin Q^{(k)}\\
 & s' = s[u \leftarrow s_0,u_1 \leftarrow s_1,\ldots,u_n \leftarrow s_n].
 \end{aligned}\right.$}
\end{center}

Otherwise, $\mathrm{EXT}_k(s(f))$ is similarly defined; the parameter $f$ is ignored and we extend all positions that have labels in $Q^{(k)} \setminus Q$.
\end{definition}

Secondly, we define a representation for the case that {\Out} is ahead (or on par) compared to $\mathcal T$.

\begin{definition}[Output information tree wrt.\ {\Out}]\label{def:out-info-tree}
 Recall $S_{\Gamma^{(k)}\setminus \Gamma}$ are special trees over $T_{\Gamma^{(k)}\setminus \Gamma}$.
 The special tree $s \in S_{\Gamma^{(k)}\setminus \Gamma}$ is said to be an \emph{output information tree}, if removing the $\circ$-labeled node yields a tree in $N_{\Gamma^{(k)}\setminus \Gamma}(X)$, or if $s = \circ$.
 Furthermore, a tree $s \in T_{\Gamma^{(k)}\setminus \Gamma}$ is also an \emph{output information tree}.
\end{definition}

Examples of output information trees wrt.\ {\Out} are 
  \begin{tikzpicture}[thick,baseline=(current bounding box.center)]
  \tikzstyle{level 1}=[sibling distance=10mm]
   \tikzstyle{level 2}=[sibling distance=8mm]
   \tikzstyle{level 3}=[sibling distance=8mm]

 
  \path[level distance=10mm,]
  node (root1) at (0,0) {$g^{(7)}$}
   child{
     node(0){$x_2$}
   }
   child{
       node(1){$\circ$}
   }
  ;
  \end{tikzpicture},
  \begin{tikzpicture}[thick,baseline=(current bounding box.center)]
  \tikzstyle{level 1}=[sibling distance=10mm]
   \tikzstyle{level 2}=[sibling distance=8mm]
   \tikzstyle{level 3}=[sibling distance=8mm]

 
  \path[level distance=10mm,]
  node (root1) at (0,0) {$g^{(4)}$}
   child{
     node(0){$\circ$}
   }
   child{
       node(1){$a^{(2)}$}
   }
  ;
  \end{tikzpicture}, and $b^{(3)}$.
As before, the intuition behind an annotation $a$ is that the other party has to recover the output in at most $a$ computation steps.
Since output information trees wrt.\ {\Out} represent by how much {\Out} is ahead, here, the annotations bound the number of computation steps $\mathcal T$ can use to recover the output.

Also, we define the update of such a tree for a new output choice of {\Out} (which has been annotated).

\begin{definition}[Extension of an OIT wrt.\ {\Out}]\label{def:out-expansion}
  For an output information tree $s \in S_{\Gamma^{(k)} \setminus \Gamma}$ and an annotated (partial) output tree $s' \in N_{\Gamma^{(k)}\setminus \Gamma}(X)$, we define the extension of $s$ by $s'$ as $s\cdot s'$.
  For an output information tree $s \in T_{\Gamma^{(k)}\setminus \Gamma}$, we define the extension of $s$ by an annotated (partial) output tree $s' \in N_{\Gamma^{(k)}\setminus \Gamma}(X)$ to be $s$ (ignoring the parameter $s'$), because $s$ can not be extended as it is already a (completely transformed) tree which only consists of (annotated) output symbols.
\end{definition}

Examples of extended output information trees are depicted in \cref{subfig:synca} and \cref{subfig:syncb} on the left-hand side, e.g., the extended tree 
  \begin{tikzpicture}[thick,baseline=(current bounding box.center)]
  \tikzstyle{level 1}=[sibling distance=10mm]
   \tikzstyle{level 2}=[sibling distance=8mm]
   \tikzstyle{level 3}=[sibling distance=8mm]

  \path[level distance=10mm,]
  node (root1) at (0,0) {$g^{(7)}$}
   child{
     node(0){$h^{(7)}$}
      child{
       node(00){$x_1$}
      }
   }
   child{
       node(1){$x_2$}
   }
  ;
\end{tikzpicture}
could be obtained from concatenating the output information tree
  \begin{tikzpicture}[thick,baseline=(current bounding box.center)]
  \tikzstyle{level 1}=[sibling distance=10mm]
   \tikzstyle{level 2}=[sibling distance=8mm]
   \tikzstyle{level 3}=[sibling distance=8mm]

 
  \path[level distance=10mm,]
  node (root1) at (0,0) {$g^{(7)}$}
   child{
     node(0){$\circ$}
   }
   child{
       node(1){$x_2$}
   }
  ;
  \end{tikzpicture}
with the (annotated partial) output tree
  \begin{tikzpicture}[thick,baseline=(current bounding box.center)]
  \tikzstyle{level 1}=[sibling distance=10mm]
   \tikzstyle{level 2}=[sibling distance=8mm]
   \tikzstyle{level 3}=[sibling distance=8mm]

 
  \path[level distance=10mm,]
     node(0){$h^{(7)}$}
      child{
       node(00){$x_1$}
      }
  ;
  \end{tikzpicture}.
  

Now that we have constructed ways to represent output information trees wrt.\ {\Out} and $\mathcal T$ and how to extend this information wrt.\ a computation step, we need a way to compare extensions of output information trees.
Assume $s$ is an extended output information tree wrt.\ {\Out} and $s'$ is an extended output information tree wrt.\ $\mathcal T$, then we want to ensure that it is either possible to extend $s$ to $s'$ or the other way around.
The function \emph{sync} can be seen as a function that removes the greatest common prefix of the outputs in $s$ and $s'$ from both trees, if one output can be extended to the other output and otherwise fails,
see \cref{fig:sync} for an intuition.

\begin{figure}[t!]
\begin{subfigure}{\textwidth}
\begin{center}
  \begin{tikzpicture}[thick,baseline=(current bounding box.north)]
  \tikzstyle{level 1}=[sibling distance=10mm]
   \tikzstyle{level 2}=[sibling distance=8mm]
   \tikzstyle{level 3}=[sibling distance=8mm]

 \node[xshift=-1cm] at (0,0) {$s:$};
 
  \path[level distance=10mm,]
  node (root1) at (0,0) {$g^{(7)}$}
   child{
     node(0){$h^{(7)}$}
      child{
       node(00){$x_1$}
      }
   }
   child{
       node(1){$x_2$}
   }
  ;
  
  \begin{scope}[xshift=2.5cm]
   \tikzstyle{level 1}=[sibling distance=10mm]
   \tikzstyle{level 2}=[sibling distance=8mm]

 \node[xshift=-1cm] at (0,0) {$s':$};
 
  \path[level distance=10mm,] node(root2) at (0,0) {$g^{(5)}$}
    child{
      node(0){$h^{(5)}$}
        child{
          node(00){$q_1$}
           child{
             node(000){$x_1$}
           }
        }
    }
    child{
      node(1){$f^{(5)}$}
       child{
        node(10){$q_2$}
         child{
          node{$x_2$}
         }
       }
       child{
         node{$b^{(5)}$}
       }
    }
  ;  
  \end{scope}
  
    \begin{scope}[xshift=5.5cm]
   \tikzstyle{level 1}=[sibling distance=10mm]
   \tikzstyle{level 2}=[sibling distance=8mm]

 
 \node[xshift=-1cm] at (0,0) {$x_1 \mapsto$};
  \path[level distance=10mm,] 
  
          node(00a) at (0,0){$q_1$}
           child{
             node(000){$x_1$}
           }
 
  ;
   \end{scope}
   
   \begin{scope}[xshift=7.5cm]
   \node[xshift=-1cm] at (0,0) {$x_2 \mapsto$};
  \path[level distance=10mm,]
      node(1b) at (0,0) {$f^{(5)}$}
       child{
        node(10){$q_2$}
         child{
          node{$x_2$}
         }
       }
       child{
         node{$b^{(5)}$}
       }
  ;

  \end{scope}

  \end{tikzpicture}
  \end{center}
  \caption{Let $s$ resp.\ $s'$ be extended output information trees wrt.\ {\Out} resp.\ $\mathcal T$.
 Consider the result of $sync(s,s')$.
  Following {\Out} in direction 1, {\Out} and $\mathcal T$ are now on par; following {\Out} in direction 2, $\mathcal T$ is now ahead.
  }
  \label{subfig:synca}
\end{subfigure}

\begin{subfigure}{\textwidth}
\vspace{1em}
\begin{center}
  \begin{tikzpicture}[thick,baseline=(current bounding box.north)]
  \tikzstyle{level 1}=[sibling distance=10mm]
   \tikzstyle{level 2}=[sibling distance=8mm]
   \tikzstyle{level 3}=[sibling distance=8mm]

 \node[xshift=-1cm] at (0,0) {$s:$};
 
  \path[level distance=10mm,]
  node (root1) at (0,0) {$g^{(3)}$}
   child{
     node(0){$x_1$}
   }
   child{
       node(1){$f^{(3)}$}
         child{
          node(10){$x_2$}
         }
         child{
          node(11){$x_3$}
         }
   }
  ;
  
  \begin{scope}[xshift=2.5cm]
   \tikzstyle{level 1}=[sibling distance=10mm]
   \tikzstyle{level 2}=[sibling distance=8mm]

 \node[xshift=-1cm] at (0,0) {$s':$};
 
  \path[level distance=10mm,] node(root2) at (0,0) {$g^{(7)}$}
    child{
      node(0){$q_1$}
        child{
         node(00){$x_3$}
        }
      }
    child{
      node(1){$q_2$}
        child{
         node(10){$x_2$}
        }
      }
  ;  
  \end{scope}
  
    \begin{scope}[xshift=5.5cm]
   \tikzstyle{level 1}=[sibling distance=10mm]
   \tikzstyle{level 2}=[sibling distance=8mm]

 
 \node[xshift=-1cm] at (0,0) {$x_1 \mapsto$};
  \path[level distance=10mm,] 
  
          node(00a) at (0,0){$q_1$}
           child{
             node(000){$x_3$}
           }
 
  ;
   \end{scope}
   
   \begin{scope}[xshift=6.7cm]
   
   \path[level distance=10mm,]
      node(1b) at (0,0) {$f^{(3)}$}
       child{
        node(10){$x_2$}
       }
       child{
         node{$x_3$}
       }
  ;  
   \node[xshift=-1cm] at (1.8,0) {$\mapsto$};
  \path[level distance=10mm,] 
  
          node(00a) at (1.5,0){$q_2$}
           child{
             node(000){$x_2$}
           }
 
  ;
  
  \end{scope}

  \end{tikzpicture}
\end{center}
 \caption{Let $s$ resp.\ $s'$ be extended output information tree wrt.\ {\Out} resp.\ $\mathcal T$.
 Consider the result of $sync(s,s')$.
  Following {\Out} in direction 1, {\Out} and $\mathcal T$ are now on par; following {\Out} in direction 2 and 3, {\Out} is now ahead.
 }
 \label{subfig:syncb}
\end{subfigure}
\caption{Two examples of the application of the $\mathit{sync}$ function according to \cref{def:sync}.}
\label{fig:sync}
\end{figure}

\begin{definition}[Synchronization]\label{def:sync}
 Let $s \in T_{\SigmaO^{(k)}}(X)$ and $s' \in N_{\SigmaO^{(k)} \cup \SigmaI \cup Q^{(k)}}(X)$ be extensions of output information trees of {\Out} and $\mathcal T$, respectively.
 The result of $sync(s,s')$ is defined if for all $u \in \dom{s} \cap \dom{s'}$ with $\val{s}(u) \in \Gamma^{(k)}$ and $\val{s'}(u) \in \Gamma^{(k)}$ holds that $strip(\val{s}(u)) = strip(\val{s'}(u))$.
 Otherwise, the result of $sync(s,s')$ is undefined.
 
 If defined, the result is a mapping $\lambda: T_{\SigmaO^{(k)}}(X) \rightarrow N_{\SigmaO^{(k)} \cup \SigmaI \cup Q^{(k)}}(X)$ which is computed as follows.
 We define a partition of $\dom{s}$ into $D_1$ and $D_2$, where $D_1$ contains each node $u$ with $\val{s}(u) \in \Gamma^{(k)}$ and $strip(\val{s}(u)) = strip(\val{s'}(u))$.
 The set $D_1$ are the nodes from the greatest common prefix of $s$ and $s'$.
 From the set $D_2$ select those nodes that do not have a predecessor in $D_2$, say these are the nodes $u_1,\dots,u_n$.
 Note that $u_1,\dots,u_n$ are also part of $\dom{s'}$, because all parent nodes of $u_1,\dots,u_n$ are in $D_1$ which is a subset of $\dom{s} \cap \dom{s'}$ and their labels have the same arity in both $s$ and $s'$.
 We define the result of $sync(s,s')$ to be the function $\lambda$ that maps $s|_{u_i}$ to $s'|_{u_i}$ for each $i$.
\end{definition}

Note that the mapping $\lambda$ is always of the form $x \mapsto s'$ for some $x \in X$ and $s' \in N_{\SigmaO^{(k)} \cup \SigmaI \cup Q^{(k)}}(X)$, or $s \mapsto q(x)$  for some $s\in T_{\SigmaO^{(k)}}(X)$, $q \in Q$ and $x \in X$.
The former case means that $\mathcal T$ is ahead (or on par) compared to \Out\ and the latter case means that $\mathcal T$ is behind.
See the mappings given on the right-hand side in \cref{fig:sync} for examples.

The next two definitions are used in the construction of the parity tree automaton according to \cref{lemma:linear}.
The parity tree automaton we construct is nondeterministic; sometimes, we need to guess an output tree $s$ that we require both \Out\ and $\mathcal T$ to produce in the future with certain distance conditions.
The idea behind the definition is that $\mathcal T$ and {\Out} continue in divergent directions to produce the same output.
The distance of the origins of the output should be inside the given bound $k$.
Thus, we guess how much computation steps $\mathcal T$ and {\Out} take, respectively, to produce the output.
This information is given by the annotations.
If the sum of the annotations (for each output node) is below $k$, the distance of the origins is inside the origin bound $k$.

\begin{definition}[Duplicate with annotations]
 For $k \geq 0$ and $s \in T_{\SigmaO}$, let
  \begin{align*}
 (s_1,s_2) \in \mathrm{DUP}_k(s) \Longleftrightarrow & \ strip(s_1) = strip (s_2) = s, \text{ and}\\
 & \ \forall \thinspace u \in \dom{s} : \mathit{ann}(\val{s_1}(u)) + \mathit{ann}(\val{s_2}(u)) \leq k.
 \end{align*}
\end{definition}

Finally, after a computation step, the number of computation steps that a party can wait to catch up the output decreases by one.

\begin{definition}[Decrease]
 For an annotated tree $s$ let $\mathrm{dec}(s)$ be the tree $s$ where every annotation $(a)$ has been replaced by $(a-1)$ if $a > 1$, otherwise the result is not defined.
\end{definition}

\section{Construction}
\label{app:construction}

Let $\mathcal A$ denote the parity tree automata we construct for the proof of \cref{lemma:linear}.
The idea behind a run of $\mathcal A$ on ${\mathsf H}^{\mathsf\frown}{\mathsf t}^{\mathsf\frown}{\mathsf s}$ (which encodes $(t,s,o\colon \dom{s} \to \dom{t})$) is to guess a run of $\mathcal T$ on $t$ that produces $s$ with origin function $o'\colon \dom{s} \to \dom{t}$, such that the distance bound of $k$ between origins mappings $o$ and $o'$ is respected, and the correctness of the guess is verified while reading ${\mathsf H}^{\mathsf\frown}{\mathsf t}^{\mathsf\frown}{\mathsf s}$.
To achieve this, $\mathcal A$ simulates a run of $\mathcal T$ on $t$ and in each step compares the outputs of both \Out\ and $\mathcal T$.
Furthermore, $\mathcal A$ keeps track of the parts of the output that have not yet been produced by the other party together with distance information (in terms of output annotations) to indicate the number of steps the other party can take to catch up again.
If at some point the outputs are not the same or either \Out\ or $\mathcal T$ fails to catch up, then the tree ${\mathsf H}^{\mathsf\frown}{\mathsf t}^{\mathsf\frown}{\mathsf s}$ is rejected.

\subsection{Parity tree automaton}

We now present the construction of the above described parity tree automaton $\mathcal A$ parameterized by $\mathcal T$ and a given $k \geq 0$.

\subsubsection{States} 
The state set $Q_{\mathcal A}$ consists of sets of states made up from

\begin{center}
\mbox{$\underbrace{\mathcal S}_{\substack{\text{Set of states that store information,}\\\text{implicitly defined further below}}}$ $\cup$ $\underbrace{Q_{O}}_{\text{Validate output component}}$ $\cup$ $\underbrace{Q_{I}}_{\text{Validate input component}}$ $\cup$}\\
\vskip 1em
\noindent$\{q_{acc}\} \cup \{q_{err}\}$,
\end{center}
where $Q_{O} := \{ q_{out = t} \mid t \in T_{\SigmaO^{(k)}} \text{ and } h(t) \leq Mk\}$ and $Q_{I} := \{ q_{in=t}\mid t \in N_{\SigmaI}^k\}$.

\phantom{\quad}

The intuition behind the size of the trees used as states in $Q_{I}$ and $Q_{O}$ is that these states are used when origins diverge.
If origins diverge, say $v$ is their last common ancestor, then their distance grows with each computation step, meaning the knowledge of a $k$-neighborhood of the input tree around $v$ is sufficient to compute the remaining output, thus the height of the trees in $Q_{I}$ is bounded by the given origin bound $k$.
Also, since the remaining output has to be computed in at most $k$ steps, its height can be at most $Mk$, thus the height of the trees in $Q_{O}$ is bounded by $Mk$.

\subsubsection{Priorities}

If a set of states contains $q_{err}$, its priority is $1$, otherwise its priority is $0$.

\subsubsection{The next relation}
\label{app:next}

Before we are able to define the transitions, we need to further describe the state set wrt.\ $\mathcal S$.
The states of $\mathcal S$ store information of the form $(f,o,S_{O},S_{\mathcal T})$, where
\begin{itemize}
 \item $f \in T_{\SigmaIrk{i}}$ is the current input symbol,
 \item $o \in T_{out,i}$ is the current output choice ,
 \item $S_{O} \in S_{\Gamma^{(k)}}$ is the current output information tree wrt.\ {\Out}, and
 \item $S_{\mathcal T} \in N_{\SigmaO^{(k)} \cup \SigmaI \cup Q^{(k)} \cup Q}$ is the current output information tree wrt.\ $\mathcal T$.
\end{itemize}

In the following, we use $/$ as a tuple entry to indicate that there is currently no information.

We define the $\mathrm{next}$ relation that sets a state with information $(f,o,S_{O},S_{\mathcal T})$ in relation to possible successor states which is the heart of the construction.
To ease the understanding, we recall that an annotation at a $\Gamma$-labeled node indicates how much steps a party can take to produce this output, and an annotation at a $Q$-labeled node indicates the distance between two input nodes.

First, we give an informal step-by-step construction how to compute from $(f,o,S_{O},S_{\mathcal T})$ with $f \in \Sigma_i$ successor states $P_1,\dots,P_i$ such that $(f,o,S_{O},S_{\mathcal T}) \rightarrow_{\mathrm{next}} (P_1,\dots,P_i)$.
\begin{itemize}
 \item Extend the stored output $S_O$ of \Out\ by the current output $o$ and set new annotations.
 \item Extend the stored output $S_\mathcal T$ of $\mathcal T$ by choosing an extension $e$ according to the current input symbol $f$ and set new annotations.
 \item Remove the greatest common prefix of $S_O \cdot o$ and $e$ by applying the sync function.
 \item From the result $\lambda$ of sync obtain the information for \Out\ resp.\ $\mathcal T$ that is passed to the children.
 \item In case that \Out\ and $\mathcal T$ follow divergent directions, guess and verify (partial) input trees (resp.\ output trees with annotations) for these divergent directions.
 \cref{ex:next-2} will make this clear.
 \item Update annotations by applying the dec function to indicate that one computation step has been made.
\end{itemize}

It is helpful to consider the following examples before reading the formal definition.

\begin{example}\label{ex:next-1}
 Consider $(f,o,S_{O},S_{\mathcal T})$ and let the extension $S_{O} \cdot o^{(k)}$ be $s$ from \cref{subfig:synca} and the chosen extension of $S_\mathcal{T}$ and $f$ from $\mathrm{EXT}_{(k)}$ be $s'$ from \cref{subfig:synca}.
 The result $\lambda$ of $\mathit{sync}(s,s')$ is also given in the same figure.
 
 Assume that $f$ was binary, we have to compute states $P_1$ and $P_2$ which are passed to the children of the current node in ${\mathsf H}^{\mathsf\frown}{\mathsf t}^{\mathsf\frown}{\mathsf s}$.
 
 Following {\Out} in dir.~1, we see that {\Out} and $\mathcal T$ are on par, and $\mathcal T$ also continues to read in dir.~1., thus we let $P_1 = \{(/,/,S_{O}^1,S_{\mathcal T}^1))\}$ with $S_{O}^1 = \circ$ and $S_{\mathcal T}^1 = dec(q_1) = q_1$. 
 
 Following {\Out} in dir.~2, we see that $\mathcal T$ is ahead, and $\mathcal T$ also continues to read in dir.~2., thus we let $P_2 = \{(/,/,S_{O}^2,S_{\mathcal T}^2))\}$ with $S_{O}^2 = \circ$ and 

 $S_{\mathcal T}^2 = dec\left(
   \begin{tikzpicture}[thick,baseline=(current bounding box.center)]
  \tikzstyle{level 1}=[sibling distance=10mm]
   \tikzstyle{level 2}=[sibling distance=8mm]
   \tikzstyle{level 3}=[sibling distance=8mm]

  \path[level distance=10mm,]
      node(1b) at (0,0) {$f^{(5)}$}
       child{
        node(10){$q_2$}
       }
       child{
         node{$b^{(5)}$}
       }
  ;
  \end{tikzpicture}\right) =$ 
  \begin{tikzpicture}[thick,baseline=(current bounding box.center)]
  \tikzstyle{level 1}=[sibling distance=10mm]
   \tikzstyle{level 2}=[sibling distance=8mm]
   \tikzstyle{level 3}=[sibling distance=8mm]

  \path[level distance=10mm,]
      node(1b) at (0,0) {$f^{(4)}$}
       child{
        node(10){$q_2$}
       }
       child{
         node{$b^{(4)}$}
       }
  ;
  \end{tikzpicture}.
\end{example}

Consider an example that is more involved.

\begin{example}\label{ex:next-2}
 Consider $(f,o,S_{O},S_{\mathcal T})$ and let the extension $S_{O} \cdot o^{(k)}$ be $s$ from \cref{subfig:sync-b} and the chosen extension of $S_\mathcal{T}$ and $f$ from $\mathrm{EXT}_{(k)}$ be $s'$ from \cref{subfig:sync-b}.
 The result $\lambda$ of $\mathit{sync}(s,s')$ is also given in the same figure.
 
 Assume that $f$ was ternary, we have to compute states $P_1$, $P_2$ and $P_3$ which are passed to the children of the current node in ${\mathsf H}^{\mathsf\frown}{\mathsf t}^{\mathsf\frown}{\mathsf s}$.
 
 Following {\Out} in dir.~1, we see that {\Out} and $\mathcal T$ are on par, but $\mathcal T$ continues to read in dir.~3.
 Since the construction is such that we follow {\Out} in dir.~1, we have to guess how the input looks in dir. 3., thus we pick some $t_3 \in N_\Sigma^{k}$.
 As explained before, choosing some tree from $N_\Sigma^{k}$ suffices, because the distances of the origins of the eventually produced outputs will only increase.
 
 Thus, in dir.~1, we go to a state $P_1 = \{(/,/,S_{O}^1,S_{\mathcal T}^1)\}$ with $S_{O}^1 = \circ$ and 
 $S_{\mathcal T}^1 = dec\left( q_1^{(2+1)}(t_3)\right) = q_1^{(2)}(t_3)$, where the annotation indicates that {\Out} and $\mathcal T$ are now on divergent positions with distance 2.
 Furthermore, we have to verify that the guess of $t_3$ was correct, thus the state $P_3$ at dir.~3 contains $q_{in=t_3}$.
 
 Following {\Out} in dir.~2, we see that {\Out} is ahead and that $\mathcal T$ also continues in dir.~2.
 Further we note that {\Out}, while ahead, also conditions from dir.~3.
 This means that $\mathcal T$ eventually produces both outputs that {\Out} chooses to produce from dir.~2 and dir.~3 solely from dir.~2.
 A consequence of this is that the outputs that {\Out} eventually produces from dir.~3 will have divergent origins when produced by $\mathcal T$ from dir.~2.
 Thus, we guess an output tree $s_3 \in T_\Gamma$ with height at most $Mk$, since the origins diverge.
 Since both parties have to produce this output as of yet, we pick some $(s_3^{a},s_3^{b}) \in \mathrm{DUP}_{k}$ .
 Combining all these information, in dir. 2, we go to a state $P_2 = \{(/,/,S_{O}^2,S_{\mathcal T}^2)\}$ with $S_\mathcal T^2 = q_2$ and 
 
 $S_{O}^2 = dec\left(
   \begin{tikzpicture}[thick,baseline=(current bounding box.center)]
  \tikzstyle{level 1}=[sibling distance=10mm]
   \tikzstyle{level 2}=[sibling distance=8mm]
   \tikzstyle{level 3}=[sibling distance=8mm]

   \path[level distance=10mm,]
      node(1b) at (0,0) {$f^{(3)}$}
       child{
        node(10){$\circ$}
       }
       child{
         node{$s_3^{a}$}
       }
  ; 
  \end{tikzpicture}\right) =$ 
  \begin{tikzpicture}[thick,baseline=(current bounding box.center)]
  \tikzstyle{level 1}=[sibling distance=10mm]
   \tikzstyle{level 2}=[sibling distance=8mm]
   \tikzstyle{level 3}=[sibling distance=8mm]

   \path[level distance=10mm,]
      node(1b) at (0,0) {$f^{(2)}$}
       child{
        node(10){$\circ$}
       }
       child{
         node{$dec(s_3^{a})$}
       }
  ; 
  \end{tikzpicture}.
  The annotations in $dec(s_3^{a})$ indicate the number of computations steps that $\mathcal T$ may use to produce $s_3$.
  
  Also, we have to verify that {\Out} indeed produces $s_3$ from dir.~3 according to the annotations from $dec(s_3^b)$.
  Thus, in dir.~3, the state $P_3$ contains $q_{out=dec(s_3^b)}$.
  
  Note that for dir. 3\ we have defined that $P_3$ contains both $q_{in=t_3}$ and $q_{out=dec(s_3^b)}$, so we set $P_3 = \{q_{in=t_3},q_{out=dec(s_3^b)}\}$.
\end{example}

Recall that the result of the sync function, cf.~ \cref{def:sync}, is a mapping $\lambda$ which is always of the form $x \mapsto s'$ for some $x \in X$ and $s' \in N_{\SigmaO^{(k)} \cup \SigmaI \cup Q^{(k)}}(X)$, or $s \mapsto q(x)$ for some $s\in T_{\SigmaO^{(k)}}(X)$, $q \in Q$ and $x \in X$.
The former case means that $\mathcal T$ is ahead (or on par) compared to \Out\ and the latter case means that $\mathcal T$ is behind.
With this in mind, we formally define the \emph{next relation $\rightarrow_{\mathrm{next}}$} in \cref{def:next} on \cpageref{def:next}.
We make the assumption that all operations that were applied in the construction given in \cref{def:next} had defined results.

\begin{figure*}
    \begin{minipage}{\textwidth}
    \begin{definition}[$\rightarrow_{\mathrm{next}}$]\label{def:next}
    \noindent Let $(f,o,S_{O},S_{\mathcal T}) \rightarrow_{\mathrm{next}} (P_1,\dots,P_i)$ for $f \in \Sigma_i$ if 
    \begin{itemize}
     \item there exists an extension $e \in \mathrm{EXT}_{(k)}(S_{\mathcal T}(f))$ of the OIT $S_\mathcal T$ according to the current input symbol $f$,
     \item there exists $\lambda$ as a result of $sync(S_{O} \cdot o^{(k)},e)$, where $S_O \cdot o^{(k)}$ is the extension of the OIT $S_O$ by $o^{(k)}$ according to the current output choice $o$, and
     \item the states $P_1,\dots,P_i$ passed to the children of the current input node are build up from $P_1 = \emptyset,\dots,P_i = \emptyset$ according to $\lambda$ as specified in $(\dagger)$, where $(\dagger) := $ \label{dagger}

    \vskip 0.75em
    $\left\lmoustache
    \begin{aligned}
     &\texttt{\small \color{grayish} $\slash\slash\!\!$ deal with the children where $\mathcal T$ is ahead (or on par) compared to {\Out}}\phantom{\Big(}\\
     &\texttt{\small \color{grayish} $\slash\slash\!\!$ both {\Out} and $\mathcal T$ continue to read at the $j$th child}\\
     & \exists\thinspace j \thinspace \big(\lambda: x_j \mapsto w[q_1(x_{j_1}),\dots,q_m(x_{j_m}),\dots,q_n(x_{j_n})] \text{ with } j_m = j\big) \rightarrow \phantom{\Big(}\\
     &\texttt{\small \color{grayish} $\slash\slash\!\!$ chose and verify input neighborhoods for all other ($\neq j$)}\\
     &\texttt{\small \color{grayish} children that $\mathcal T$ also reads}\\ 
     &\phantom{t}\bigg(\exists\thinspace t_1,\dots,t_{m-1},t_{m+1},\dots,t_n \in N_{\SigmaI}^k : \bigwedge_{\ell \neq m} P_{j_{\ell}} \ni q_{in = t_{\ell}}\\
     &\texttt{\small \color{grayish} $\slash\slash\!\!$ set new OITs for the $j$th child}\\
     &\phantom{text}\wedge P_j \ni (/,/,S_{O}^j,S_{\mathcal T}^j) \text{ with } S_{O}^j = \circ, \text{ and } S_{\mathcal T}^j =\phantom{\Big(}\\
     & \texttt{\small \color{grayish} $\slash\slash\!\!$ at a child only it's own input is followed, use chosen}\\
     &\texttt{\small \color{grayish} input neighborhoods in the OIT for the other children to}\\
     &\texttt{\small \color{grayish} be able to further select $\mathcal T$s computation}\\
     & \phantom{text}\mathrm{dec}\big(w[q_1^{(2+1)}(t_1),\dots,q_{m-1}^{(2+1)}(t_{m-1}),q_m,q_{m+1}^{(2+1)}(t_{m+1}),\dots,q_n^{(2+1)}(t_n)]\big)\bigg)\\
     &\\
     &\texttt{\small \color{grayish} $\slash\slash\!\!$ only {\Out} continues to read at the $j$th child }\quad\texttt{\small \color{grayish} (similar to above)}\\
     & \exists\thinspace j \thinspace \big(\lambda: x_j \mapsto w[q_1(x_{j_1}),\dots,q_n(x_{j_n})] \text{ with  no } j_m = j\big) \rightarrow \\
     &\phantom{t}\bigg(\exists\thinspace t_1,\dots,t_n \in N_{\SigmaI}^k : \bigwedge_{\ell} P_{j_{\ell}} \ni q_{in = t_{\ell}}\thinspace \wedge P_j \ni (/,/,S_{O}^j,S_{\mathcal T}^j)\\
     &\phantom{text}\text{with } S_{O}^j = \circ, \text{ and } S_{\mathcal T}^j = \mathrm{dec}\big(w[q_1^{(2+1)}(t_1),\dots,q_n^{(2+1)}(t_n)]\big)\bigg)\\
     &\\
     &\texttt{\small \color{grayish} $\slash\slash\!\!$ deal with the children where {\Out} is ahead compared to $\mathcal T$}\\
     \end{aligned}
     \right.$
     \vspace{-1mm}
     \flushleft $\left\rmoustache \phantom{text} \right.$
    \end{itemize}
    \hspace{2cm}\textit{Def.\ is continued at \cpageref{def:nextcont}} 
    \end{definition}
    \end{minipage}
    \end{figure*}

    \begin{figure*}\label{def:nextcont}
    \begin{minipage}{\textwidth}
     $\left\rmoustache
     \begin{aligned}
     & \phantom{texttext}\text{ \textit{Def.\ continued from \cpageref{def:next}}}\\
     &\texttt{\small \color{grayish} $\slash\slash\!\!$ deal with the children where {\Out} is ahead compared to $\mathcal T$}\phantom{\Big(}\\
     &\texttt{\small \color{grayish} $\slash\slash\!\!$ {\Out} and $\mathcal T$ continue to read at the same child}\\
     & \exists\thinspace j \thinspace \big(\lambda: w[x_{j_1},\dots,x_{j_n}] \mapsto q(x_j) \text{ with } j_m = j \big) \rightarrow \phantom{\Big(}\\
     &\texttt{\small \color{grayish} $\slash\slash\!\!$ chose output trees for all other ($\neq j$) children that {\Out} also reads}\\
     &\phantom{t}\bigg(\exists\thinspace s_1,\dots,s_{m-1},s_{m+1},\dots,s_n \in T_{\SigmaO} : \bigwedge_{\ell \neq m} h(s_{\ell}) \leq Mk \thinspace\wedge\phantom{\big(}\\
     &\texttt{\small \color{grayish} $\slash\slash\!\!$ both {\Out} and $\mathcal T$ have to produce these, chose annotations that}\\
     &\texttt{\small \color{grayish} indicate how much steps each party can use}\\ 
     &\phantom{text} \exists\thinspace (s_1^a,s_1^b) \in \mathrm{DUP}_{(k)}(s_1),\dots,(s_{m-1}^a,s_{m-1}^b) \in \mathrm{DUP}_{(k)}(s_{m-1}),\\
     &\phantom{texttext}(s_{m+1}^a,s_{m+1}^b) \in \mathrm{DUP}_{(k)}(s_{m+1}),\dots,(s_{n}^a,s_{n}^b) \in \mathrm{DUP}_{(k)}(s_{n}):\\
     &\phantom{texttext}\bigwedge_{\ell \neq m} P_{j_{\ell}} \ni q_{out = \mathrm{dec}(s_{\ell}^b)}\texttt{\small \color{grayish}\  $\slash\slash\!\!$ verify that {\Out} produces the chosen}\\
     &\texttt{\small \color{grayish} output trees (annotations wrt.\ {\Out})}\\
     &\phantom{texttext}\wedge P_{j} \ni (/,/,S_{O}^{j},S_{\mathcal T}^{j})\texttt{\small \color{grayish}\  $\slash\slash\!\!$ set new OITs for the $j_1$th child\ }\phantom{\Big(} \\
     &\texttt{\small \color{grayish} $\slash\slash\!\!$ use chosen output trees (annotations wrt.\ $\mathcal T$) in the OIT to be}\\
     &\texttt{\small \color{grayish} able to further synchronize this with $\mathcal T$s computation}\\ 
     &\phantom{texttextte}\text{with } S_{\mathcal T}^{j} = q \text { and } S_{O}^{j} = \mathrm{dec}\big(w[s_1^a,\dots,s_{m-1}^a,\circ,s_{m+1}^a,s_n^a]\big)\bigg)\\
     &\\
     &\texttt{\small \color{grayish} $\slash\slash\!\!$ {\Out} and $\mathcal T$ share no read direction}\quad\texttt{\small \color{grayish} (similar to above)}\\
    & \exists\thinspace j \thinspace \big(\lambda: w[x_{j_1},\dots,x_{j_n}] \mapsto q(x_j) \text{ with no } j_m = j \\
    &\phantom{t}\bigg(\exists\thinspace s_1,\dots,s_n \in T_{\SigmaO} : \bigwedge_{\ell} h(s_{\ell}) \leq Mk \thinspace\wedge\\
     &\phantom{text} \exists\thinspace (s_1^a,s_1^b) \in \mathrm{DUP}_{(k)}(s_1),\dots,(s_{n}^a,s_{n}^b) \in \mathrm{DUP}_{(k)}(s_{n}):\\
     &\phantom{textte} \bigwedge_{\ell} P_{j_{\ell}} \ni q_{out = \mathrm{dec}(s_{\ell}^b)} \wedge P_{j} \ni (/,/,S_{O}^{j_1},S_{\mathcal T}^{j_1}) \text{ with }\\
     &\phantom{texttextte} S_{\mathcal T}^{j} = q \text{ and } S_{O}^{j} = \mathrm{dec}\big(w[s_1^a,\dots,s_n^a]\big)\bigg)
    \end{aligned}
    \right.$
    \end{minipage}
    \end{figure*}

%

We define the following auxiliary relation.

\begin{definition}[$\mathrm{NEXT}$]
For $S \in \mathcal{S}$, let 
\[\mathrm{NEXT}(S) = \{ (P_1,\dots,P_i) \mid S \rightarrow_{\mathrm{next}} (P_1,\dots,P_i)\} \cup \{(\{q_{err}\})^i\}\]
be the set of possible successor states.
\end{definition}

The idea of the construction for the parity automaton is that the automaton guesses an application of the next relation (which means guessing an application of a rule from $\mathcal T$ and comparing the information with the output choice of {\Out} defined in ${\mathsf H}^{\mathsf\frown}{\mathsf t}^{\mathsf\frown}{\mathsf s}$) such that the origins obtained from computations of $\mathcal T$ and {\Out} never have a distance of more than $k$.

With all definitions in place, we are ready to present the construction of $\Delta_{\mathcal A}$ that implements the idea described in the paragraph before.
The automaton only checks properties on the part of ${\mathsf H}^{\mathsf\frown}{\mathsf t}^{\mathsf\frown}{\mathsf s}$ that are relevant for $t$ and $s$, and collect the needed information along the way.
However, before we do so, we take a closer look at one special case that can occur.

\begin{example}\label{ex:special-case}
In order to compute a pair such that $(f,o,S_{O},S_{\mathcal T}) \rightarrow_{\mathrm{next}} (P_1,\dots,P_i)$, assume that the result of the sync function used to compute this pair contains 
\begin{tikzpicture}[thick,baseline=(current bounding box.center)]
  \tikzstyle{level 1}=[sibling distance=10mm]
   \tikzstyle{level 2}=[sibling distance=8mm]
   \tikzstyle{level 3}=[sibling distance=8mm]
   \path[level distance=10mm,]
      node(1b) at (0,0) {$g^{(6)}$}
       child{
        node(10){$x_2$}
       }
       child{
         node{$c^{(3)}$}
       }
  ;  
   \node[xshift=-1cm] at (1.8,0) {$\mapsto$};
  \path[level distance=10mm,] 
  
          node(00a) at (1.5,0){$q$}
           child{
             node(000){$x_1$}
           }
 
  ;
  \end{tikzpicture}.
  We see that {\Out} is ahead compared to $\mathcal T$, and that they do not share a read direction.
  The construction given in \cref{def:next} for last case is applied.
  That means we follow $\mathcal T$ in dir.~1 and guess an output tree that {\Out} has to produce from dir.~2 and verify that both {\Out} and $\mathcal T$ will eventually produce the chosen output tree.
  Therefore, an output tree is chosen, for example $b \in T_\Gamma$, and a pair from $\mathrm{DUP}_{(k)}(b)$, for example $(b^{(2)},b^{(3)})$.
  Then, in dir.~2, the state $P_2$ that contains $q_{out=dec(b^{(3)})} = q_{out=b^{(2)}}$ is reached, and in dir.~1 the state $P_1$ that contains $(/,/,S_O^1,S_\mathcal T^1)$ is reached with $S_\mathcal T^1 = q$ and $S_O^1 = dec\left ( 
  \begin{tikzpicture}[thick,baseline=(current bounding box.center)]
  \tikzstyle{level 1}=[sibling distance=10mm]
   \tikzstyle{level 2}=[sibling distance=8mm]
   \tikzstyle{level 3}=[sibling distance=8mm]
   \path[level distance=10mm,]
      node(1b) at (0,0) {$g^{(6)}$}
       child{
        node(10){$b^{(2)}$}
       }
       child{
         node{$c^{(3)}$}
       }
  ;  
  \end{tikzpicture}
  \right) =$
    \begin{tikzpicture}[thick,baseline=(current bounding box.center)]
  \tikzstyle{level 1}=[sibling distance=10mm]
   \tikzstyle{level 2}=[sibling distance=8mm]
   \tikzstyle{level 3}=[sibling distance=8mm]
   \path[level distance=10mm,]
      node(1b) at (0,0) {$g^{(5)}$}
       child{
        node(10){$b^{(1)}$}
       }
       child{
         node{$c^{(2)}$}
       }
  ;  
  \end{tikzpicture}.
  
 According to the construction of the transitions of $\mathcal A$ presented below, the following can happen.
 Recall that ${\mathsf H}^{\mathsf\frown}{\mathsf t}^{\mathsf\frown}{\mathsf s}$ is identified with some input tree $t$, output tree $s$, and corresponding origin mapping $o$.
 Consider a run of $\mathcal A$ on ${\mathsf H}^{\mathsf\frown}{\mathsf t}^{\mathsf\frown}{\mathsf s}$, and assume that a node will be reached that corresponds to some $v \in \dom{t}$ and $\mathcal A$ is in a state that contains $(f,o,S_{O},S_{\mathcal T})$.
 Assume that in order to continue the run, the successor states $P_1,\dots,P_i$ from this example are used.
 From this node, following the strategy annotations of {\In}, $\mathcal A$ will reach a node in ${\mathsf H}^{\mathsf\frown}{\mathsf t}^{\mathsf\frown}{\mathsf s}$ that corresponds to $v1 \in \dom{t}$.
 This node will have a strategy annotation for {\Out}, however this strategy annotation has no relevance for the the output tree $s$, because only strategy annotations for {\Out} that are reachable from nodes in ${\mathsf H}^{\mathsf\frown}{\mathsf t}^{\mathsf\frown}{\mathsf s}$ that corresponds to (descendants of) $v2 \in \dom{t}$ contribute further to $s$.
 The described node that corresponds to $v1 \in \dom{t}$ will be reached with a state that contains some $(f',o',S_{O}^1,S_{\mathcal T}^1)$ with $S_{O}^1$ and $S_{\mathcal T}^1$ as above and $o'$ is the new strategy annotation for {\Out}.
 To continue the run of $\mathcal A$, we chose successor states according to the next relation. To compute these, as formally defined, $S_{O}^1 \cdot o'^{(k)}$ is used.
 Recall that $S_{O}^1 \cdot o'^{(k)}$ is $S_O^1$, because $S_O^1 = g^{(5)}(b^{(1)},c^{(2)}) \in T_{\Gamma^{(k)}}$.
 We see that the irrelevant strategy annotation plays no role in the computation of the successor states.
 Since this case is rather hidden below, as it is covered as a special case of the statements in the upcoming correctness proof, we explicitly showed this example here.
 
 The situation that $\mathcal A$ reaches a node of ${\mathsf H}^{\mathsf\frown}{\mathsf t}^{\mathsf\frown}{\mathsf s}$ (with a meaningful state), in which the strategy annotation of {\Out} is irrelevant, happens because for the special case where {\Out} is ahead compared to $\mathcal T$ and they do not share a read direction, we follow $\mathcal T$ (instead of {\Out} which we follow otherwise).
 This is only done because it simplifies the correctness proof of the following construction in this section.
\end{example}
 
\subsubsection{Transitions}

As we have seen in \cref{ex:next-2}, it can happen that more than one state is passed to a direction, then we follow each state as described below in the construction of $\Delta_{\mathcal A}$ and implicitly assume that we go into a set of states.

Recall that the parity tree automaton reads nodes of ${\mathsf H}^{\mathsf\frown}{\mathsf t}^{\mathsf\frown}{\mathsf s}$.
Hence, as alphabet of the parity tree automaton, we use triples.
The first entry is the label from $\mathsf H$, the second and third entry are the strategy annotations for {\In} and {\Out}, respectively.
We assume the strategy annotations are the number of the child that is chosen.
We use $\#$ in the second resp.\ third entry to indicate that there is no annotation for either {\In} or {\Out}.

\paragraph*{Construction of Delta wrt.\ stored output information trees}

\begin{itemize}
 \item We define $q_{in} := (/,/,\circ,q_0)$ to be the initial state.
 
 \item For $(/,/,S_O,S_{\mathcal T}) \in \mathcal S$, $f \in \SigmaIrk{i}$ and $j \in \{1,\dots,|T_{out,i}|\}$, we add
 \[\bigg((/,/,S_O,S_{\mathcal T}),(f,\#,j),q_{acc},\dots,q_{acc},\underbrace{(f,/,S_O,S_{\mathcal T})}_{j\text{th child}},q_{acc},\dots,\underbrace{q_{acc}}_{|T_{out,i}| \text{th child}}\bigg) \text{ to } \Delta_{\mathcal A}.\]
  
 \quad\emph{We collect the current input symbol and follow the choice of {\Out}.}
 
 \item For $(f,/,S_O,S_{\mathcal T}) \in S$, $f \in \SigmaIrk{i}$ and $o \in T_{out,i}$, we add \[\bigg((f,/,S_O,S_{\mathcal T}),(o,\#,\#),p_1,\dots,p_i\bigg) \text{ to } \Delta_{\mathcal A} \text{ for each } (p_1,\dots,p_i) \in \mathrm{NEXT}(f,o,S_O,S_{\mathcal T}).\]
 
 \quad\emph{We collect the current output choice and compute the successor states according to the next relation which guesses an application of rule from $\mathcal T$ as defined above.}
 
 \item For $(/,/,S_O,S_{\mathcal T}) \in \mathcal S$, $\ell \in (\{\varepsilon\} \cup \{1,\dots,m\})$ and $j \in \{1,\dots,|\SigmaI|\}$, we add
  \[\bigg((/,/,S_O,S_{\mathcal T}),(\ell,j,\#),q_{acc},\dots,q_{acc},\underbrace{(/,/,S_O,S_{\mathcal T})}_{j\text{th child}},q_{acc},\dots,\underbrace{q_{acc}}_{|\SigmaI|\text{th child}}\bigg) \text{ to } \Delta_{\mathcal A}.\]
  
  \quad\emph{We follow the choice of {\In}.}
\end{itemize}

For ease of presentation, we speak of collecting (input and output) information along the tree and choosing successor states as ``one step'' of $\mathcal A_k$.

\phantom{\quad}

The next two parts are the constructions used to verify the input resp.\ the output guesses.

\paragraph*{Construction of Delta wrt.\ input guesses}

\begin{itemize}
 \item For $q_{in=t} \in Q_I$, $f \in \SigmaIrk{i}$ and $j \in \{1,\dots,|T_{out,i}|\}$, we add
 \[\bigg(q_{in=t},(f,j,\#),q_{acc},\dots,q_{acc},\underbrace{q_{in=t}}_{j\text{th child}},q_{acc},\dots,\underbrace{q_{acc}}_{|T_{out,i}| \text{th child}}\bigg) \text{ to } \Delta_{\mathcal A} \text{ if } \val{t}(\varepsilon) = f.\] 
 
 \quad\emph{We check the correctness of (the root node of) the guess.}
 
   \phantom{\quad}
 
  \item For $q_{in=f} \in Q_I$, $f \in \SigmaIrk{i}$ and $j \in \{1,\dots,|T_{out,i}|\}$, we add
 \[\bigg(q_{in=f},(f,\#,j),q_{acc},\dots,\underbrace{q_{acc}}_{|T_{out,i}| \text{th child}}\bigg) \text{ to } \Delta_{\mathcal A}.\]
 
 \quad\emph{We follow the choice of {\Out}.}

  \phantom{\quad}
 
 \item For $q_{in=t} \in Q_I$ with $\val{t}(\varepsilon) \in \SigmaIrk{i}$ and $o \in T_{out,i}$, we add
  \[\bigg(q_{in=t},(o,\#,\#),q_{in=t|_1},\dots,q_{in=t|_i}\bigg) \text{ to } \Delta_{\mathcal A}.\]
  
  \quad\emph{We pass on the remainder of the guess.}

   \phantom{\quad}
 
 \item For $q_{in=t} \in Q_I$, $\ell \in (\{\varepsilon\} \cup \{1,\dots,m\})$ and $j \in \{1,\dots,|\SigmaI|\}$, we add
  \[\bigg( q_{in=t},(\ell,j,\#),q_{acc},\dots,q_{acc},\underbrace{ q_{in=t}}_{j\text{th child}},q_{acc},\dots,\underbrace{q_{acc}}_{|\SigmaI|\text{th child}}\bigg) \text{ to } \Delta_{\mathcal A}.\]
  
  \quad\emph{We follow the choice of {\In}.}
\end{itemize}

\paragraph*{Construction of Delta wrt.\ output guesses}

\begin{itemize}
 \item For $q_{out=t} \in Q_O$, $f \in \SigmaIrk{i}$ and $j \in \{1,\dots,|T_{out,i}|\}$, we add
 \[\bigg(q_{out=t},(f,\#,j),q_{acc},\dots,q_{acc},\underbrace{q_{out=t}}_{j\text{th child}},q_{acc},\dots,\underbrace{q_{acc}}_{|T_{out,i}| \text{th child}}\bigg) \text{ to } \Delta_{\mathcal A}.\]
 
 \quad\emph{We follow the choice of {\Out}.}

 \phantom{\quad}
 
 \item For $q_{out=t} \in Q_O$ and $o \in T_{out,i}$ such that $sync(o,t)$ is defined and $dom(\lambda) \neq \emptyset$, we add
 \[\bigg(q_{out=t},(o,\#,\#),q_1,\dots,q_i\bigg) \text{ to } \Delta_{\mathcal A}, \text{ where } q_j \text{ is } q_{out=dec(w)}, \]
 if $\lambda: w \mapsto x_j$ and $dec(w)$ is valid, otherwise $q_j = q_{err}$ for all $1 \leq j \leq i.$
 
\phantom{\quad}

 \quad\emph{We check the correctness of a part of the guess.}

 \phantom{\quad}
 
 \item For $q_{out=t} \in Q_O$ and $o \in T_{out,i}$ such that $sync(o,t)$ is defined and $dom(\lambda) = \emptyset$, we add
 \[\bigg(q_{out=t},(o,\#,\#),q_{acc},\dots,\underbrace{q_{acc}}_{i\text{th child}}\bigg) \text{ to } \Delta_{\mathcal A}.\]
 
 \quad\emph{We have successfully verified the guess.}

 \phantom{\quad}
 
 \item For $q_{out=t} \in Q_O$, $\ell \in (\{\varepsilon\} \cup \{1,\dots,m\})$ and $j \in \{1,\dots,|\SigmaI|\}$, we add
  \[\bigg( q_{out=t},(\ell,j,\#),q_{acc},\dots,q_{acc},\underbrace{ q_{out=t}}_{j\text{th child}},q_{acc},\dots,\underbrace{q_{acc}}_{|\SigmaI|\text{th child}}\bigg) \text{ to } \Delta_{\mathcal A}.\]
  
  \quad\emph{We follow the choice of {\In}.}

  \phantom{\quad}
\end{itemize}

\section{Correctness of the construction}
\label{app:correctness}

We prove our key lemma for linear top-down tree transductions, that is, \cref{lemma:linear}.
First, we prove the direction from left to right.

\begin{figure}[p]
\begin{subfigure}{\textwidth}
\begin{center}
  \begin{tikzpicture}[thick,baseline=(current bounding box.north),scale=1]
  
  \tikzstyle{zigzag}=[decoration = {zigzag,segment length = 5mm, amplitude = 0.5mm},decorate]
  
   \node (r0i) at ( -5.0,  0.0) {}; 
   \node (s0i) at (-7.0, -5.0) {}; 
   \node (s1i) at ( -3.0, -5.0) {}; 
   \node (sii) at ( -5.5, -3.0) {}; 
   
   \node at ($(r0i) + (-1.5,0)$) {\small input tree $t$}; 
   
   \node (p) at ( -5.4, -2) {}; 
   \node (u2) at ( -5, -3.0) {}; 

   \node (r0) at ( 0.0,  0.0) {}; 
   \node (s0) at (-2.0, -5.0) {}; 
   \node (s1) at ( 2.0, -5.0) {}; 
   \node (si) at ( 0.3, -2.0) {}; 
   
   \node at ($(si) + (-0.15,-0.7)$) {\tiny\color{black} $S_\mathcal T$}; 
   
   \node at ($(r0) + (-1.5,0)$) {\small output tree $s$}; 
   \node at ($(r0) + (-2,-0.3)$) {\scriptsize (output tree $s_\rho$ embedded)}; 
   
   \node (l) at ($(r0)!0.6!(s0)$) {};
   \node (r) at ($(r0)!0.6!(s1)$) {};
   
   \node (li) at ($(si) + (-0.5,-1)$) {};
   \node (ri) at ($(si) + ( 0.7,-1)$) {};
   \node (c) at ($(li)!0.4!(ri)$) {};
   
   \node (c2) at ($(li)!0.8!(ri)$) {};

   \path[draw] (r0.center)--(s0.center);
   \path[draw] (s0.center)--(s1.center);
   \path[draw] (s1.center)--(r0.center);
   
   \path[draw] (r0i.center)--(s0i.center);
   \path[draw] (s0i.center)--(s1i.center);
   \path[draw] (s1i.center)--(r0i.center);

   \path[draw,line width=0.5pt] (si.center)--(li.center);
   \path[draw,line width=0.5pt] (si.center)--(ri.center);
  
   \draw[color=black, line width=0.5pt, zigzag]
   (r0.center) to [out=-95, in=90] (si.center);

   \draw[color=black, line width=0.5pt, zigzag]
   (si.center) to [out=-95, in=90] (c.center);
   
   \draw[color=black, line width=0.5pt, zigzag]
   (si.center) to [out=-95, in=90] (c2.center);
   
   \draw[color=black, line width=0.5pt, zigzag]
   (r0i.center) to [out=-95, in=90] (sii.center);
   
      \draw[color=black, line width=0.5pt, zigzag]
   (p.center) to [out=-95, in=90] (u2.center);
   

  \path (c) edge [bend left=25, ->, >=stealth, above, color = gray,line width=0.5pt] node (text) {} (sii);
   \node[yshift=0.8cm] at (text) {\scriptsize\color{black}for the unique $x$};
   \node[yshift=0.55cm] at (text) {\scriptsize\color{black}with $\val{S_\mathcal T}(x) \in Q$:};
   \node[yshift=0.25cm] at (text) {\scriptsize{\color{black}\ \ $\varphi(vx) = u$ in $\rho$}};

  \path (c2) edge [bend left=40, ->, >=stealth, below, near start, color = gray,line width=0.5pt] node (text2) {} (u2);
   \node[xshift=0.3cm,yshift=-0.25cm] at (text2) {\scriptsize{\color{black}$\varphi(v\bar x) = \bar u$ in $\rho$ for $\bar x$}};
 \node[xshift=0.5cm,yshift=-0.5cm] at (text2) {\scriptsize{\color{black} with $\val{S_\mathcal T}(\bar x) \in Q^{(k)}\setminus Q$}};  
   
   \draw[color=black, line width=0.5pt]
   (l.center) to [out=340, in=190] (li.center) to [out=10, in=180] (c.center) to [out=0, in=190] (ri.center) to [out=10, in=200] (r.center);


    \draw[color=black!20!white, fill=black!20!white, line width=0.5pt]   (sii) circle (.15);
    \node[] at (sii) {\footnotesize\color{black}$u$};
    
    \draw[color=black!20!white, fill=black!20!white, line width=0.5pt]   (u2) circle (.15);
    \node[] at (u2) {\footnotesize\color{black}$\bar u$};
    
   \draw[color=black!20!white, fill=black!20!white, line width=0.5pt]   (p) circle (.15);

    \draw[color=black!20!white, fill=black!20!white, line width=0.5pt]   (si) circle (.15);
    \node[] at (si) {\footnotesize\color{black}$v$};

    \draw[color=black!20!white, fill=black!20!white, line width=0.5pt]   (c) circle (.15);
    \node[] at (c) {\footnotesize\color{black}$vx$};

    \draw[color=black!20!white, fill=black!20!white, line width=0.5pt]   (c2) circle (.15);
    \node[] at (c2) {\footnotesize\color{black}$v\bar x$};
   

%


   
   \begin{scope}[on background layer]
    \fill[green!20!white,on background layer] (r0) -- (l.center) to [out=340, in=190] (li.center) to [out=10, in=180] (c.center) to [out=0, in=190] (ri.center) to [out=10, in=200] (r.center) -- (r0.center) -- cycle;
   \end{scope}

\end{tikzpicture}
  \end{center}
  \caption{Situation where $S_O$ is $\circ$ and $S_\mathcal T$ is of the form $w[q_1^{(a_1)}(t_1),\dots,q,\dots,q_n^{(a_n)}(t_n)]$, where $w$ is a context over $\Gamma^{(k)}\setminus \Gamma$, $q \in Q$, $q_i^{(a_i)} \in Q^{(k)}\setminus Q$, and $t_i \in N_\Sigma$ is an input neighborhood for all $i$. This indicates that $\mathcal T$ is ahead compared to {\Out}.
  The tree $strip(w)[q_1,\dots,q,\dots,q_n]$ is located at $v$ in $s_\rho$.
  For the $x \in \dom{S_\mathcal T}$ with $\val{S_\mathcal T}(x) = q$ holds $\varphi(vx) = u$, for an $\bar x \in \dom{S_\mathcal T}$ with $\val{S_\mathcal T}(\bar x) \in Q^{(k)}\setminus Q$, say $q_i^{(a_i)}$, and $\varphi(v\bar x) = \bar u$ holds that the distance between $u$ and $\bar u$ is $a_i$ and $t_i$ is located at $\bar u$ in $t$.}
  \label{subfig:sync-a}
\end{subfigure}

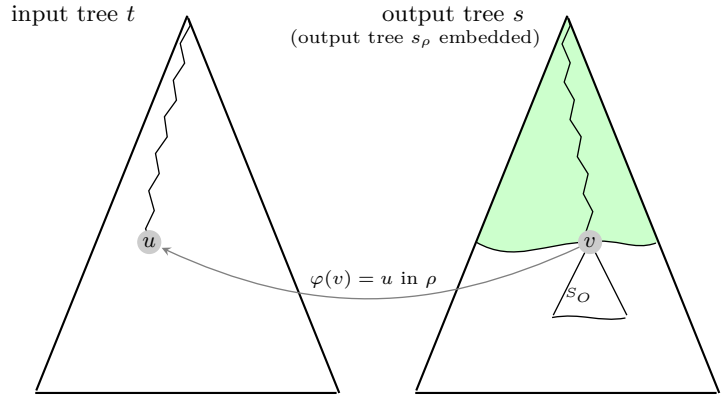
\begin{subfigure}{\textwidth}
\vspace{1em}
\begin{center}
   \begin{tikzpicture}[thick,baseline=(current bounding box.north),scale=1]
  
  \tikzstyle{zigzag}=[decoration = {zigzag,segment length = 5mm, amplitude = 0.5mm},decorate]
  
   \node (r0i) at ( -5.0,  0.0) {}; 
   \node (s0i) at (-7.0, -5.0) {}; 
   \node (s1i) at ( -3.0, -5.0) {}; 
   \node (sii) at ( -5.5, -3.0) {}; 
   
   \node at ($(r0i) + (-1.5,0)$) {\small input tree $t$}; 

   \node (r0) at ( 0.0,  0.0) {}; 
   \node (s0) at (-2.0, -5.0) {}; 
   \node (s1) at ( 2.0, -5.0) {}; 
   \node (si) at ( 0.3, -3.0) {}; 
   
   \node at ($(si) + (-0.15,-0.7)$) {\tiny\color{black} $S_O$}; 
   
   \node at ($(r0) + (-1.5,0)$) {\small output tree $s$}; 
   \node at ($(r0) + (-2,-0.3)$) {\scriptsize (output tree $s_\rho$ embedded)}; 
   
   \node (l) at ($(r0)!0.6!(s0)$) {};
   \node (r) at ($(r0)!0.6!(s1)$) {};
   
   \node (li) at ($(si) + (-0.5,-1)$) {};
   \node (ri) at ($(si) + ( 0.5,-1)$) {};
   \node (c) at ($(li)!0.5!(ri)$) {};

   \path[draw] (r0.center)--(s0.center);
   \path[draw] (s0.center)--(s1.center);
   \path[draw] (s1.center)--(r0.center);
   
   \path[draw] (r0i.center)--(s0i.center);
   \path[draw] (s0i.center)--(s1i.center);
   \path[draw] (s1i.center)--(r0i.center);

   \path[draw,line width=0.5pt] (si.center)--(li.center);
   \path[draw,line width=0.5pt] (si.center)--(ri.center);
  
   \draw[color=black, line width=0.5pt, zigzag]
   (r0.center) to [out=-95, in=90] (si.center);

   \draw[color=black, line width=0.5pt, zigzag]
   (r0i.center) to [out=-95, in=90] (sii.center);
   
  \path (si) edge [bend left=25, ->, >=stealth, above, color = gray,line width=0.5pt] node (text) {\scriptsize{\color{black}\ $\varphi(v) = u$ in $\rho$}} (sii);
   
   \draw[color=black, line width=0.5pt]
   (l.center) to [out=330, in=180] (si.center) to [out=10, in=200] (r.center);
   
   \draw[color=black, line width=0.5pt]
   (li.center) to [out=10, in=190] (ri.center);


    \draw[color=black!20!white, fill=black!20!white, line width=0.5pt]   (sii) circle (.15);
    \node[] at (sii) {\footnotesize\color{black}$u$};

    \draw[color=black!20!white, fill=black!20!white, line width=0.5pt]   (si) circle (.15);
    \node[] at (si) {\footnotesize\color{black}$v$};

   \begin{scope}[on background layer]
    \fill[green!20!white,on background layer] (r0) -- (l.center) to [out=330, in=180] (si.center) to [out=10, in=200] (r.center) -- (r0.center) -- cycle;
   \end{scope}

\end{tikzpicture}
\end{center}
 \caption{Situation where $S_O$ is a (special) tree (which is not $\circ$) over $\Gamma^{(k)}$ and $S_\mathcal T \in Q$. This indicates that {\Out} is ahead compared to $\mathcal T$.
 The tree $S_O$, without annotations and without $\circ$, is located at $v$ in $s$. The tree $S_\mathcal T$ is located in $s_\rho$ at $v$.
 }
 \label{subfig:sync-b}
\end{subfigure}
\caption{Visualization of the proof of \cref{lemma:proof-left-to-right}.
In the run of $\mathcal A$ on ${\mathsf H}^{\mathsf\frown}{\mathsf t}^{\mathsf\frown}{\mathsf s}$, a node corresponding to some $u \in \dom{t}$ with information $(f,o,S_O,S_\mathcal T)$ is reached.
There is a (partial) run $\rho$ of $\mathcal T$ on $t$ with (partial output tree) $s_\rho$ such that $s_\rho$ (partially) defines the output tree $s$ (this is drawn as the green part of $s$) and there exists a unique $v \in \dom{s}$ such that the current configuration of $\mathcal T$ satisfies \cref{c:1,c:2,c:3,c:4,c:5,c:6}; some of these conditions are visualized.
}
\label{fig:proof}
\end{figure}

\begin{lemma}\label{lemma:proof-left-to-right}
 ${\mathsf H}^{\mathsf\frown}{\mathsf t}^{\mathsf\frown}{\mathsf s} \in L(\mathcal A) \Rightarrow (t,s,o) \in_k R_o(\mathcal T) \text{ and } o \text{ is a linear transduction}$
\end{lemma}

\begin{proof}
Each reachable (according to the strategies of {\In} and {\Out}) node of ${\mathsf H}^{\mathsf\frown}{\mathsf t}^{\mathsf\frown}{\mathsf s}$ is identified with a node of the input tree $t$.
We show by induction on the height of a level of the input tree $t$ that the following statement holds:

Fix an accepting run of $\mathcal A$ on ${\mathsf H}^{\mathsf\frown}{\mathsf t}^{\mathsf\frown}{\mathsf s}$.
For each level of the input tree $t$, there exists a sequence of configurations $\rho:~(t, q_0, \varphi_0)~\rightarrow_*~(t,s_\rho,\varphi)$ of $\mathcal T$ with associated origin function $o_\rho$ that extends the sequence of configurations from the previous level such that the following conditions are satisfied:

Refer to \cref{fig:proof} for a graphical representation of the relation between some of the mentioned components below.
For a reachable (according to the strategies of {\In} and {\Out}) node of ${\mathsf H}^{\mathsf\frown}{\mathsf t}^{\mathsf\frown}{\mathsf s}$ that corresponds to a node $u \in \dom{t}$ on that level, that was reached with a state that contains $(f,o,S_O,S_\mathcal T)$ in the run of $\mathcal A_k$, and $e \in \mathrm{EXT}_k(S_\mathcal T(f))$ was the extension chosen to compute the states at the children in the run of $\mathcal A_k$, there exists a unique node $v \in \dom{s}$ such that

\begin{enumerate}\label{enum:claims}
  \item \label{c:1} for all $x \in \dom{S_O \cdot {o^{(k)}}}$ with $\val{S_O \cdot {o^{(k)}}}(x) \in \SigmaO^{(k)}$ holds 
  \[
    \val{s}(vx) = strip(\val{S_O \cdot {o^{(k)}}}(x)),
  \]
\end{enumerate}

 This condition guarantees that the run of $\mathcal A$ is build according to the choices of {\Out}.

\begin{enumerate}\setcounter{enumi}{1}
  \item \label{c:2} for all $x \in \dom{e}$ with $\val{e}(x) \in \SigmaO^{(k)}$ holds 
  \[
    \val{s_\rho[v \leftarrow e]}(vx)~=~strip(\val{e}(x)), \text{ and}
  \]
  and if $x \in \dom{S_O \cdot {o^{(k)}}}$ and $\val{S_O \cdot {o^{(k)}}}(x) \in \SigmaO^{(k)}$, then 
  \[
    strip(\val{e}(x))~=~\val{s}(vx),
  \]
\end{enumerate}

This condition guarantees that the output tree $s_\rho$ (with the additional computation step(s) from the extension $e$) is build up from the computations of $\mathcal T$ guessed in the run of $\mathcal A$, and these correspond to the output choices of {\Out}.
Eventually, this condition guarantees that the computation of {\Out} and the guessed computation of $\mathcal T$ yield the same output tree.

\begin{enumerate}\setcounter{enumi}{2}
  \item \label{c:3} for the unique $x \in \dom{S_\mathcal T}$ with $\val{S_\mathcal T}(x) \in Q$ holds 
  \[
    \val{s_\rho}(vx) = strip(\val{S_\mathcal T}(x)), \varphi(vx) = u \text{ and } \val{t}(u) = f,
  \]
\end{enumerate}

This condition guarantees that the extension $e$ of the current output information tree $S_\mathcal T$ in the run of $\mathcal A$ was chosen wrt.\ the correct current input symbol.

\begin{enumerate}\setcounter{enumi}{3}
  \item \label{c:4} for all $xy \in \dom{S_\mathcal T}$ with $\val{S_\mathcal T}(x) \in Q^{(k)} \setminus Q$ holds 
  \[
    \val{t}(\varphi(vx)y) = \val{S_\mathcal T}(xy),
  \]
\end{enumerate} 

Recall that if the computation of $\mathcal T$ on $t$ (which is guessed in the run of $\mathcal A$) and the computation of ${\Out}$ on $t$ (which is obtained from ${\mathsf H}^{\mathsf\frown}{\mathsf t}^{\mathsf\frown}{\mathsf s}$) diverge, input neighborhoods for the divergent directions are chosen to continue the computation of $\mathcal T$, cf.~\cref{def:next}.
This condition guarantees that these input neighborhood choices are correct.
 Recall the form of output information trees wrt.\ $\mathcal T$, cf.~\cref{def:tdtt-info-tree}.
The input neighborhoods in such a tree are found below the nodes with labels from $Q^{(k)} \setminus Q$.
 
As a consequence of \cref{c:3,c:4}, we can define the run $\rho+e$ to be the run that is obtained from $\rho$ by additionally making the computation step(s) from the extension $e$.
Let $o_{\rho+e}$ be the associated origin function of this run.

\begin{enumerate}\setcounter{enumi}{4}
  \item \label{c:5} for all $x \in \dom{S_\mathcal T}$ with $\val{S_\mathcal T}(x) \in Q^{(k)} \setminus Q$ holds 
  \[
    \mathit{dist}(u,\varphi(vx)) = \mathit{ann}(\val{S_\mathcal T}(x)), \text{ and}
  \]
 for all $x \in \dom{e}$ with $\val{e}(x) \in \SigmaO^{(k)}$ holds
 \[
    \mathit{dist}(u,o_{\rho+e}(vx)) \text{ is at most } k - \mathit{ann}(\val{e}(x)),
 \]
\end{enumerate} 

This condition relates annotations of nodes from $S_\mathcal T$ resp.\ $e$ to distances.
The run of $\mathcal A$ follows the paths through $t$ that {\Out} reads.
The first part of the condition guarantees that the distance between the current input node $u$ that {\Out} reads and some other divergent input node $\varphi(vx)$ that $\mathcal T$ reads is given by the respective annotation.
The second part of the condition guarantees that the distance of current input node $u$ that {\Out} reads to the origin $\varphi(vx)$ of some already by $\mathcal T$ produced output node $vx$ is bounded by the respective annotation subtracted from $k$.
This condition is an auxiliary condition that is only needed to prove the next condition.

\begin{enumerate}\setcounter{enumi}{5}
  \item \label{c:6} for all $x \in \dom{S_O \cdot {o^{(k)}}} \cap \dom{e}$ with $\val{S_O \cdot {o^{(k)}}}(x) \in \SigmaO^{(k)}$ and $\val{e}(x) \in \SigmaO^{(k)}$ holds 
  \[
    \mathit{dist}(u,o_{\rho+e}(vx)) \text{ is at most } \mathit{ann}(\val{S_O \cdot {o^{(k)}}}(x)), \text{ and}
  \]
  \[
    \mathit{dist}(u,o(vx)) \text{ is at most } k - \mathit{ann}(\val{S_O \cdot {o^{(k)}}}(x)), 
  \]
  consequently, $\mathit{dist}(o(vx),o_{\rho+e}(vx)) \leq k$.
\end{enumerate}

This condition guarantees for output, that was in this computation step finally produced by both parties (one party could have produced the output before), that the distance between their origins is at most $k$.
Eventually, it guarantees that the desired distance bound of $k$ between the origins from the computation of {\Out} and the guessed computation of $\mathcal T$ is respected.

Altogether, it is then easy to see that $(t,s,o) \in_k R_o(\mathcal T)$, because we inductively build a run $\rho$ of $\mathcal T$ on $t$ with final transformed output $s$ and origin mapping $o_\rho$ such that $dist(o(x),o_{\rho}(x)) \leq k$ for all $x \in \dom{s}$.

\paragraph*{Induction base}

We begin with the root level.
We chose $\rho$ to be the initial configuration $(t,q_0,\varphi_0)$.
We have to consider $u = \varepsilon$ with current information $(f,o,\circ,q_0)$, where $f$ is the label of the root, $o$ is the output choice made by {\Out} in the root, and $e \in \mathrm{EXT}_{k}(S_{\mathcal T}(f))$ is the extension used in the run of $\mathcal A$.
We chose $v = \varepsilon$, it is easy to see that \cref{c:1,c:2,c:3,c:4,c:5,c:6} are satisfied

\paragraph*{Induction step}

Assume the claim holds for level $n$.
Let $\rho: (t, q_0, \varphi_0) \rightarrow_* (t,s_\rho,\varphi)$ be a sequence of configurations such that the claim holds.
We show that the claim also holds for level $n+1$ by showing for any reachable node that corresponds to some input node $u' \in \dom{t}$ with height $n+1$ that we can extend the sequence $\rho$ and pick some $v' \in \dom{s}$ such that \cref{c:1,c:2,c:3,c:4,c:5,c:6} are satisfied for $u'$.
The desired sequence of configurations for level $n+1$ is obtained by applying all extensions to $\rho$ defined for each reachable input node from level $n+1$.

To define the extension of $\rho$ for some $u' \in \dom{t}$, we consider the unique node that corresponds to an input node $u \in \dom{t}$ with height $n$ such that $u'$ is a child of $u$, say the $i$th child, thus $u' = ui$.
We assume that $(f,o,S_O,S_\mathcal T)$ is the state reached in the run of $\mathcal A$ for this node, and that $e \in \mathrm{EXT}_{k}(S_{\mathcal T}(f))$ is the extension that was used  to compute the states at the children in the run of $\mathcal A$.
Let $v \in \dom{s}$ be the unique node such that \cref{c:1,c:2,c:3,c:4,c:5,c:6} are satisfied.
For $u' \in \dom{t}$, assume that $(f',o',S_O',S_\mathcal T')$ is the state reached in the run of $\mathcal A$ for this node, and that $e' \in \mathrm{EXT}_{k}(S_{\mathcal T'}(f'))$ is the extension that was used to compute the states at the children in the run of $\mathcal A$.

Since for $u \in \dom{t}$ \cref{c:1,c:2,c:3,c:4,c:5,c:6} are satisfied, especially since \cref{c:3,c:4} are satisfied, we define $\rho'$ as a sequence $(t, q_0, \varphi_0) \rightarrow_* (t,s_\rho,\varphi) \rightarrow_* (t,s_{\rho'},\varphi')$ such that $s_{\rho'},$ is obtained from $s_\rho$ by applying the rules corresponding to the extension $e$.
Recall how $S_O'$ and $S_\mathcal T'$ are computed for $u' = ui$.
Therefore, recall \cref{def:next} which defines the next relation and \cref{def:sync} which defines the sync function.
The output information trees $S_O'$ and $S_\mathcal T'$ are computed from the result of $sync(S_O \cdot o^{(k)},e)$.
We are interested in the node $z$ from $\dom{S_O \cdot o^{(k)}} \cap \dom{e}$ such that $\lambda$ maps $S_O \cdot o^{(k)}|_z$ to $e|_z$, and these are the upper parts of $S_O'$ and and $S_\mathcal T'$.
We then define $v' \in \dom{s}$ as the node $vz \in \dom{s}$.
We now show that \cref{c:1,c:2,c:3,c:4,c:5,c:6} are satisfied for $u'$ and the choice of $v'$.

\phantom{\quad}

\noindent {\sffamily \bfseries \cref{c:1}.}
We have to show that for all $x \in \dom{S_O' \cdot {o'^{(k)}}}$ with $\val{S_O' \cdot {o'^{(k)}}}(x) \in \SigmaO^{(k)}$ holds 
\[
  \val{s}(v'x) = strip(\val{S_O' \cdot {o'^{(k)}}}(x)).
\]

We only have to show the claim for all newly added nodes that were not already part of $\dom{S_O \cdot {o^{(k)}}}$, formally, these are the nodes $x \in \dom{S_O' \cdot {o'^{(k)}}}$ such that $v'x \notin \{v\} \cdot \dom{S_O \cdot {o^{(k)}}}$.
There are two cases how such a new node could have been introduced.
First, $x$ belongs to the $o'^{(k)}$ part of $S_O' \cdot {o'^{(k)}}$.
Then clearly the claim is true, because the run of $\mathcal A$ uses $o'$ by construction.
Secondly, $x$ belongs to the $S_O'$ part of $S_O' \cdot {o'^{(k)}}$ not present in $S_O \cdot {o^{(k)}}$.
That means the node $x$ was introduced to $S_O'$ as follows.
Recall the computation of the next states of $\mathcal A$ in the run.
Since $S_O'$ is not $\circ$, the result of $sync(S_O \cdot o,e)$ must yield $\lambda$ such that $\lambda: w[x_{j_1},\dots,x_{j_n}] \mapsto q(x_i)$ for some context $w$, some state $q$ of $\mathcal T$, and direction $i$.
In order for $\mathcal A$ to continue the computation from $u$, states at the child nodes of $u$ have to be chosen.
Recall that therefore $\mathcal A$ guesses output trees $s_j$ produced by \Out\ while continuing to read from $uj$ for all $j \neq i$.
These guesses are correct, because the run of $\mathcal A$ is successful, i.e., especially the part of the run of $\mathcal A$ starting at the $j$th child with $q_{out=dec(s_j^b)}$ for some $(s_j^a,s_j^b) \in \mathrm{DUP}_{(k)}(s_j)$ is successful for all $j \neq i$.
The same guesses are also part in state at $ui$, that is, in the $S_O'$ information at the corresponding positions.
Thus, for all these $x$ the claim is true.

\phantom{\quad}

\noindent {\sffamily \bfseries \cref{c:2}.}
We have to show that for all $x \in \dom{e'}$ with $\val{e'}(x) \in \SigmaO^{(k)}$ holds 
\[
  \val{s_\rho'[v' \leftarrow e']}(v'x)~=~strip(\val{e'}(x)),
\] 
and if $x \in \dom{S_O' \cdot {o'^{(k)}}}$ and $\val{S_O' \cdot {o'^{(k)}}}(x) \in \SigmaO^{(k)}$, then 
\[
  strip(\val{e'}(x))~=~\val{s}(v'x).
\]

Since the run of $\mathcal A$ as well as $\rho'$ (and thus $s_{\rho'}$ and also $v'$) were defined according to $e$, the claim is true.

\phantom{\quad}

\noindent {\sffamily \bfseries \cref{c:3}.}
We have to show that for the unique $x \in \dom{S_\mathcal T'}$ with $\val{S_\mathcal T'}(x) \in Q$ holds 
\[
  \val{s_{\rho'}}(v'x) = strip(\val{S_\mathcal T'}(x)), \varphi(v'x) = u' \text{ and } \val{t}(u') = f'.
\]

Since the run of $\mathcal A$ as well as $\rho'$ were defined according to $e$, the claim is true.

\phantom{\quad}

\noindent {\sffamily \bfseries \cref{c:4}.}
We have to show that for all $xy \in \dom{S_\mathcal T'}$ with $\val{S_\mathcal T'}(x) \in Q^{(k)} \setminus Q$ holds 
\[
  \val{t}(\varphi(v'x)y) = \val{S_\mathcal T'}(xy).
\]

Recall the description of the condition, it states that this is to make sure that the neighborhood guesses included in $S_\mathcal T'$ (for $u'=ui$) are correct.
We only have to show this for neighborhoods that were introduced in the step from $u$ to $ui$, because for neighborhoods from previous steps this condition was already shown to be true.
For the newly introduced neighborhoods the condition is true, because in the node corresponding to $u$, the automaton $\mathcal A$ guesses neighborhoods $t_1,\dots,t_n$ to compute $S_\mathcal T'$ passed to the node corresponding to $ui$.
Subsequently, $\mathcal A$ validates the guesses, because the run of $\mathcal A$ is successful from the $j$th child (corresponding to $uj$) with state $q_{in=t_j}$ for all $j\neq i$.

\phantom{\quad}

As a consequence of \cref{c:3,c:4}, we can define the run $\rho'+e'$ to be the run that is obtained from $\rho'$ by additionally making the computation step(s) from the extension $e'$.
Let $o_{\rho'+e'}$ be the associated origin function of this run.

\phantom{\quad}

\noindent {\sffamily \bfseries \cref{c:5}.}
We have to show that for all $x \in \dom{S_\mathcal T'}$ with $\val{S_\mathcal T'}(x) \in Q^{(k)} \setminus Q$ holds 
\[
  \mathit{dist}(u',\varphi'(v'x)) = \mathit{ann}(\val{S_\mathcal T'}(x)),
\]
and for all $x \in \dom{e'}$ with $\val{e'}(x) \in \SigmaO^{(k)}$ holds
\[
  \mathit{dist}(u',o_{\rho'+e'}(v'x)) \text{ is at most } k - \mathit{ann}(\val{e'}(x)).
\]

Recall the description of this condition.
The first part of the claim is easily seen to be true, we only give a brief explanation.
Pick one $x \in \dom{S_\mathcal T'}$ such that $\val{S_\mathcal T'}(x) \in Q^{(k)} \setminus Q$.
If the annotation is $> 2$, then there is some node $y \in \dom{S_\mathcal T}$ with $\val{S_\mathcal T}(y) \in Q^{(k)}\setminus Q$ and $ann(\val{S_\mathcal T'}(x)) = ann(\val{S_\mathcal T}(y)) + 2$ such that $\varphi'(v'x)$ is the child of $\varphi(vy)$.
Such a node exists because $\varphi'$ is obtained from $\varphi$ according to $e$ which extends all computations present in $S_\mathcal T$ by one step.
From the induction hypothesis we know that $dist(u,\varphi(vy)) = ann(\val{S_\mathcal T}(y))$, and since $dist(u,u') = 1$ and $dist(\varphi(vy),\varphi'(v'x)) = 1$, we can conclude that $dist(u',\varphi'(v'x)) = ann(\val{S_\mathcal T}(y)) + 2 = ann(\val{S_\mathcal T'}(x))$.
The correctness of the statement for the case that the annotation is $2$ can be shown with similar arguments.
Note that by construction the annotation is never $0$ nor $1$, because when the computations read the same node and then follow divergent directions to read in the tree, both advance one step in divergent directions making their distance $2$.

We prove the second part of the claim.
For $x \in \dom{e'}$ with $\val{e'}(x) \in \SigmaO^{(k)}$ we distinguish whether $x$ comes from the $S_\mathcal T'$ part or was introduced by extending $S_\mathcal T'$.

In the former case, the node was already represented in $e$, meaning there exists some $y \in \dom{e}$ such that $vy = v'x$.
From the induction hypothesis we know that $dist(u,o_{\rho+e}(vy))$ is at most $k - ann(\val{e}(y))$.
From $e$, the information output tree $S_{\mathcal T}'$ is computed, thereby reducing the annotation by one.
Thus, $dist(u',o_{\rho'+e'}(v'x))$ is at most $k - ann(\val{e'}(x))$, because $dist(u',o_{\rho'+e'}(v'x)) = dist(u',o_{\rho+e}(vy)) = dist(u,o_{\rho+e}(vy)) + 1$ and $ann(\val{e'}(x)) = ann(\val{e}(x)) - 1$.

In the latter case, the node was introduced in $e'$.
There are two possibilities how nodes with labels in $\Gamma^{(k)}$ can be introduced in $e'$.
Either it is introduced by extending the computation (stored in $S_\mathcal T$) at the unique node that is labeled by a symbol from $Q$, or by extending the computation at a node that is labeled by a symbol from $Q^{(k)}\setminus Q$.
In the first case, from the third condition (which we already proved), we know that this node represents the state that is reached by $\mathcal T$ at the node $u'$, and the output node $v'x$ was produced at $u'$.
By construction of $e'$ the annotation is set to $k$, and thus clearly $dist(u',o_{\rho'+e'}(v'x))$ is at most $k-ann(\val{e'}(x))$, because it is exactly zero.
In the second case, as we proved in the paragraph before, we know that such a node $z$ represents the state that is reached by $\mathcal T$ at the node $\varphi'(v'z)$ with $dist(u',\varphi'(v'z)) = ann(\val{S_\mathcal T'}(x))$.
Let $a$ denote the value $ann(\val{S_\mathcal T'}(x))$.
From $dist(u',\varphi'(v'z)) = a$ it follows directly that $dist(u',o_{\rho'+e'}(v'x))$ is $a$, because $o_{\rho'+e'}$ derived from $\varphi'$ and $e'$, i.e., $o_{\rho'+e'}(v'x) = \varphi'(v'z)$.
By construction of $e'$, the annotation at $x$ is set to $k-a$.
We have to show that $dist(u',o_{\rho'+e'}(v'x))$ is at most $k - ann(\val{e'}(x)) = k - (k-a) = a$.
This is true, because we have proved that the distance is exactly $a$.

\phantom{\quad}

\noindent {\sffamily \bfseries \cref{c:6}.}
We have to show that for all $x \in \dom{S_O' \cdot {o'^{(k)}}} \cap \dom{e'}$ with $\val{S_O' \cdot {o'^{(k)}}}(x) \in \SigmaO^{(k)}$ and $\val{e'}(x) \in \SigmaO^{(k)}$ holds 
\[
  \mathit{dist}(u',o_{\rho'+e'}(v'x)) \text{ is at most } \mathit{ann}(\val{S_O' \cdot {o'^{(k)}}}(x)), \text{ and}
\]
\[
  \mathit{dist}(u',o(v'x)) \text{ is at most } k - \mathit{ann}(\val{S_O' \cdot {o'^{(k)}}}(x)).
\]
Consequently, $dist(o(v'x),o_{\rho'+e'}(v'x)) \leq k$.

\phantom{\quad}

We prove the first part of the claim, that is, for such an $x \in \dom{S_O' \cdot {o'^{(k)}}} \cap \dom{e'}$ holds $dist(u',o_{\rho'+e'}(v'x)) \leq ann(\val{S_O' \cdot {o'^{(k)}}}(x))$.
From the previous condition we know that $dist(u',o_{\rho'+e'}(v'x))$ is at most $k - ann(\val{e'}(x))$.
Showing that $k - ann(\val{e'}(x)) \leq ann(\val{S_O' \cdot {o'^{(k)}}}(x))$ proves the claim.

Towards this, we recall how $S_O'$ and $S_\mathcal T'$ can look like, we recall the definition of the next relation, cf.~\cref{def:next}.
Since only one of the players can be ahead, we observe one of the following two situations.
$S_O' = \circ$ and $S_\mathcal T'$ is a tree from $N_{\SigmaO^{(k)} \cup \SigmaI \cup Q^{(k)} \cup Q}$ (indicating that $\mathcal T$ is ahead (or on par if $w$ is the context $x_1$ (which is the same as $\circ$)), or $S_O'$ is a (special) tree (which is not $\circ$) over $\Gamma^{(k)}$ and $S_\mathcal T \in Q$ (indicating that {\Out} is ahead).

If $\mathcal T$ is ahead, we have $S_O' \cdot o'^{(k)} = o'^{(k)}$ which means every annotation in $S_O' \cdot o'^{(k)}$ is $k$.
Then it is easy to see that $k - ann(\val{e'}(x)) \leq ann(\val{S_O' \cdot {o'^{(k)}}}(x))$, because $k - ann(\val{e'}(x)) \leq k$ is obviously true.

If {\Out} is ahead, going from $S_T'$ to $e'$ only yields $\Gamma^{(k)}$-labeled nodes whose annotation is $k$.
Clearly, $k - ann(\val{e'}(x)) = 0 \leq ann(\val{S_O' \cdot {o'^{(k)}}}(x))$.

We prove the second part of the claim, that is, for such an $x \in \dom{S_O' \cdot {o'^{(k)}}} \cap \dom{e'}$ holds $dist(u',o(v'x))$ is at most $k - ann(\val{S_O' \cdot {o'^{(k)}}}(x))$.
We distinguish whether $x$ comes from $S_O'$ or from $o'^{(k)}$.

In the former case, where $x$ comes from $S_O'$, we need another distinction.
First, we consider the case that node was already present in $S_O \cdot o^{(k)}$, i.e., there exists some $y \in \dom{S_O \cdot o^{(k)}}$ such that $vy = v'x$.
From the induction hypothesis, we know that $dist(u,o(vy))$ is at most $k - ann(\val{S_O \cdot {o^{(k)}}}(y))$.
Thus, it follows that $dist(u',o(v'x))$ is at most $k - ann(\val{S_O' \cdot {o'^{(k)}}}(x))$, because $dist(u,u') = 1$, and $ann(\val{S_O' \cdot {o'^{(k)}}}(x)) = ann(\val{S_O \cdot {o^{(k)}}}(y)) - 1$.

Secondly, we consider the case that the node was not present in $S_O \cdot o^{(k)}$, i.e., it has been introduced when computing $S_O'$ from $S_O \cdot o^{(k)}$.
We recall how this can happen, therefore, we recall the definition of the next relation, cf.~ \cref{def:next}.
We are in the situation that {\Out} is ahead of $\mathcal T$ and follows more than one direction whereas $\mathcal T$ follows only one.
For the directions that are not shared by both, by construction, output trees are chosen that have to be produced eventually by {\Out} and $\mathcal T$.
The node $x$ we consider is part of such an output tree, say the tree $s$.
By construction, two numbers $a$ and $b$ with $a + b \leq k$ are chosen such that $dist(u,o(v'x))$ is at most $b$ (otherwise the run of $\mathcal A$ would not exist because this condition is verified from a sibling of $u'$ beginning in state $q_{out=dec(s^b)}$).
That means that $dist(u',o(v'x))$ is at most $b+1$.
The annotation of $x$ in $S_O'$ is $a-1$, thus, from the fact that $dist(u',o(v'x))$ is at most $b+1$, it follows directly that $dist(u',o(v'x))$ is at most $k - ann(\val{S_O' \cdot {o'^{(k)}}}(x)) = k - (a - 1)$, because $b+1 \leq k - a + 1$.

In the latter case, where $x$ comes from $o'^{(k)}$, it follows that $dist(u',o(v'x))$ is at most $k - ann(\val{S_O' \cdot {o'^{(k)}}}(x))$, because the output node $v'x$ has origin $u'$, i.e., their distance is zero, and $x$ is introduced with annotation $k$.

\phantom{\quad}

We have shown that \cref{c:1,c:2,c:3,c:4,c:5,c:6} are satisfied for $u'$ and $v'$ which completes the proof of the lemma.
\end{proof}

Now we show the other direction.

\begin{lemma}\label{lemma:proof-right-to-left}
${\mathsf H}^{\mathsf\frown}{\mathsf t}^{\mathsf\frown}{\mathsf s} \in L(\mathcal A) \Leftarrow (t,s,o) \in_k R_o(\mathcal T) \text{ and } o \text{ is a linear transduction}$
\end{lemma}

\begin{proof}

We show that $(t,s,o) \in_k R_o(\mathcal T)$ obtained from the strategy annotations $s$ an $t$ implies that $\mathcal A$ accepts ${\mathsf H}^{\mathsf\frown}{\mathsf t}^{\mathsf\frown}{\mathsf s}$.

Since $(t,s,o_\rho) \in_k R_o(\mathcal T)$, there is a run $\rho$ of $\mathcal T$ on $t$ with result $s$ and origin function $o_\rho$ that for all $x \in \dom{s}$ holds $dist(o(x),o_\rho(x)) \leq k$.
To obtain an accepting run on ${\mathsf H}^{\mathsf\frown}{\mathsf t}^{\mathsf\frown}{\mathsf s}$ the parity tree automaton $\mathcal A$ can make its guesses according to the run $\rho$, then it accepts.
\end{proof}

This completes the proof of \cref{lemma:linear}.

\section{The full set of top-down tree transductions}
\label{app:fullcase}

Now we extend the constructions from the previous sections to handle the full set of top-down tree transductions, in other words, we lift the restriction that the top-down tree transductions are linear.

We fix some values.

\begin{assumption}\label{ass:non-lin}
  \begin{itemize}
    \item Let $\Sigma,\Gamma$ be ranked alphabets, and let $m$ be the maximal rank of $\Sigma$.
    \item Let $\mathcal T$ be a (possibly non-linear) {\tdtt} of the form $(Q,\Sigma,\Gamma,q_0,\Delta)$.
    \item Let $M$ be the maximal height of a tree appearing on the right-hand side of a transition rule in $\Delta$.
  \end{itemize}
\end{assumption}

Recall the notations and definitions given in \cref{app:notations}.
None of them are restricted to linear computations.
We revisit the next relation, see \cref{def:next} in \cref{app:next}, and take a look at some examples where non-linear extensions of {\Out} and $\mathcal T$ occur.

First, consider the following example.
Let $S \in \mathcal S$ be a state for which we want to compute successor states $P_1,\dots, P_i$ such that $S \rightarrow_{next} (P_1,\dots,P_i)$.
Let $s$ be the extension of the output information tree wrt.\ {\Out} stored in $S$ and $s'$ be the considered extension of the output information tree wrt.\ $\mathcal T$ stored in $S$.
These trees and the result $\lambda$ of $sync(s,s')$ are depicted below.
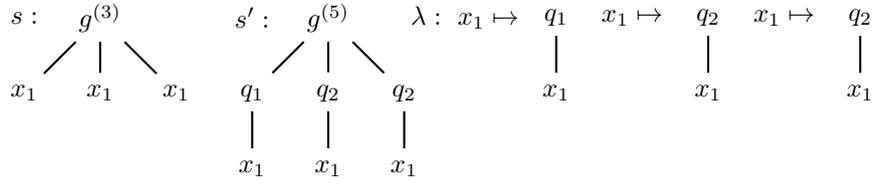
\begin{center}
  \begin{tikzpicture}[thick,baseline=(current bounding box.north)]
  \tikzstyle{level 1}=[sibling distance=10mm]
   \tikzstyle{level 2}=[sibling distance=8mm]
   \tikzstyle{level 3}=[sibling distance=8mm]

 \node[xshift=-1cm] at (0,0) {$s:$};
 
  \path[level distance=10mm,]
  node (root1) at (0,0) {$g^{(3)}$}
   child{
       node(00){$x_1$}
      }
   child{
       node(1){$x_1$}
   }
     child{
       node(1){$x_1$}
   }
  ;
  
  \begin{scope}[xshift=3cm]
   \tikzstyle{level 1}=[sibling distance=10mm]
   \tikzstyle{level 2}=[sibling distance=8mm]

 \node[xshift=-1cm] at (0,0) {$s':$};
 
  \path[level distance=10mm,] node(root2) at (0,0) {$g^{(5)}$}
    child{
          node(00){$q_1$}
           child{
             node(000){$x_1$}
           }
    }
    child{
        node(10){$q_2$}
         child{
          node{$x_1$}
         }
    }
    child{
        node(10){$q_2$}
         child{
          node{$x_1$}
         }
    }
  ;  
  \end{scope}
  
    \begin{scope}[xshift=6cm]
   \tikzstyle{level 1}=[sibling distance=10mm]
   \tikzstyle{level 2}=[sibling distance=8mm]

 
 \node[xshift=-1.2cm] at (0,0) {$\lambda:\ x_1 \mapsto$};
  \path[level distance=10mm,] 
          node(00a) at (0,0){$q_1$}
           child{
             node(000){$x_1$}
           }
  ;
   \end{scope}
   
   \begin{scope}[xshift=8cm]
   \node[xshift=-1cm] at (0,0) {$x_1 \mapsto$};
  \path[level distance=10mm,] 
          node(00a) at (0,0){$q_2$}
           child{
             node(000){$x_1$}
           }
  ;
  \end{scope}
  
   \begin{scope}[xshift=10cm]
   \node[xshift=-1cm] at (0,0) {$x_1 \mapsto$};
  \path[level distance=10mm,] 
          node(00a) at (0,0){$q_2$}
           child{
             node(000){$x_1$}
           }
  ;
  \end{scope}
  \end{tikzpicture}
  \end{center}
This describes a situation where both {\Out} and $\mathcal T$ continue to read in dir.~1 in order to produce output for each of the three children of the current output symbol.
Thus, if we consider the definition of next relation, see \cref{def:next}, we have three states $p_1 = (/,/,\circ,q_1)$, $p_1' = (/,/,\circ,q_2)$, and $p_1'' = (/,/,\circ,q_2)$ that we want to pass to dir. 1 (of which two are the same).
Thus, coming from $\{S\}$ we send $P_1 = \{p_1,p_1',p_1''\} = \{(/,/,\circ,q_1),(/,/,\circ,q_2)\}$ to dir.~1.

It is important to note that before, as {\Out} was linear, it was not possible that more than one output branch was depended on the same input branch.
This was also reflected by the form of the states that were reachable in $\mathcal A_k$ before;
if a state was reached such that it contained two or more states of the form $(\ast,\ast,S_O,S_\mathcal T)$, then at most one the output information trees wrt.\ {\Out} were special trees, i.e., from $S_{\Gamma^{(k)}\setminus \Gamma}$, and all the other output information trees wrt.\ {\Out} were output trees, i.e., from $T_{\Gamma^{(k)}\setminus \Gamma}$.
We recall  \cref{ex:special-case} as a reminder how it can happen that an output information tree wrt.\ {\Out} is from $T_{\Gamma^{(k)}\setminus \Gamma}$.
Thus, before, it was sufficient that {\Out} could only make one output choice per node.

We return to the current example.
Here, three of the output branches are depended on the same input branch.
It is easy to see that, since for the second and third output branch the same output information trees are computed, it suffices to follow the computations for the first and second output branch.
Thus, in a run of $\mathcal A$ on $\mathsf H$ with strategy annotations for {\In} and {\Out}, if a node that corresponds to some node $v$ of the associated input tree was reached with some state $\{S\}$, and $s$ and $s'$ are computed from $S$ as above, then a node that corresponds to $v1$ is reachable with the state $\{(/,/,\circ,q_1),(/,/,\circ,q_2)\}$ and the strategy annotation of {\Out} must specify two output choices in order to meaningful continue the followed computations.

We take a look at one more example.
Starting from some $S \in \mathcal S$, let $s$ be the extension of the output information tree wrt.\ {\Out} stored in $S$ and $s'$ be the considered extension of the output information tree wrt.\ $\mathcal T$ stored in $S$.
These trees and the result $\lambda$ of $sync(s,s')$ are depicted below.
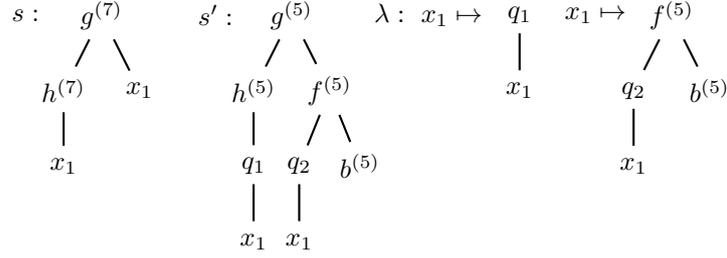
\begin{center}
  \begin{tikzpicture}[thick,baseline=(current bounding box.north)]
  \tikzstyle{level 1}=[sibling distance=10mm]
   \tikzstyle{level 2}=[sibling distance=8mm]
   \tikzstyle{level 3}=[sibling distance=8mm]

 \node[xshift=-1cm] at (0,0) {$s:$};
 
  \path[level distance=10mm,]
  node (root1) at (0,0) {$g^{(7)}$}
   child{
     node(0){$h^{(7)}$}
      child{
       node(00){$x_1$}
      }
   }
   child{
       node(1){$x_1$}
   }
  ;
  
  \begin{scope}[xshift=2.5cm]
   \tikzstyle{level 1}=[sibling distance=10mm]
   \tikzstyle{level 2}=[sibling distance=8mm]

 \node[xshift=-1cm] at (0,0) {$s':$};
 
  \path[level distance=10mm,] node(root2) at (0,0) {$g^{(5)}$}
    child{
      node(0){$h^{(5)}$}
        child{
          node(00){$q_1$}
           child{
             node(000){$x_1$}
           }
        }
    }
    child{
      node(1){$f^{(5)}$}
       child{
        node(10){$q_2$}
         child{
          node{$x_1$}
         }
       }
       child{
         node{$b^{(5)}$}
       }
    }
  ;  
  \end{scope}
  
    \begin{scope}[xshift=5.5cm]
   \tikzstyle{level 1}=[sibling distance=10mm]
   \tikzstyle{level 2}=[sibling distance=8mm]

 
 \node[xshift=-1.2cm] at (0,0) {$\lambda:\ x_1 \mapsto$};
  \path[level distance=10mm,] 
  
          node(00a) at (0,0){$q_1$}
           child{
             node(000){$x_1$}
           }
 
  ;
   \end{scope}
   
   \begin{scope}[xshift=7.5cm]
   \node[xshift=-1cm] at (0,0) {$x_1 \mapsto$};
  \path[level distance=10mm,]
      node(1b) at (0,0) {$f^{(5)}$}
       child{
        node(10){$q_2$}
         child{
          node{$x_1$}
         }
       }
       child{
         node{$b^{(5)}$}
       }
  ;

  \end{scope}

  \end{tikzpicture}
  \end{center}
In this case, from $\{S\}$, we go to the state $\left\{ (/,/,\circ,q_1), (/,/,\circ,
\begin{tikzpicture}[thick,baseline=(current bounding box.center)]
  \tikzstyle{level 1}=[sibling distance=10mm]
   \tikzstyle{level 2}=[sibling distance=8mm]
   \tikzstyle{level 3}=[sibling distance=8mm]
   \path[level distance=10mm,]
      node(1b) {$f^{(4)}$}
       child{
        node(10){$q_2$}
       }
       child{
         node{$b^{(4)}$}
       }
  ;  
   \end{tikzpicture}   
   )\right\}$ in dir.~1, and again {\Out} has to make two output choices at the next node.

We now formalize how to keep track of the different output choices that have to be made in the infinite tree ${\mathsf H_{\mathcal T,k}}$.
We have to adapt ${\mathsf H_{\mathcal T,k}}$ such that {\Out} can make an output choice for each output branch that depends on the current input branch, up to output branches where the same output information trees wrt.\ {\Out} and $\mathcal T$ are reached, meaning the computations of {\Out} and $\mathcal T$ can be continued in the same way for these output branches.
We recall that $\mathcal S$ are the states that store output information trees wrt.\ {\Out} and $\mathcal T$, thus the number of different computations of {\Out} and $\mathcal T$ for output branches (that can stay within the given origin bound $k$) is capped by $|\mathcal S|$.
Obviously, ${\mathsf H_{\mathcal T,k}}$ also has to represent non-linear output choices in this scenario.
It suffices to make the following changes.

\begin{definition}[Changes to ${\mathsf G_{\mathcal T,k}}$]\label{def:gzwei}
Given $k\geq 0$, let ${\mathsf G_{\mathcal T,k}}$ as defined in \cref{def:graphG} with the change that in each vertex of {\Out} she can make up to $|\mathcal S|$ (linear or non-linear) output choices (with the same hight as before) simultaneously.
  The output choices are represented as tuples of length at most $|\mathcal S|$.
  Recall that the set $\mathcal S$ is constructed in \cref{app:construction}.
 \end{definition}

\begin{definition}[Changes to ${\mathsf H_{\mathcal T,k}}$]
Given $k \geq 0$, let ${\mathsf H_{\mathcal T,k}}$ be the unraveling from ${\mathsf G_{\mathcal T,k}}$ as defined in \cref{def:gzwei} with root node $\varepsilon$.
\end{definition}

We also have to adapt the \pta $\mathcal A$.
Now, it must additionally verify whether a move of {\Out} is valid in the sense that for each of the outputs that depend on the same input a choice is made.
We consider a move valid if {\Out} makes as many output choices as the number of output information trees wrt.\ {\Out} (which are special trees) that are stored in the current state $S \subseteq \mathcal S$ (the set of output information trees) of $\mathcal A$.
As explained above, this captures the desired property.
As before, $\mathcal A$ has to collect the output choices in its run to build the extensions of output information trees wrt.\ {\Out}.
Since now possibly more than one output choice is made at a node, we define which output choice belongs to which followed pair of computations of {\Out} and $\mathcal T$.
Now, the \pta $\mathcal A$ stores sets of output information trees.
We fix an ordering of these trees meaning we have ordered sets.
Then the $i$th output choice of {\Out} is used to build the extension of the $i$th output information tree wrt.\ {\Out} that is a special tree, i.e., a tree that can indeed be meaningful extended, that is stored in the current state of $\mathcal A$.

Finally, we are ready to prove our main technical lemma, \cref{lemma:regular}, restated below.
We do not provide a full formal proof of the correctness, but give a proof sketch.

\lemmaregular*

\begin{proof}[Proof sketch.]
 Let $\mathcal A$ be the \pta described in the previous paragraph.

 \phantom{\quad}

 \noindent``${\mathsf H_{\mathcal T,k}}^{\mathsf\frown}{\mathsf t}^{\mathsf\frown}{\mathsf s} \in L(\mathcal A) \Rightarrow (t,s,o) \in_k R_o(\mathcal T)$''

 \phantom{\quad}
 
Let $t$, $s$, and $o$ be the input tree, the output tree, and the origin function obtained from ${\mathsf H_{\mathcal T,k}}^{\mathsf\frown}{\mathsf t}^{\mathsf\frown}{\mathsf s}$.
The proof is a simple adaptation from the proof of \cref{lemma:proof-left-to-right}.
Each reachable (according to the strategies of {\In} and {\Out}) node of ${\mathsf H_{\mathcal T,k}}^{\mathsf\frown}{\mathsf t}^{\mathsf\frown}{\mathsf s}$ is identified with a node of the input tree $t$.
We show by induction on the height of a level of the input tree $t$ that the following statement holds:

Fix an accepting run of $\mathcal A$ on ${\mathsf H_{\mathcal T,k}}^{\mathsf\frown}{\mathsf t}^{\mathsf\frown}{\mathsf s}$.
For each level of the input tree $t$, there exists a sequence of configurations $\rho:~(t, q_0, \varphi_0)~\rightarrow_*~(t,s_\rho,\varphi)$ of $\mathcal T'$ with associated origin function $o_\rho$ that extends the sequence of configurations from the previous level such that the following conditions are satisfied:

For a reachable (according to the strategies of {\In} and {\Out}) node of ${\mathsf H_{\mathcal T,k}}^{\mathsf\frown}{\mathsf t}^{\mathsf\frown}{\mathsf s}$ that corresponds to a node $u \in \dom{t}$ on that level, that was reached with a state that contains $(f,o,S_O,S_\mathcal T)$ in the run of $\mathcal A$, and $e \in \mathrm{EXT}_k(S_\mathcal T(f))$ was the extension chosen to compute the states at the children in the run of $\mathcal A$, there exists a unique maximal set $V \subseteq \dom{s}$ such that for each $v \in V$ we have that \cref{c:1,c:2,c:3,c:4,c:5,c:6} as in the proof of \cref{lemma:proof-left-to-right} on \cpageref{enum:claims} are satisfied.

The only difference to the previous proof is that each such input node is now associated with a set of output nodes, instead of a single output node.
 
From the induction it follows that $(t,s,o) \in_k R_o(\mathcal T)$.
To see this, one has to realize that we indeed build a run of $\mathcal T$ on $t$ that fully specifies the output tree $s$.
 
\phantom{\quad}
 
 \noindent``${\mathsf H_{\mathcal T,k}}^{\mathsf\frown}{\mathsf t}^{\mathsf\frown}{\mathsf s} \in L(\mathcal A) \Leftarrow (t,s,o) \in_k R_o(\mathcal T)$''

 \phantom{\quad}
 
 Let $t$, $s$, and $o$ be the input tree, the output tree, and the origin function obtained from ${\mathsf H_{\mathcal T,k}}^{\mathsf\frown}{\mathsf t}^{\mathsf\frown}{\mathsf s}$.
 By assumption we have that $(t,s,o) \in_k R_o(\mathcal T)$.
 We show that $\mathcal A$ accepts ${\mathsf H_{\mathcal T,k}}^{\mathsf\frown}{\mathsf t}^{\mathsf\frown}{\mathsf s}$.
 We let 
\begin{align*}
  \{\ \rho \mid & \ \rho \text{ is a run of } \mathcal T \text{ on } t \text{ with final transformed output } s \text{ and origin mapping } o_\rho\\
   & \ \text{such that } \mathit{dist}(o(x),o_\rho(x)) \leq k \text{ for all } x \in \dom{s} \}
\end{align*}
 be the set of all runs whose resulting origin mappings have at most an origin distance of $k$ to $o$.
 We select one run $\rho$ from this set such that the parity tree automaton $\mathcal A$ can make its guesses according to $\rho$ and accepts ${\mathsf H_{\mathcal T,k}}^{\mathsf\frown}{\mathsf t}^{\mathsf\frown}{\mathsf s}$.
 Therefore, $\rho$ has to fulfill the following property:
 
 Recall that {\Out} can make up to $|\mathcal S|$ many output choices simultaneously.
 Thus, it can happen that different branches of the output tree are build up from the same output choices.
 Each point from which this happens is call this a merge.
 We have to select $\rho$ such that these merges are reflected in $\rho$.
 Therefore, we compare the computation of {\Out} and $\rho$ using pairs of output state information trees $(S_O,S_\mathcal T)$ as used in the definition of $\rightarrow_{next}$.
 Assume {\Out} produces output that is mapped to some node $v$ and then the next outputs are mapped to, wlog., $v'$ and $v''$ with $v \sqsubset v'$ and $v \sqsubset v''$.
 If a merge has occurred for {\Out}, it must be the case that the same pair of output state information trees describes the situation at $v'$ and $v''$ and the part of $\rho$ that produces the output at $v'$ and below has to be the same as the part of $\rho$ that produces output at $v''$ and below.
 If this criterion is met for every merge, then $\rho$ describes a run that $\mathcal A$ can use to build an accepting run on ${\mathsf H_{\mathcal T,k}}^{\mathsf\frown}{\mathsf t}^{\mathsf\frown}{\mathsf s}$.
\end{proof}

\end{document}